\documentclass{statsoc}
\usepackage[a4paper]{geometry}
\usepackage[textwidth=8em,textsize=small]{todonotes}
\usepackage{amsmath}
\usepackage{amssymb}
\usepackage{bm}
\usepackage{enumitem}
\usepackage{graphicx,xcolor}
\usepackage{natbib}

\usepackage{xr}
\makeatletter
\newcommand*{\addFileDependency}[1]{
  \typeout{(#1)}
  \@addtofilelist{#1}
  \IfFileExists{#1}{}{\typeout{No file #1.}}
}
\makeatother
\newcommand*{\myexternaldocument}[1]{
    \externaldocument{#1}
    \addFileDependency{#1.tex}
    \addFileDependency{#1.aux}
}
\myexternaldocument{supplementary2}

\newcommand{\E}{\ensuremath{\mathrm{E}}}

\newcommand{\Ink}[1]{\overline{I}_{\bn_k}(\bk;#1)}
\newcommand{\In}[1]{\overline{I}_{\bn}(\bk;#1)}
\newcommand{\Jx}{\mathcal{J}_{\mathbf{n}}}

\newcommand{\N}{\mbox{I\hspace{-.15em}N}}
\newcommand{\R}{\mbox{I\hspace{-.15em}R}}
\newcommand{\T}{\mathcal{T}}

\newcommand{\Z}{\mbox{Z\hspace{-.48em}Z}}

\newcommand{\bX}{\mathbf{X}}

\newcommand{\bdelta}{\boldsymbol{\delta}}
\newcommand{\bgamma}{\boldsymbol{\gamma}}
\newcommand{\blambda}{\boldsymbol{\lambda}}
\newcommand{\bg}{\mathbf{g}}

\newcommand{\bk}{\boldsymbol{\omega}}
\newcommand{\bn}{\mathbf{n}}
\newcommand{\bomega}{\boldsymbol{\omega}}
\newcommand{\bs}{\mathbf{s}}
\newcommand{\btheta}{\boldsymbol{\theta}}

\newcommand{\bu}{\mathbf{u}}

\newcommand{\cov}{\ensuremath{\mathrm{cov}}}

\newcommand{\n}{|\bn|}
\newcommand{\revision}[1]{#1}

\newcommand{\var}{\ensuremath{\mathrm{var}}}
\newcommand{\pq}{{(q)}}
\newcommand{\pr}{{(r)}}
\newcommand{\pqr}{{(qr)}}

\newcommand{\cum}{\text{cum}}

\newtheorem{assumption}{Assumption}
\newtheorem{definition}{Definition}
\newtheorem{lemma}{Lemma}
\newtheorem{corollary}{Corollary}
\newtheorem{proposition}{Proposition}
\newtheorem{theorem}{Theorem}

\title{The Debiased Spatial Whittle Likelihood}
\author[Arthur P.~Guillaumin et al.]{Arthur P.~Guillaumin}
\address{Queen Mary University of London, United Kingdom}
\email{ag7531@nyu.edu}
\author{Adam M.~Sykulski}
\address{Lancaster University, United Kingdom}
\author{Sofia C.~Olhede}
\address{\'Ecole Polytechnique F\'ed\'erale de Lausanne, Switzerland\\ 
         University College London, United Kingdom}
\author[Arthur P.~Guillaumin et al.]{Frederik J.~Simons}
\address{Princeton University, USA}
\begin{document}

\begin{abstract}
We provide a computationally and statistically efficient method for
estimating the parameters of a stochastic covariance model observed on a
regular spatial grid in any number of dimensions. Our proposed method, which
we call the Debiased Spatial Whittle likelihood, makes important corrections
to the well-known Whittle likelihood to account for large sources of bias
caused by boundary effects and aliasing. We generalise the approach to
flexibly allow for significant volumes of missing data including those with
lower-dimensional substructure, and for irregular sampling boundaries. We
build a theoretical framework under relatively weak assumptions which
ensures consistency and asymptotic normality in numerous practical settings
including missing data and non-Gaussian processes. We also
extend our consistency results to multivariate processes. We provide
detailed implementation guidelines which ensure the estimation procedure can
be conducted in $\mathcal{O}(n\log n)$ operations, where $n$ is the
number of points of the encapsulating rectangular grid, thus keeping the
computational scalability of Fourier and Whittle-based methods for large
data sets. We validate our procedure over a range of simulated and real-world settings, and compare with state-of-the-art alternatives,
demonstrating the enduring practical appeal of Fourier-based methods,
provided they are corrected by the procedures developed in this paper.
\end{abstract}



\section{Introduction}
\label{sec:Intro}

Among the challenges of modern data analysis is making sense of large
volumes of spatial and spatiotemporal data. State-of-the-art parameter
estimation methods are based on various likelihood approximations
designed to combine statistical and computational efficiency. Such approximations
are primarily reliant on spatial/pixel models
\citep{anitescu2017,guinness2017circulant,katzfuss2017multi,stroud2017bayesian},
spectral/Fourier understanding
\citep{kaufman2008covariance,matsuda2009fourier,shaby2012tapered,guinness2017spectral},
or other methods of likelihood approximation
\citep{stein2004approximating,banerjee2008gaussian,lee2013locally,sang2012full}.
On the one hand, Fourier-based methods, typically based on the Whittle likelihood, are fast and
scale well to massive data sets. On the other hand, Fourier-based methods
are known to engender large sources of bias, particularly in dimensions
greater than one~\citep{dahlhaus1987edge}, in the presence of missing data,
or under irregular sampling
\citep{fuentes2007approximate,matsuda2009fourier}. In this paper we propose
a novel methodology that simultaneously addresses these challenges for
spatial data observed on a regular grid, with potentially missing data or
irregular sampling boundaries, and in any number of dimensions.

The bias which we remove is due to finite-domain effects, the
multidimensional boundary, and aliasing. Much of the
literature on Whittle estimation has focused on modifications to the
periodogram to reduce bias, such as tapering~\citep{dahlhaus1987edge},
edge-effect estimation~\citep{robinson2006modified}, or accounting for
non-standard sampling scenarios
\citep{fuentes2007approximate,matsuda2009fourier,rao2018statistical}. The
solution we propose is simple yet effective: determine the true expectation
of the periodogram under the proposed model and sampling regime, and construct a quasi-likelihood using this quantity
rather than the true spectrum---further developing and generalizing a
procedure recently proposed by \cite{sykulski2019debiased} for
one-dimensional completely observed time series. We shall show that the
Debiased Spatial Whittle likelihood almost completely removes estimation
bias in spatial inference, even in the presence of significant amounts of
missing data, while leaving estimation variance essentially unaffected. We
also establish a convergence rate under very general sampling and model
assumptions.

Debiasing Whittle estimates using the expected periodogram has been
notionally investigated in various more restrictive frameworks by
\cite{fernandez2010spatial}, \cite{simons2013maximum}, and
\cite{deb2017asymptotic}. This article, however, is the first to formalize
the estimation procedure by providing theoretical guarantees that apply in
any number of dimensions, allow for missing and/or
  non-Gaussian data, and account for aliasing and
irregular sampling boundaries. To achieve this we introduce the concept of
\emph{significant correlation contribution}, which provides weak conditions
on sampling regimes that allow for consistent parameter
estimation---leveraging ideas from modulated time series proposed by
\cite{guillaumin2017analysis}. Boundary effects play a significant role as
$d$, the dimensionality of the sampling domain, increases: the bias for a
$d$-dimensional cube with side $l$ scales like $1/l$ while the standard
deviation scales like $1/l^{d/2}$. Thus for $d>2$ the bias is of primary
significance, and it is important even for $d=2$. This paper is also the
first to provide fast $n \log n$ computational implementation, including for
missing data and higher dimensions.  We also prove consistency for
multivariate processes which may exhibit different missingness patterns
across components.

We establish the choice of notation and assumptions in
Section~\ref{sec:notations}. We propose our spatial quasi-likelihood in
Section~\ref{sec:Meth}. In Section~\ref{sec:scc} we introduce significant
correlation contribution, with conditions guaranteeing consistent estimation
under a wide range of sampling schemes. Section~\ref{sec:theory} develops
our theoretical results which include consistency, convergence rates, and asymptotic normality of parameter estimates in a wide range of settings. Section~\ref{sec:sim} shows the improved
performance on simulated data, and on actual data of Venus' topography. We
conclude with discussion in Section~\ref{sec:discussion}.

\section{Notation and assumptions}
\label{sec:notations}

Consider a finite-variance and zero-mean random field $X(\bs)$, for
$\bs\in\R^d$, where $d\geq 1$ is a positive integer. Under the assumption of
homogeneity, we denote the covariance function of $X(\bs)$ by $c_X(\bu)$,
$\bu\in\R^d$, and assume the existence of a positive piecewise-continuous
Riemann-integrable spectral density function $f_X(\bk)$, such that
$\forall\bu,\bs\in\R^d,$
\begin{equation}\label{eq:defSpectralDensity}
c_X(\bu)
=\E\left\{X(\bs)X(\bs+\bu)\right\} 
=\int_{\R^d}{f_X(\bk)\exp(i\bk\cdot\bu)\,d\bk}
,
\end{equation}
and
$f_X(\bk)=(2\pi)^{-d}\int_{\R^d}{c_X(\bu)\exp(-i\bk\cdot\bu)\,d\bu}$. We
shall assume the spectral density belongs to a parametric family indexed by
the parameter $\bgamma\in \Theta$, with $f_X(\bk) = f(\bk;\btheta)$, denoting
the true parameter value by $\btheta\in\Theta\subset\R^{p_\theta}, 
p_\theta\geq 1$.
Equivalently, we write $c_X(\bu;\bgamma), \gamma\in\Theta$
for the parametric family of covariances. The random field $X(\bs)$ is
taken to be homogeneous but not necessarily isometric. We denote
$\mathbf{n}=(n_1, \ldots, n_d)\in(\N^+)^d$, with $\N^+$ the set of positive
integers, the dimensions of an orthogonal regular and rectangular
\emph{bounding grid}, defined by
\begin{equation}\label{eq:completelattice}
\Jx = \left\{\bdelta \circ [x_1, \ldots, x_d]^T: 
(x_1,\ldots,x_d)\in\N^d, \ 0\leq x_i \leq n_i-1, i=1,\ldots, d\right\},
\end{equation}
and denote by $\n=\prod_{i=1}^d{n_i}$ the total number of points of this
grid. We denote by $X_{\bs}, \ \bs\in\mathcal{J}_\bn$ the values of the
process on the grid. In~\eqref{eq:completelattice}, the quantity
$\bdelta\in(\R^+)^d$ indicates the regular spacing along each axis, with
$\R^+$ the set of positive real numbers, and $\circ$ denotes the pointwise
Hadamard product between two vectors. We always take
$\bdelta=[1,\ldots,1]^T$ for simplicity, yet without loss of generality. We
write $f_{X,\bdelta}(\bk)$ for the spectral density of the sampled process,
the \emph{aliased} spectral density, defined by
\begin{equation}\label{aliased}
f_{X,\bdelta}(\bk) = \sum_{\bu\in\Z^d}{f_X\left(\bk + 2\pi \bu\right)},
\qquad\bk\in\R^d,
\end{equation}
which is a Fourier pair with
$c_X(\bu)=\int_{\T^d}{f_{X,\delta}(\bk)\exp(i\bk\cdot\bu)\,d\bk}$,
$\forall\bu\in\Z^d$, and $\T=[0,2\pi)$, with $\Z$ the set of
  integers.

To account for irregular domain shapes and missing data, we define a
deterministic modulation value $g_\bs$ at each location of the grid
$\mathcal{J}_\bn$. If a point on the regular grid is missing then $g_\bs =
0$, otherwise $g_\bs = 1$. By convention, $g_\bs$ is extended to the whole
set $\Z^d$, defining $g_\bs = 0$ if $\bs\notin\mathcal{J}_\bn$. Using this
notation, the periodogram of the observed data takes the form of
\begin{equation}\label{eq:definitionPeriodogram}
I_{\bn}(\bk) = \frac{(2\pi)^{-d}}{\sum_{\bs\in\mathcal{J}_\bn}{g_\bs}^2}
\left| \sum_{\bs\in\Jx}{g_\bs X_\bs
  \exp(-i\bk\cdot\bs}) \right|^2, 
\qquad\bk\in\R^d,
\end{equation}
where normalizing by $\sum_{\bs\in\mathcal{J}_\bn}{g_\bs}^2$ rescales the
periodogram for missing data, as performed by
\cite{fuentes2007approximate}. Note that, despite this similarity, our
approach is fundamentally different to that
of~\cite{fuentes2007approximate} who uses this extended definition of the
periodogram in the Whittle procedure to address missing data. While
this uniform rescaling is central to the method proposed
by~\cite{fuentes2007approximate}, it is merely a convention in our case. In
practice, this rescaling is not actually required in our implementation, as it will be cancelled out by the rescaling in the expected periodogram, as we shall shortly show.
Evaluating the periodogram on the multidimensional Fourier grid
\begin{equation*}
\prod_{j=1}^d{\left \{2\pi kn_j^{-1}:k=0,\ldots,n_j-1\right\}}
\end{equation*}
associated with the spatial grid $\Jx$ requires $\mathcal{O}(\n\log \n)$
elementary operations using the Fast Fourier Transform (FFT). If a taper is
used in the spectral estimate of~\eqref{eq:definitionPeriodogram}, then the
values of the taper are directly incorporated into $g_\bs$, such that
$g_\bs$ is proportional to the taper at locations where data are observed
(and still set to zero otherwise). We shall assume that $g_\bs$ takes values
in the interval $[0,1]$ as would be the case when using the periodogram, however this condition could be relaxed to assuming
an upper-bound for the absolute value.

\section{Methodology}
\label{sec:Meth}


We shall now introduce the Debiased Spatial Whittle likelihood
and an algorithm for its computation that only
requires FFTs, even in the scenario of missing data and general
boundaries. Thus our estimation method retains the
$\mathcal{O}(\n\log \n)$ computational cost of frequency-domain approaches
for regular grids.

\subsection{Estimation procedure}
\label{sec:parametricestimator}

Exact likelihood has optimal statistical properties in the framework
of an increasing domain~\citep{Mardia1984}, however it is computationally inadequate
for large data sets of spatial observations due to the determinant
calculation and linear system that needs to be solved. A common approach is
to trade off computational cost with statistical efficiency by using
approximations of the likelihood function
\citep{fuentes2007approximate,Varin2011,guinness2017circulant}. Such
functions are commonly called quasi-likelihood methods. Our proposed
estimation method uses the following quasi-likelihood, which we call
the Debiased Spatial Whittle Likelihood,
\begin{equation}\label{eq:pseudolkh}
\ell(\bgamma) = \n^{-1}\sum_{\bk\in\Omega_{\bn}}\left\{\log
\overline{I}_{\bn}(\bk;\bgamma) +
\frac{I_{\bn}(\bk)}{\overline{I}_{\bn}(\bk;\bgamma)}\right\}
\end{equation}
where, for all $\bgamma\in\Theta$, 
\begin{equation}
\overline{I}_{\bn}(\bk;\bgamma) = \E_{\bgamma}\{I_{\bn}(\bk)\}, \quad
\forall\bk\in\T^d,
\end{equation}
is the expected periodogram given the modulation values $g_\bs$, under the
mean-zero distribution of $X_\bs$ with covariance structure specified by the
parameter vector $\bgamma$---see also \citet{fernandez2010spatial}. In
Section~\ref{sec:multivariate} we describe the multivariate extension
to~\eqref{eq:pseudolkh} in~\eqref{multdebias}.  Note that
in~\eqref{eq:pseudolkh} the summation is over $\Omega_\bn\subset\T^d$.  It
is common to use the~\emph{natural} set of Fourier frequencies
$\Omega_{\bn}^{(1)}\equiv\prod_{j=1}^d\{2k\pi n_i^{-1}:k=0,\ldots, n_i-1\}$
for $\Omega_\bn$ in Whittle estimation, or a subset of these for
semi-parametric modelling.  To ensure identifiability in degenerate sampling
scenarios, when one or more of the dimensions of the domain are not growing
to infinity, we shall set $\Omega_\bn$ to be the set of Fourier frequencies
$\Omega_{\bn}^{(2)}\equiv\prod_{j=1}^d\{k\pi n_i^{-1}:k=0,\ldots, 2n_i-1\}$
in our theoretical developments. In practice, we shall use the natural set
of Fourier frequencies $\Omega_{\bn}^{(1)}$ in our simulations and real-data
example, as this is computationally faster and the practical difference was found to be insignificant.

Replacing $\overline{I}_{\bn}(\bk;\bgamma)$ with $f_X(\bk;\bgamma)$
in~\eqref{eq:pseudolkh} yields the discretised form of the standard Whittle
likelihood. Note however that, unlike the spectral density $f_X(\bk)$, the
expected periodogram $\overline{I}_{\bn}(\bk;\bgamma)$ directly accounts
for the sampling, as it depends on the dimensions of the lattice $\bn$ and
on the modulation values $g_\bs$ that account for missing data. We
minimize~\eqref{eq:pseudolkh} over $\Theta$ to obtain our estimate
\begin{equation}\label{eq:DebiasedWhittleLKH}
\widehat{\btheta} = \arg\min_{\bgamma\in\Theta}\{\ell(\bgamma)\}.
\end{equation}

By minimizing~\eqref{eq:pseudolkh}, we find the maximum-likelihood estimate
of the data under the following parametric model,
\begin{equation}\label{eq:approxmodel}
I_\bn(\bk) \stackrel{i.i.d.}{\sim}
\text{Exp}\left\{\overline{I}_\bn^{-1}(\bk;\btheta)\right\},
\qquad\bk\in\Omega_\bn,
\end{equation}
where $\text{Exp}(\lambda)$ stands for the exponential distribution with
parameter $\lambda$. Hence the quantity given in~\eqref{eq:pseudolkh} can be
seen as a composite likelihood~\citep{Varin2011,Bevilacqua2015comparing}. We
also observe that $\nabla_{\btheta}\ell(\btheta) = \mathbf{0}$ such that our
method fits within the general theory of estimating
equations~\citep{Heyde1997,Jesus2017}.

\subsection{Computation of the expected periodogram}
\label{sec:computeExpectedPeriodogram}

In this section we show how the expected periodogram in~\eqref{eq:pseudolkh}
can be computed using FFTs such that our quasi-likelihood remains an
$\mathcal{O}(\n\log \n)$ procedure, for any dimension $d$, and independently
of the missing data patterns. Direct calculations show that the expected
periodogram is the convolution of the spectral density of the process with
the multi-dimensional kernel $\mathcal{F}_{\bn}(\bk)$,
\begin{equation}\label{barI}
\overline{I}_{\bn}(\bk;\bgamma) = 
\left\{f_X(\ \cdot \ ;\bgamma) \ast
\mathcal{F}_{\bn}(\cdot)\right\}(\bk) \nonumber = 
	\int_{\T^d}{
  f_{X,\bdelta}(\bk-\bk';\bgamma)\mathcal{F}_{\bn}(\bk')\,d\bk'},
\end{equation}
where
\begin{equation}
\label{eq:generalFejerKernel}
\mathcal{F}_{\bn}(\bk) = \frac{(2\pi)^{-d}}{\sum{g_\bs^2}}
\left|\sum_{\bs\in\Jx}{g_\bs \exp(i\bk\cdot\bs)}\right|^2, 
\qquad\bk\in\R^d.
\end{equation}
When $g_\bs=1, \ \forall\bs\in\mathcal{J}_\bn$, $\mathcal{F}_{\bn}(\bk)$ is
simply the multi-dimensional rectangular F\'ejer kernel, i.e. a separable
product of one-dimensional F\'ejer kernels. For this reason we call
$\mathcal{F}_{\bn}(\bk)$ a \emph{modified} F\'ejer kernel. We now provide
two lemmata stating that the expected periodogram can be computed via FFTs
for any value of the modulation $g_\bs$ on the grid $\mathcal{J}_{\bn}$.
\begin{lemma}[Expected periodogram as a Fourier series]
\label{eq:computeExpectedPeriodogram}
The expected periodogram can be written as the following Fourier series,
\begin{equation}
\label{eq:perFourierSeries}
\In{\bgamma} = (2\pi)^{-d}\sum_{\bu\in\Z^d}{\overline{c}_\bn(\bu;\bgamma) \exp(-i\bk\cdot\bu)},
\quad\forall\bk\in\T^d, \forall\bgamma\in\Theta,
\end{equation}
where $\overline{c}_\bn(\bu;\bgamma)$ is defined by
\begin{equation}
  \overline{c}_\bn(\bu;\bgamma) = c_{g,\bn}(\bu)c_X(\bu; \bgamma), 
\qquad\bu\in\Z^d,\quad\text{with}
\end{equation}
\begin{equation}\label{eq:cgdef}
  c_{g,\bn}(\bu) = \frac{\sum_{\bs\in\mathcal{J}_\bn}{g_\bs g_{\bs+\bu}}}
	{\sum_{\bs\in\mathcal{J}_\bn}{g_\bs^2}}, 
\qquad\bu\in\Z^d.
\end{equation}
\end{lemma}
\begin{proof}
Direct calculation upon taking the expectation of the periodogram as defined
in~\eqref{eq:definitionPeriodogram}.
\end{proof}
Note that, having set $g_\bs$ to take value zero outside of the sampling
domain, we can rewrite \eqref{eq:cgdef} as
\begin{equation}
c_{g,\bn}(\bu) = \frac{\sum_{\bs\in\Z^d}{g_\bs g_{\bs+\bu}}}
{\sum_{\bs\in\Z^d}{g_\bs^2}}, 
\qquad\bu\in\Z^d.
\end{equation}
In practice we can evaluate the expected periodogram at the set of Fourier
frequencies through a multidimensional FFT, as detailed in the following
lemma.
\begin{lemma}[\revision{Computation of the expected periodogram via FFT}]
\label{lemmacompfft}
The expected periodogram can be expressed as
\begin{equation}
\overline{I}_{\bn}(\bk;\bgamma) = (2\pi)^{-d}
\sum_{u_1=0}^{n_1 - 1}
\ldots
\sum_{u_d=0}^{n_d - 1}
{
\widetilde{c}_{\bn}(\bu)\exp(-i\bk_k\cdot\bu)
}, \quad \forall\bk\in\Omega_\bn^{(1)},
\label{Ibar}
\end{equation}
where
\begin{equation}
	\label{eq:defctilde}
	\widetilde{c}_{\bn}(\bu) = \sum_{\mathbf{q}}{
		\overline{c}_\bn\left(\bu - \mathbf{q} \circ \bn;\bgamma\right) ,
	} \quad \bu\in\Z^d,
\end{equation}
and where the sum over $\mathbf{q}$ ranges over all vectors of size~$d$ with
elements in the set $\{0, 1\}$ (hence, $2^d$ of them), and where $\circ$
denotes the Hadamard product. \revision{Thus the expected periodogram can be
  computed via FFT.  Note that $\widetilde{c}_{\bn}$ is a \emph{periodized}
  version of $\overline{c}_{\bn}$
  as $\widetilde{c}_{\bn}(\bu -\mathbf{q} \circ \bn)=
 \widetilde{c}_{\bn}(\bu).$ }
\end{lemma}
\begin{proof}
Please see the Supplementary Material.
\end{proof}

As an example, in dimension $d=2$, $q$ takes values in $\left\{[0 \ 0]^T,
[1\ 0]^T, [0\ 1]^T, [1\ 1]^T\right\}$, and~\eqref{Ibar} therefore takes the
form
\begin{align*}
\overline{I}_{\bn}(\bk;\bgamma) = (2\pi)^{-d}\sum_{u_1=0}^{n_1-1}
\sum_{u_2=0}^{n_2-1}
&\left\{
\overline{c}_\bn\left(u_1,  u_2; \ \bgamma\right) +
\overline{c}_\bn\left(u_1-n_1,  u_2-n_2; \ \bgamma\right)\right.\\\nonumber 
&
+{}\left.\overline{c}_\bn\left(u_1, u_2-n_2; \ \bgamma\right)+
\overline{c}_\bn\left(u_1-n_1, \ u_2; \ \bgamma\right)
\right\}
\exp(-i\bk_k\cdot\bu).
\end{align*}
We remind the reader that $g_\bs$ is defined to be zero outside
$\mathcal{J}_\bn$. Hence, in the case of no tapering, $c_{g,\bn}(\bu)$ in~\eqref{eq:cgdef} is
the ratio of the number of pairs of observations \emph{separated} by the
vector $\bu$ over the total number of \emph{observed} points of the
rectangular grid $\mathcal{J}_\bn$. In the special case of complete
observations on the rectangular grid, \eqref{eq:cgdef} simplifies to
\begin{align}\label{eq:cgcomplete}
  c_{g,\bn}(\bu) = 
  \begin{cases} 
    \n^{-1}\prod_{i=1}^d \left(n_i-|u_i|\right) =
    \prod_{i=1}^d{\left(1-\frac{|u_i|}{n_i}\right)}& \ \ \text{if }
    |u_i|\leq n_i-1, i = 1,\ldots,d,\\ 
    0 & \ \ \text{otherwise},
  \end{cases}
\end{align}
which is a multidimensional form
of the triangle kernel found 
in \citet[p.198]{percival1993spectral} for the expected periodogram of regularly sampled time series.
In the general case, $c_{g,\bn}(\bu)$ is precomputed for all relevant values
of $\bu$ via an FFT independently of the parameter value $\bgamma$, such
that our method can be applied to scenarios of missing data without loss of
computational efficiency. Similarly, we can combine our debiasing procedure
with tapering by using a tapered spectral estimate for $I_{\bn}(\bk)$
in~\eqref{eq:pseudolkh} with adjusted values for $g_\bs$ (as discussed at
the end of Section~\ref{sec:notations}). The expected periodogram,
$\overline{I}_{\bn}(\bk;\bgamma)$, is then computed on $\Omega_{\bn}$ by
using these values of $g_\bs$ in the formulation of $c_{g,\bn}(\bu)$
in~\eqref{eq:cgdef}. Combining debiasing and tapering therefore remains an
$\mathcal{O}(\n\log \n)$ procedure.  The procedure of~\eqref{Ibar}
automatically incorporates sampling effects (geometry of the observation
region, missing observations), aliasing, and boundary effects, in one
$\mathcal{O}(\n\log \n)$ operation.  Note that merely calculating the
aliased spectral density, and using this in the Whittle likelihood, requires
knowledge of the full decay of the spectrum, and deciding on how many
aliased terms to include; a procedure that in general requires non-automatic
intervention and is not guaranteed to be $\mathcal{O}(\n\log \n)$.

\section{Properties of sampling patterns}
\label{sec:scc}

To account for missing observations on the rectangular grid $\Jx$, we
replace missing values with zeros via the modulation function
$g_\bs$. Depending on $g_\bs$ this may result in losing identifiability of
the parameter vector from the second-order moment quantities available from
the data. More generally, we wish to understand how the sampling pattern
affects the consistency of our estimation procedure. To this end, we define
the notion of significant correlation contribution for spatial random
fields, which determines whether the sampling pattern samples
a sufficient number of \emph{spatial lags} where information about the model lies. This
generalizes ideas from modulated time series
\citep{guillaumin2017analysis}. Following three simple lemmata on some
properties of $c_{g,\bn}(\bu)$, we shall provide the formal definition of
Significant Correlation Contribution (SCC), and follow with some general
cases and an example with an isometric model family to provide more
intuition and demonstrate the generality of our framework.

\subsection{Basic properties of $c_{g,\bn}(\bu)$ and $\mathcal{F}_\bn(\bk)$}

We state three basic properties of the introduced quantity $c_{g,\bn}(\bu)$
 in order to provide more intuition, but also for further use in this
paper.
\begin{lemma}We have
  \begin{equation}
    0\leq c_{g,\bn}(\bu)\leq 1, \quad \forall\bu\in\Z.
  \end{equation}
\end{lemma}
\begin{proof}
  The left side of the inequality is obvious as, by assumption, $g_\bs \geq
  0$. The right side is obtained by direct application of the
  Cauchy-Schwarz
  inequality.\qed
\end{proof}
\begin{lemma}[Finite support]
  \label{lemma:finitesupport}
  The spatial kernel $c_{g, \bn}(\bu)$ vanishes for 
  $\bu\in\Z^d$ if for any $j=1, \ldots, d$, $|u_j|\geq n_j$.
\end{lemma}
\begin{proof}
  This is immediate from the definition.
\end{proof}
\begin{lemma}[Fourier pair]
  \label{lemma:fourierpair}
  The kernel $\mathcal{F}_\bn(\bk)$, $\bk\in\T^d$, defined
  in~\eqref{eq:generalFejerKernel}, and 
  $c_{g,\bn}(\bu)$, $\bu\in\Z^d$, defined in~\eqref{eq:cgdef}, form a Fourier pair.
\end{lemma}
\begin{proof}
  This is a direct application of the convolution theorem, having noted
  that $c_{g,\bn}(\bu)$ is a discrete convolution.\qed
\end{proof}

\subsection{Definitions}

Our concept of Significant Correlation Contribution (SCC) is defined in
asymptotic terms, since we shall make use of this to establish
consistency of our estimator. More specifically, we consider a sequence of
grids, indexed by $k\in\N$, which goes to infinity, rather than a single
grid.
\begin{definition}[Significant Correlation Contribution (SCC)]
\label{def=significantCorrelation}
A sequence of observed grids \\ $(\mathcal{J}_{\bn_k}, g_k)_{k\in\N}$ leads
to significant correlation contribution for the model family
$\{f_X(\ \cdot\ ;\bgamma) : \bgamma\in\Theta\}$ if it satisfies both
\begin{equation}
\label{eq:SCCdef}
\begin{cases}
\sum_{\bu\in\Z^d}c_{g,\bn_k}(\bu) c^2_X(\bu;\bgamma) \underset{k\to\infty}{=} o\left(\sum{g_\bs^2}\right),\\
\underline{\lim}_{k\rightarrow\infty}
S_k(\btheta_1,\btheta_2) > 0, \quad \forall\btheta_1\neq\btheta_2\in\Theta,
\end{cases}
\end{equation}
where $\underline{\lim}_{k\rightarrow\infty}$ denotes the limit inferior and
where we have defined, for all $\btheta_1, \btheta_2\in\Theta$,
\begin{align}\label{eq:Skdef}
  S_k(\btheta_1,\btheta_2)&\equiv\sum_{\bu\in\Z^d}
  c_{g,\bn_k}(\bu)^2
  \left\{
  c_X(\bu;\btheta_1)-
  c_X(\bu;\btheta_2)
  \right\}^2
  .
\end{align}
\end{definition}
The rationale for this definition of 
$S_k(\btheta_1, \btheta_2)$ is that 
\begin{equation*}
  S_k(\btheta_1, \btheta_2)
  =
  (2\pi)^{-d}\int_{\T^d}
  \left\{\Ink{\btheta_1} - \Ink{\btheta_2}\right\}^2d\bk,
\end{equation*}
due to~\eqref{eq:perFourierSeries} and Parseval's identity for
Fourier series.  We remind the reader that the sums in~\eqref{eq:SCCdef}
and~\eqref{eq:Skdef} are \textit{de facto} finite for a 
given $\bn$, due to the definition of $c_{g,\bn}(\bu)$,
which for fixed $\bn$ has finite support according to
Lemma~\ref{lemma:finitesupport}.  We observe that the above definition
depends on both the sequence of grids, from $c_{g,\bn_k}(\bu)$, and on the
model family, from $c_X(\bu;\bgamma)$. In the rest of this paper we shall say that a
sequence of grids leads to SCC, if the model family that this applies to is
obvious from the context. In addition we define the notion of Highly
Significant Correlation Contribution (HSCC), which will allow us to
establish a convergence rate.
\begin{definition}[Highly Significant Correlation Contribution]
\label{def:HSCC}
A sequence of observed grids \\ $(\mathcal{J}_{\bn_k}, g_k)_{k\in\N}$ leads
to Highly Significant Correlation Contribution for the model family
$\{f_X(\ \cdot\ ;\bgamma) : \bgamma\in\Theta\}$
\begin{itemize}
\item if it leads to Significant Correlation Contribution,
\item if the covariance function is differentiable with respect to the
  parameter vector, and in particular, the quantity
  $\min_{\mathbf{v}\in\R^{p_\theta}, \|\mathbf{v}\| = 1}\sum_{\bu\in\Z^d}
  {c^2_{g,\bn_k}(\bu)\left(\sum_{j=1}^p{v_j \left\{\frac{\partial
        c_X}{\partial\btheta_j}(\bu; \btheta)\right\}}\right)^2}$ is
  asymptotically lower-bounded by a non-zero value, denoted $S(\btheta)$, and
\item if the expected periodogram is twice differentiable with respect to
  the parameter vector, and such that its first and second derivatives are
  both upper-bounded in norm by a constant denoted
  $M_{\partial\theta^2}>0$.
\end{itemize}
\end{definition}
Note that a necessary and more intuitive condition for the second item of
the above definition is that for all $j=1\ldots,d$, $\sum_{\bu\in\Z^d}
{c^2_{g,\bn_k}(\bu)\left[\frac{\partial c_X}{\partial\btheta_j}(\bu; \btheta)\right]^2} $
be lower-bounded by a positive value. Broadly speaking, the first part
of~\eqref{eq:SCCdef} is required so that information grows fast enough. It
can be compared to necessary conditions of decaying covariances in laws of
large numbers, with the additional requirement of accounting for sampling
when considering spatial data. Note that the first part of~\eqref{eq:SCCdef}
is obviously satisfied if the sample covariance sequence is assumed square
summable and the number of observations grows infinite.

The second part of~\eqref{eq:SCCdef} ensures that the expected periodograms
for any two parameter vectors of the parameter set remain
\emph{asymptotically distant} in terms of $\mathcal{L}_2$ norm.  In
Lemma~\ref{lemma=061120182} in Section~\ref{sec:theory}, we show how this
transfers to the expectation of the likelihood function, ensuring that it
attains its minimum at the true parameter vector uniquely. Then in
Lemma~\ref{lemma:cvglkh} we show that the likelihood function converges
uniformly in probability to its expectation over the parameter set, as long
as the first part of~\eqref{eq:SCCdef} is satisfied. This all together will
eventually lead to the consistency of our inference procedure, which is the
result of Theorem~\ref{th:consistency}. Hence the second part
of~\eqref{eq:SCCdef} is required to ensure that the sampling allows to
distinguish parameter vectors based on the expectation of our approximate
likelihood function. To provide further understanding, we shall now consider
some general cases and specific examples with respect to this definition.

\subsection{General sampling cases and sampling example}

Definition~\ref{def=significantCorrelation} extends the definition of SCC
provided by~\citet{guillaumin2017analysis} for time series in two
ways. First, it provides a generalization for spatial data with the notable
difference that spatial \emph{sampling} is more complex than sampling in
time. Indeed, one needs to not only account for the frequency of the
sampling but also for the spatial sampling direction. Secondly, even in
dimension one, the version provided by~\citet{guillaumin2017analysis}
implies the version provided here, while the reverse is not always
true---thus relaxing the assumptions required for consistency. Specifically,
in the second part of~\eqref{eq:SCCdef}, we do not require observing a
specific finite set of lags that will allow identification of the
parameters, unlike~\citet{guillaumin2017analysis}. We now provide more
intuition about SCC through general cases, and a specific example.

\subsubsection{General sampling cases}

Under standard sampling conditions, SCC takes a simpler form, as we show
through the two following lemmata.
\begin{lemma}[SCC for full grids]
\label{lemma:scc1}
If we observe a sequence of full rectangular grids that grow unbounded in
all directions (i.e., $n_j\rightarrow\infty, \ j=1,\ldots,d$), then SCC is
equivalent to the standard assumption that for any two distinct parameter
vectors $\btheta_1, \btheta_2\in\Theta$, the measure of the set
$\{\bk\in\T^d : f_{X,\bdelta}(\bk;\btheta_1) \neq
f_{X,\bdelta}(\bk;\btheta_2)\}$ is positive.
\end{lemma}
\begin{proof}
Please see the Supplementary Material.
\end{proof}
Importantly, we do not require the growth to happen at the same rate in all
directions. We do require that grids grow unbounded in all directions to
obtain this equivalence when we have no further knowledge on the functional
form of the spectral densities. However, in many practical cases, such as
that of an isometric exponential covariance function, our results still hold
if the grid grows unbounded in one direction rather than all. Another
important case for practical applications is that of a fixed shape of
observations that grows unbounded, which is the subject of the following
lemma.
\begin{lemma}[Fixed shape of observations]
\label{lemma:scc2}
Consider a fixed shape defined by a function $\Xi:[0,1]^d \mapsto \{0,1\}$,
and let $g_{k,\bs} = \Xi(\bs\circ \bn_k^{-1}),
\forall\bs\in\mathcal{J}_{\bn_k}, \forall k\in\N$. If the grids grow
unbounded in all directions, and if the interior of the support of $\Xi$ is
not empty, then SCC is again equivalent to the condition stated in
Lemma~\ref{lemma:scc1} on the parametric family of spectral densities.
\end{lemma}
\begin{proof}
Please see the Supplementary Material.
\end{proof}
In Section~\ref{sec:missing} we provide a simulation study for the
particular case of a circular shape of observations which satisfies this lemma. 
Finally, from a frequency-domain point of view, the second part of SCC can
be understood according to the following lemma.
\begin{lemma}
The second part of SCC is equivalent to
\begin{align*}
S_k(\bm{\theta}_1,\bm{\theta}_2)=\int_{\T^d}\left| \int_{\T^d} \mathcal{F}_{\bn_k}(\bk')
\left\{  f_X(\bm{\omega}'-\bm{\omega};\bm{\theta}_1)- f_X(\bm{\omega}'
-\bm{\omega};\bm{\theta}_2)\right\}d\bm{\omega}'\right|^2\,d\bm{\omega}>0.
\end{align*}
\end{lemma}
\begin{proof}
This comes as a consequence of Lemma~\ref{lemma:fourierpair} and standard
Fourier theory.
\end{proof}
Most importantly, note that in general SCC requires more than the necessary
requirement that for two distinct parameters, the expected periodograms for
the sequence of grids should be non-equal, and this is to correctly account
for missing data mechanisms and their impact on consistency.  To obtain SCC, cf.~\eqref{eq:Skdef}, this means we require that information about
$\bm{\theta}_1$ relative to $\bm{\theta}_2$ grows as we observe ever larger
patches of data.  Our vulnerability to adversarial sampling will depend on
the structure of the covariance pattern under study; for example if we only
sample along a boundary then between points on the boundary we get
information about very short scales, or between parts of the boundary only
very long scales.  We will now
provide further intuition about SCC through a specific example.

\subsubsection{Examples}
We consider a separable exponential covariance function ($d=2$ here) with
parameters $\rho_1 > 0$ and $\rho_2 > 0$ defined by
\begin{equation}
  c_X(\bu) = \sigma^2 \exp\left(-\rho_1^{-1}|u_1|\right)
  \exp\left(-\rho_2^{-1}|u_2|\right), \quad \bu\in\R^2.
\end{equation}
If we sample along one axis only, it is clear that the second part of SCC
fails as the range parameter along the other axis cannot be identified from
the data. In contrast, the second part of SCC will be satisfied for this
particular model and for a full rectangular grid as long as $n_1 \geq 2$ and
$n_2\geq 2$. The first part of SCC is valid as long as the sample size grows
to infinity, since the sample covariance function is square summable. For
this model class, SCC is therefore satisfied if and only if $n_1 \geq 2$ and
$n_2\geq 2$ and $n_1 n_2$ goes to infinity. It is also worth observing that
under those conditions, the convergence rate of our estimator will be
$\mathcal{O}\left((n_1n_2)^{-1/2}\right)$ (see
Theorem~\ref{th=asymptNormality}), irrespective of the ratio $n_1/n_2$,
which, in particular, is allowed to converge to zero or infinity.  The
Supplementary Material provides an example where SCC fails.
  
These two examples show the flexibility of SCC compared to standard
assumptions. They show that the two parts of SCC are complimentary and help
understand their role in establishing consistency. The second part is
required to ensure identifiability of the parameter vector from the expected
periodogram. The first part of SCC is required to ensure that some form of
law of large numbers holds for linear combinations of the periodogram.

\subsubsection{Application to randomly missing data}

Our extended definition of SCC can be applied to the scenario where data are
missing at random, on the condition that the randomness scheme for the
missing data is independent from that of the observed process. For such
applications we shall say that a sequence of grids leads to SCC almost
surely if \eqref{eq:SCCdef} is satisfied almost surely under the probability
that defines the missingness scheme. If a sequence of grids leads to SCC
almost surely, it is easy to verify that all our consistency results derived
in Section~\ref{sec:theory} still hold. Yet again for consistency we need
our information about $\bm{\theta}_1$ relative to $\bm{\theta}_2$ to grow as
we observe ever larger patches of data with randomly missing
observations. This need not correspond to a linear relationship between the
observed number of samples and the nominal number of samples in the
observational domain, but instead depends on the true covariance of the
random field under study.

A simple application of these considerations is one where each point of a
rectangular grid is observed or missed according to a Bernoulli random
variable (with a positive probability of being observed), independently of
other points of the grid, and independently of the observed process.

\subsubsection{Extension to multivariate random fields}\label{sec:multivariate}

In this section we define the notation necessary for multivariate random
fields. Assume we observe $p\geq 1$ random fields jointly,
\begin{equation}
Y_{\bs}^{(q)}=g_{\bs}^{(q)}X_{\bs}^{(q)},\quad \bs\in\R^d,\; q\in\{1,\dots, p\},
\end{equation}
and allow the observation pattern defined by the modulations $g_{\bs}^{(q)}$
to differ across the $p$ random fields.  This is a realistic observation
scheme in many real-world settings, e.g.  for multi-spectral and repeated
remote-sensing observations, where cloud cover will contribute to varying
degrees of censoring, yet with the underlying grids essentially unchanged
\cite[e.g.][]{Song+2018}.

Just like~\cite{rao1967cross} we compute the cross-periodogram of pairs of
processes. Assume we observe the $p$-variate process $\mathbf{X}_s$ and that
for each process sampled at the same grid we have a masking function
$g_s^{(q)}$ for $1\leq q\leq p$, so that we can incorporate some variation
in sampling frequency, see, e.g.,~\cite{gotway2002combining}.  We calculate
the DFT to be
\[J^{(q)}(\bm{\omega})=
\frac{(2\pi)^{-d/2}}{\sqrt{\sum_{\bs\in{\cal J}_{\mathbf{n}}}g_\bs^{(q)2}}}
\sum_\bs g_\bs^{(q)} 
X_s^{(q)} 
\exp\{-i \bs\cdot \bm{\omega}\}
,\]
and we collect the DFTs in the vector
$\mathbf{J}(\bm{\omega})^T=\begin{pmatrix}
J^{(1)}(\bm{\omega})&
\dots&
J^{(p)}(\bm{\omega})
\end{pmatrix}.
$
We can define the cross-periodogram from this quantity:
\begin{equation*}
    I_{\mathbf{n}}^{(qr)}(\bm{\omega})=
    J^{(q)}(\bm{\omega})J^{(r)\ast}(\bm{\omega}).
\end{equation*}
We can define the expected periodogram at a given wave-number $\bk$ 
by the $p\times p$ matrix,
\begin{equation*}
\overline{\mathbf{I}}(\bm{\omega})=\E\{ \mathbf{J}(\bm{\omega})\mathbf{J}^H(\bm{\omega})\},
\end{equation*}
and this is in turn requiring us to define notation for the cross-covariance
function
\[c^{(qr)}_{\mathbf{X}}(\mathbf{u})=\cov\{X_{\bs}^{(q)},
X_{\bs+\mathbf{u}}^{(r)}\}.\]
The expected periodogram matrix therefore has the elements
\begin{align}
\nonumber
\overline{\mathbf{I}}^{(qr)}(\bm{\omega})&=
\nonumber
\frac{(2\pi)^{-d}}{\sqrt{\sum_{\mathbf{s}_1\in{\cal
        J}_{\mathbf{n}}}g_{s_1}^{(q)2}\sum_{\bs_2\in{\cal
        J}_{\mathbf{n}}}g_{s_2}^{(r)2}}}\sum_{\bs}
\sum_{\mathbf{u}}g_{\bs}^{(q)} g_{\bs+\mathbf{u}}^{(r)}  
c_{\mathbf{X}}^{(qr)}({\mathbf{u}}) \exp\{-i {\mathbf{u}}\cdot \bm{\omega}\}
.\end{align}
Then with the definition
\begin{equation*}
c_{g,\mathbf{n}}^{(qr)}({\mathbf{u}})=
\frac{\sum_{\mathbf{s}\in{\cal
      J}_{\mathbf{n}}}g_{\mathbf{s}}^{(q)}g_{\mathbf{s+u}}^{(r)}}{\sqrt{\sum_{\mathbf{s}\in{\cal
        I}_{\mathbf{n}}}g_s^{(q)2}\sum_{\mathbf{s}\in{\cal
        I}_{\mathbf{n}}}g_s^{(r)2}}}, 
\end{equation*}
the expected periodogram takes the form of
\begin{align}
\nonumber
\overline{{I}}^{(qr)}(\bm{\omega})&=(2\pi)^{-d} \sum_{\mathbf{u}}c_{g,\mathbf{n}}^{(qr)}({\mathbf{u}})
c_X^{(qr)}({\mathbf{u}}) \exp\{-i {\mathbf{u}}\cdot \bm{\omega}\}.
\end{align}
The computation of the above quantity can be carried out by applying
Lemma~\ref{lemmacompfft} for each $(q,r)\in\left\{1,\cdots, p\right\}^2$.
The Whittle likelihood is then trivially extended to this setting as was
already remarked upon by~\cite{whittle1953analysis} and
\cite{shea1987estimation}.  The Whittle likelihood in the multivariate
setting can be re-written
as~\cite[e.g.,][]{hosoya1982central,hosoya1993correction,kakizawa1997parameter},
\begin{align}
\label{multdebias}
    \ell_{\mathbf{n}}(\bm{\theta})&=|\mathbf{n}|^{-1}\sum_{\bm{\omega}}
    \left\{\log {\mathrm{det}}\{
    \overline{\mathbf{I}}(\bm{\omega};\bm{\theta})\}
    +\mathbf{J}^H(\bm{\omega})\overline{\mathbf{I}}^{-1}(\bm{\omega};\bm{\theta})
    \mathbf{J}(\bm{\omega})\right\}. 
\end{align}
We can still use this for estimation, only requiring that the eigenvalues of
$\overline{\mathbf{I}}(\bm{\omega})$ are positive in the neighbourhood of
$\bm{\theta}$.  
We extend the definition of Significant Correlation Contribution (SCC)
to the multivariate SCC (m-SCC) as follows.
\begin{definition}[Multivariate SCC]
A sequence of observed grids $(\mathcal{J}_{\bn_k}, g_k)_{k\in\N}$ leads to
significant correlation contribution for the multivariate model family
$\{f(\ \cdot\ ;\bgamma) : \bgamma\in\Theta\}$ if it satisfies
\begin{equation}
\label{eq:multSCCdef}
\begin{cases}
  \sum_{q,r=1}^p
  \frac{
    \sum_{\bu}
	{	
	  c_g^\pqr(\bu){c_X^\pqr(\bu)}^2
	}
  }
       {
         \sqrt{
           \sum{{g_\bs^\pq}^2}
           \sum{{g_\bs^\pr}^2}
         }
       }
       = o(1),\\
       \underline{\lim}_{k\rightarrow\infty}
       S_k(\btheta_1,\btheta_2) > 0, \quad \forall\btheta_1\neq\btheta_2\in\Theta,
\end{cases}
\end{equation}
where $S_k(\btheta_1,\btheta_2)$ has been changed to accomodate for
the multivariate scenario,
\begin{equation}\label{eq:Skmultdef}
  S_k(\btheta_1,\btheta_2)\equiv
  \sum_{q,r=1}^p
  \sum_{\bu\in\Z^d}
      {c^\pqr_{g,\bn_k}(\bu)}^2
      \left\{c^\pqr_X(\bu; \btheta_1)-c^\pqr_X(\bu; \btheta_2)\right\}^2, 
      \quad\forall
      \btheta_1,\btheta_2\in\Theta^2.
\end{equation}
\end{definition}

\section{Theory}
\label{sec:theory}

In this section we first provide the proof of our estimator's consistency in
the general setting that encompasses both non-Gaussian and multivariate
random fields.  We then also derive its rate of convergence and the
asymptotic distribution in univariate Gaussian and non-Gaussian settings. We
assume the following set of assumptions holds in order to establish
consistency.
\begin{assumption}[Consistency assumptions]\label{ass:1}
$\left.\right.$
\begin{enumerate}[label={(1\alph*)}]
\item The parameter set $\Theta$ is compact.
\item\label{Ass:sdf}The aliased spectral density $f_{X,\delta}(\bk;\bgamma),
  \bk\in\T^d, \bgamma\in\Theta$ is bounded above by
  $f_{\delta,\text{max}}<\infty$ and below by
  $f_{\delta,\text{min}}>0$. Additionally, $f_{X,\delta}(\bk;\bgamma)$
  admits a derivative with respect to the parameter vector $\bgamma$, which
  is upper-bounded in norm by $M_{\partial\theta}$.  For a multivariate
  random field, we similarly require that the eigenvalues of the matrix
  spectral density $f(\bk; \bgamma)$ are lower and upper-bounded by positive
  analogous constants $f_{\delta, \mathrm{min}}$ and $f_{\delta,
    \mathrm{max}}$, respectively.
\item The sequence of observation grids leads to SCC for the considered
  model family.
\item The modulation $g_\bs$, $\bs\in\Z^d$, takes its values in the interval
  $\left[0,1\right]$.
\item The random field $X(\bs)$ has finite and absolutely summable
  fourth-order cumulants.
\end{enumerate}
\end{assumption}

Two main asymptotic frameworks coexist in spatial data analysis, namely
infill asymptotics and growing-domain
asymptotics~\citep{Zhang2005towards}. We study our estimator within the
latter framework, which we consider most plausible for finite-resolution
remote-sensing observations, imposing that the sample size goes to infinity
(through our SCC assumption) while having fixed $\bdelta$. Our set of
assumptions is standard, except for SCC, which generalizes the standard
assumption of a fully-observed rectangular grid associated with the
requirement that two distinct parameter vectors map to two spectral
densities that are distinct on a Lebesgue set of non-zero measure.
\begin{theorem}[Consistency]
\label{th:consistency}
Under Assumption~\ref{ass:1}, the sequence of estimates
$\widehat{\btheta}_k$ defined by~\eqref{eq:DebiasedWhittleLKH} converges in
probability to the true parameter vector $\btheta$ as the observational
domain diverges.
\end{theorem}
This result holds for a wide class of practical applications, as
\begin{itemize}
\item we do not require the rectangular grid to be fully observed. We allow
  for a wide class of observational domains, as long as SCC is satisfied;
\item we do not require the grid to grow at the same rate along all
  dimensions. Classical frequency-domain results make use of the fact that
  the multilevel Block Toeplitz with Toeplitz Blocks covariance matrix has
  its eigenvalues distributed as the spectral density. However this result
  only holds under the assumption that the sampling grid grows at the same
  rate along all dimensions.
\end{itemize}
Theorem~\ref{th:consistency} holds for Gaussian, non-Gaussian, and multivariate Gaussian
random fields that satisfy the required conditions. The proof of Theorem~\ref{th:consistency} is the same for all
three cases, but some lemmata and propositions on which
Theorem~\ref{th:consistency} relies will require additional detail for each
case.  We shall prove Theorem~\ref{th:consistency} in a series of steps.  We
start by introducing some additional notation.

\subsection{Additional notation}

The vector of the values taken by the process on the rectangular grid $\Jx$
is denoted $\bX = [X_0, \ldots, X_{\n-1}]^T$, where points are ordered into
a vector according to the colexicographical order. Therefore in dimension
$d=2$, $X_0,\ldots, X_{n_1-1}$ are values from the first row of $\Jx$,
$X_{n_1},\ldots, X_{2n_1-1}$ are values from the second row, and so
on. Similarly we denote $\bg$ the vector of the values taken by the
modulation function on $\Jx$, with points ordered in the same way. We also
denote by $\bs_0, \ldots, \bs_{\n-1}$ the locations of the grid ordered
according to the same order, such that $X_0 = X(\bs_0), X_1 = X(\bs_1),$
etc.

We also denote by $G$ the diagonal matrix with elements taken from $\bg$,
such that the vector corresponding to the observed random field (rather than
$\bX$ which corresponds to the random field on the rectangular grid $\Jx$)
is given by the matrix product $G\bX$.

Finally, for any vector $\mathbf{v}\in\R^p$ we shall denote by
$\|\mathbf{v}\|_q$ its $\mathcal{L}_q$ norm (in particular $\|\cdot\|_2$ is
the Euclidean norm), and for any $p\times p$ matrix $A$, $\|A\|$ shall
denote the spectral norm, i.e., the $\mathcal{L}_2$-induced norm,
\begin{equation}\label{eq:spectralNorm}
\|A\| = \max_{\mathbf{v}\in\R^p, 
\mathbf{v}\neq\mathbf{0}}\frac{\|A\mathbf{v}\|_2}{\|\mathbf{v}\|_2}.
\end{equation}
We remind the reader that if $H$ is a Hermitian matrix, since
$\|H\mathbf{v}\|_2^2 = \mathbf{v}^* H^*H\mathbf{v} = \mathbf{v}^*
H^2\mathbf{v}$, the spectral norm of $H$ is its spectral radius, i.e.,
\begin{equation*}
\|H\| = \rho(H) \equiv \max\{|\lambda|: \lambda \text{ eigenvalue of }H\}.
\end{equation*}

\subsection{Distributional properties of the periodogram}

It is well known for time series that the bias of the periodogram as an
estimator of the spectral density is asymptotically zero
\citep{koopmans1995spectral}. However, for spatial data in dimension $d\geq
2$, the decay of the bias of the periodogram is known to be the dominant
factor in terms of mean-squared error
\citep{dahlhaus1987edge}. Additionally, the bias is asymptotically zero
under often non-realistic assumptions, such as: full knowledge of the
aliased spectral density, fully observed grid, growth of the domain in all
directions. By directly fitting the expectation of the periodogram, rather
than the spectral density, we circumvent this major pitfall of the Whittle
likelihood for random fields. Having removed the effect of bias, we are left
with studying the correlation properties of the periodogram. We show that
the variance of a bounded linear combination of the periodogram at Fourier
frequencies goes to zero. This is the result of
Proposition~\ref{prop=varianceLinearCombinations}, which we use later, in
Lemma~\ref{lemma:cvglkh}, to prove that if
  Assumption~\ref{ass:1} holds our likelihood function converges uniformly
in probability to its expectation.

\begin{proposition}[Variance of linear functionals of the periodogram]
\label{prop=varianceLinearCombinations}
Suppose Assumption~\ref{ass:1}\\ holds and the random field is
  Gaussian.  Let $a_k(\bk)$ be a family of functions with support $\T^d$,
indexed by $k\in\N$, and uniformly bounded in absolute value. Then,
\begin{equation}
\label{eq:prop1Gaussian}
\var\left\{|\bn_k|^{-1}\sum_{\bk\in\Omega_{\bn_k}} a_k(\bk)I_{\bn_k}(\bk)
\right\} = \mathcal{O}\left\{
\frac{\sum_{\bu\in\Z^d}{c_{g,k}(\bu)\, c^2_X(\bu) }}
{\sum{g_\bs^2}}
\right\}
.
\end{equation}
\end{proposition}
\begin{proof}
Please see the Supplementary Material.
\end{proof}
\begin{corollary}[Extension to non-Gaussian random fields]
\label{cor=varianceLinearCombinations}
Suppose Assumption~\ref{ass:1} holds.  Let $a_k(\bk)$ be a
family of functions with support $\T^d$, indexed by $k\in\N$, and uniformly
bounded in absolute value. Then, for non-Gaussian random fields, the
variance of linear combinations of the periodogram behaves according to
\begin{equation}
  \label{eq:nonGaussianProp1}
  \var\left\{|\bn_k|^{-1}\sum_{\bk\in\Omega_{\bn_k}} a_k(\bk)I_{\bn_k}(\bk)
  \right\} = \mathcal{O}\left\{
  \frac{\sum_{\bu\in\Z^d}{c_{g,k}(\bu)\, c^2_X(\bu) }}
       {\sum{g_\bs^2}}
       + 
       \frac
           {
             \left|\bn_k\right|
           }
           {
             \left(\sum g_{\bs}^2\right)^2
           }
           \right\}.
\end{equation}
\end{corollary}
\begin{proof}
Please see the Supplementary Material.
\end{proof}
In the non-Gaussian case, the first requirement of SCC is adapted by
accounting for the additional term in~\eqref{eq:nonGaussianProp1} compared
to~\eqref{eq:prop1Gaussian}.  If we observe a full rectangular grid with no
tapering, then we have $\sum{g_\bs^2} = |\bn_k|$, the total number of points
of the grid.  If we assume square summability of the covariance function,
then under the Gaussian assumption, the variance under study vanishes even
if $\sum{g_\bs^2} = |\bn|^{1/2}$.  As we see
with~\eqref{eq:nonGaussianProp1}, this may not hold anymore for non-Gaussian
data.  One such example would be on a $d$-rectangular grid. Assume we
nominally sampled sides of length $\ell$ on a $d$-dimensional cube. If we
replace this by sampling $\Theta(\sqrt{\ell})$ points, leaving the rest as
missing data then $\sum{g_\bs^2} = |\bn|^{1/2}$, and convergence
is no longer guaranteed in the non-Gaussian case. If we no longer have a regularly
sampled grid with some missing data, but a very complex spatial sampling
then the DFT may not be the most convenient implementation, and we may adapt
other methods,
e.g.~\cite{barnett2019parallel}. From~\eqref{eq:nonGaussianProp1}, however,
we see that for non-degenerate sampling scenarios, we can expect consistency
of our estimator even for non-Gaussian random fields.

Finally, for multivariate random fields, the same question arises about the
variance of sesquilinear forms involving the elements of the vector-Fourier
transform. We present this as a second corollary to
Proposition~\ref{prop=varianceLinearCombinations}.
\begin{corollary}[Extension to multivariate random fields]
\label{cor=varianceSesLinearCombinations}
Let $\{{\mathbf{A}}_k(\bk)\}$ be a family of matrix-valued functions with
support $\T^d$, indexed by $k\in\N$, and uniformly bounded in terms of the
maximum eigenvalues across all frequencies by $\lambda_{\max}$.  If the
random field is $p$-multivariate Gaussian with absolutely summable
cross-covariance sequence, the variance of sesquilinear functionals of the
discrete Fourier transform behaves according to,
\begin{equation*}
\var\left\{|\bn_k|^{-1}\sum_{\bk\in\Omega_{\bn_k}} 
{\mathbf{J}}_{\mathbf{n}_k}^\ast(\bk){\mathbf{A}}_k(\bk)
{\mathbf{J}}_{\mathbf{n}_k}(\bk)
\right\} = \mathcal{O}\left\{
\sum_{q,r=1}^p\frac{\sum_{\bs}c_g^{(qr)}(\bs)\,c^{(qr)2}_{\mathbf{X}}(\bs)}
{\sqrt{\sum_{\bs_1}g_{\bs_1}^{(q)^2}\sum_{\bs_2}g_{\bs_2}^{(r)^2}}}  
\right\}. 
\end{equation*}
\end{corollary}
\begin{proof}
Please see the Supplementary Material.
\end{proof}

\subsection{Lemmata required for Theorem~\ref{th:consistency}}
\label{sec:consistency}
All the lemmata in this section suppose that
  Assumption~\ref{ass:1} holds.  We provide all the proofs of this section
in the Supplementary Material.  To establish consistency we introduce some
specific notation for the expectation of our quasi-log-likelihood,
\begin{equation}\label{eq:formulaexpectedlikelihood}
  \widetilde{\ell}_\bn(\bgamma) = \E_{\btheta}\left\{
  \ell_\bn(\bgamma)\right\} =
  |\bn|^{-1}\sum_{\bk\in\Omega_{\bn}}\left\{\log\In{\bgamma} +  
  \frac{\In{\btheta}}{\In{\bgamma}}\right\},
  \quad\forall\bn\in(\N^+)^d \setminus \{\mathbf{0}\}, \forall\bgamma\in\Theta,  
\end{equation}
which we shall regard as a function of $\bgamma$.  For multivariate random
fields this is extended according to,
\begin{align*}
\widetilde{\ell}_{\mathbf{n}}(\bm{\gamma})&=
\E_{\bm{\theta}}\left\{ \ell_{\mathbf{n}}(\bm{\gamma}) \right\}=
|\mathbf{n}|^{-1}\sum_{\bm{\omega}}
\left\{\log {\mathrm{det}}\{
\overline{\mathbf{I}}_\bn(\bm{\omega};\bm{\gamma})\}
+{\mathrm{trace}}\left[\overline{\mathbf{I}}_\bn^{-1}(\bm{\omega};\bm{\gamma})
\overline{\mathbf{I}}_\bn(\bm{\omega};\bm{\theta})\right]\right\}. 
\end{align*}

The following lemma relates the minimum of that function to the true
parameter vector (with no uniqueness property as of now).
\begin{lemma}[Minimum of the expected quasi-likelihood function]
\label{lemma:minofexptedlkh}
The expected likelihood\\ function attains its minimum at the true parameter 
value, i.e,
\begin{equation}
\widetilde{\ell}_\bn(\btheta) = \min_{\bgamma\in\Theta}{\widetilde{\ell}_\bn(\bgamma)}.
\end{equation}
\end{lemma}
We shall also make repeated use of the following lemma.
\begin{lemma}[Lower and upper bounds on the expected periodogram]
\label{lemma:bounds_periodogram}
The expected\\ periodogram satisfies, for all parameter vector 
$\bgamma\in\Theta$, and at all wave-numbers $\bk\in\T^d$, for any 
$\bn\in\left(\N^+\right)^d$,
\begin{equation*}
f_{\delta,\text{min}} \leq \In{\bgamma} \leq f_{\delta,\text{max}}.
\end{equation*}
\end{lemma}
We now provide additional lemmata which are key to proving the consistency
of our maximum quasi-likelihood estimator. Lemma~\ref{lemma=061120182}
states that the expected likelihood value at a parameter vector distinct
from the true parameter value is asymptotically bounded away from the
expected likelihood at the true parameter value. This comes as a consequence
of the second part of SCC and the upper-bound on the spectral densities of
the model family.
\begin{lemma}[Identifiability from the expected likelihood function]
\label{lemma=061120182}
Let $\bgamma\in\Theta$ distinct from~$\btheta$. Then,
\begin{equation}\label{eq:05112018}
  \underline{\lim}_{k\rightarrow\infty} \  
\left|\widetilde{\ell}_{\bn_k}(\bgamma) - \widetilde{\ell}_{\bn_k}(\btheta)\right| > 0,
\end{equation}
where $\underline{\lim}_{k\rightarrow\infty}$ denotes the limit inferior as
$k$ goes to infinity.
\end{lemma}
For multivariate random fields, the proof of Lemma~\ref{lemma=061120182}
requires an additional simple lemma,
\begin{lemma}
\label{lemma=traceinequality}
Let $H_1, H_2$ be two Hermitian positive definite Hermitian matrices. Then,
\begin{equation}
  \mathrm{trace}
  \left[
    H_1 H_2
    \right]^2
  \geq
  \left(\min \mathrm{sp}\left(H_1\right)\right)^2 \mathrm{trace}\left[H_2\right]^2,
\end{equation}
where $\mathrm{sp}\left(H_1\right)$ denotes the set of eigenvalues of $H_1$,
which are all positive.
\end{lemma}
Lemma~\ref{lemma=06112018} now states a form of regularity of our expected
likelihood functions. It relies on our regularity assumption on the
spectral model family, where we have assumed the existence and boundedness
of the partial derivatives with respect to the parameter vector, per Assumption~\ref{Ass:sdf}.
\begin{lemma}
\label{lemma=06112018}
Let $\bgamma\in\Theta$ and let $(\bgamma_k)_{k\in\N}$ be a sequence of
parameter vectors that converges to $\bgamma$. Then,
\begin{equation}
\widetilde{\ell}_{\bn_k}(\bgamma_k)-\widetilde{\ell}_{\bn_k}(\bgamma)
\longrightarrow 0, \ \ (k\longrightarrow\infty).
\end{equation}
\end{lemma}

\begin{lemma}
\label{lemma:cvgtotheta}
Let $\bgamma_k\in\Theta^{\N}$ be a sequence of parameter vectors such that
$\widetilde{\ell}_{\bn_k} (\bgamma_k)-\widetilde{\ell}_{\bn_k}(\btheta)$
converges to zero as $k$ tends to infinity. Then $\bgamma_k$ converges to
$\btheta$.
\end{lemma}

And finally, the following lemma helps us understand how the likelihood
function, as a random element, behaves with regard to the expected
likelihood function.
\begin{lemma}[Uniform convergence in probability of the likelihood function]
\label{lemma:cvglkh}
The log-likelihood function $\ell_{\bn_k}(\cdot)$ converges uniformly in
probability to $\widetilde{\ell}_{\bn_k}(\cdot)$ over the parameter set
$\Theta$ as $k$ goes to infinity.
\end{lemma}

With these lemmata we have all the necessary results to establish
Theorem~\ref{th:consistency}. This theorem is important as it establishes
the consistency of our estimator under a very wide range of sampling schemes
and model families. We contrast our results with those of
\cite{dahlhaus1987edge}, \cite{Guyon2008}, as well as
\cite{fuentes2007approximate}. The insight from
Theorem~\ref{th:consistency}, as compared to the insight of the need for
tapering provided by~\cite{dahlhaus1987edge} is clear. The aim of this paper
is to balance computational tractability with estimation performance. Very
standard assumptions allow us to still derive the results required for
estimation.

\subsection{Convergence rate and asymptotic normality}

We now study the convergence rate and asymptotic distribution of our
estimates within the increasing-domain asymptotics framework. In
Theorem~\ref{th=asymptNormality} we establish a convergence rate in the
general framework of HSCC (Definition~\ref{def=significantCorrelation}) for
both Gaussian and non-Gaussian random fields, and we also establish
asymptotic normality in the scenario of a Gaussian random field
  observed on a full grid.  Under further requirements
  (Assumption~\ref{ass:asymptoticNormality2}), asymptotic normality is shown
  for non-Gaussian random fields in Theorem~\ref{prop:asympvarianceest},
  together with a limiting form of the covariance structure of our
  estimator.

To prove our theorems, we first need to understand better the behaviour of
quantities of the form
$|\bn|^{-1}\sum_{\bk\in\Omega_{\bn_k}}{w_k(\bk)I_{\bn}(\bk)}$, for some
weights~$w_k$. In Proposition~\ref{prop=varianceLinearCombinations}, we had
already showed that under mild conditions, their variance vanished at a rate
driven by the number of observed points. Now in
Proposition~\ref{prop=normalityBoundedLinearCombinationsPeriodogram}, and
under the assumption of a full grid, by writing this quantity as a quadratic
form in the random vector~$\bX$ and extending a result
by~\citet{grenander1958toeplitz}, we show that this quantity is
asymptotically normally distributed, under mild conditions on the family of
functions $w_k(\cdot)$. Before getting there, we need the following
intermediary result, which extends a standard result for Toeplitz matrices
to their multi-dimensional counterpart, Block Toeplitz with Toeplitz Block
matrices.
\begin{lemma}[Upper bound on the spectral norm of the covariance matrix]
\label{proposition=upperboundnormcovmat}
Suppose Assumption 1 holds.
In the case of a full grid, the spectral norm of $C_{\bX}$ and that of its
inverse are upper-bounded according to
\begin{equation*}
\|C_{\bX}\| \leq f_{\delta,\text{max}}, \ \ \ 
\|C_{\bX}^{-1}\| \leq f_{\delta,\text{min}}^{-1}.
\end{equation*}
\end{lemma}
\begin{proof}
Please see the Supplementary Material.
\end{proof}

\begin{proposition}[Asymptotic normality of linear
combinations of the periodogram]
\label{prop=normalityBoundedLinearCombinationsPeriodogram}
Suppose Assumption~\ref{ass:1} holds and that the random field is Gaussian
and observed on a full grid.  Let $w_k(\cdot), k\in\N$ be a family of
real-valued functions defined on $\mathcal{T}^d$ bounded above and below by
two constants, denoted $M_W, m_W > 0$ respectively.  Then
$|\bn|^{-1}\sum_{\bk\in\Omega_{\bn_k}}{w_k(\bk)I_{\bn}(\bk)}$ is
asymptotically normally distributed.
\end{proposition} 
\begin{proof}
Please see the Supplementary Material.
\end{proof}

Before finally establishing our convergence rates, as well as the asymptotic
normality in the case of a Gaussian random field observed on a full grid, we
require one additional set of assumptions.
\begin{assumption}[Assumptions for convergence rate and asymptotic normality]
$\left.\right.$
\label{assumption:rate}
\begin{enumerate}[label={(2\alph*)}]
\item
\label{ass:interior}
The interior of $\Theta$ is non-null 
 and the true length-$p_\theta$ parameter
vector $\btheta$ lies in the interior of $\Theta$.
\item
\label{ass:rate2}
The sequence of observation grids leads to HSCC for the considered model family.
\end{enumerate}
\end{assumption}
The following lemma relates HSCC to the minimum eigenvalue of the
expectation of the Hessian matrix of $l(\cdot)$ at the true parameter
vector.
\begin{lemma}
\label{lemma:mineigenvalueHessian}
Under HSCC, the minimum eigenvalue of the expectation of the Hessian matrix
(with respect to the parameter vector) at the true parameter, given by
\begin{equation}
\left(
|\bn_k|^{-1}
\sum_{\bk\in\Omega_{\bn_k}}
    {
      \overline{I}_{\bn_k}(\bk;\btheta)^{-2}
      \nabla_{\theta}\overline{I}_{\bn_k}(\bk;\btheta)
      \nabla_{\theta}\overline{I}_{\bn_k}(\bk;\btheta)^T
    }
    \right),
\end{equation}
is lower-bounded by $S(\btheta)$, which was defined in
Definition~\ref{def:HSCC}.
\end{lemma}
\begin{proof}
This can be established by a direct adaptation of Lemma~7
of~\citet{guillaumin2017analysis}.
\end{proof}

\begin{theorem}[Convergence rate and asymptotic normality of estimates]
\label{th=asymptNormality}
Suppose Assumptions~\ref{ass:1} and~\ref{assumption:rate} hold.
Our estimate converges in probability with rate
\begin{equation*}
  r_k=
  \left(
  \frac{\sum_{\bu\in\Z^d}{c_{g,k}(\bu)\, c^2_X(\bu) }}
       {\sum{g_\bs^2}}
       + 
       \frac
           {
	     \left|\bn_k\right|
           }
           {
	     \left(\sum g_{\bs}^2\right)^2
           }
           \right)^{1/2}.
\end{equation*}
If the random field is Gaussian the convergence rate simplifies to,
\begin{equation*}
r_k=
\left(
\frac{\sum_{\bu\in\Z^d}{c_{g,k}(\bu)\, c^2_X(\bu) }}
{\sum{g_\bs^2}}
\right)^{1/2}.
\end{equation*}
In addition, if the grid is fully-observed and the random field in Gaussian,
then $\hat{\btheta}$ is asymptotically normally distributed. 
\end{theorem}
\begin{proof}
Please see the Supplementary Material.
\end{proof}
Note that in Theorem~\ref{th=asymptNormality} we do not make
  assumptions about the dimensions of the observation domain, as is usually
  the case for Whittle-type estimators where a common growth rate in all
  directions is typically assumed. Asymptotic normality of our estimate can
  also be established for non-Gaussian random fields under appropriate
  assumptions on high-order cumulants, which we introduce below.
  \begin{assumption}
    \label{ass:asymptoticNormality2}
    $\left.\right.$
    \begin{enumerate}[label={(3\alph*)}]
    \item \emph{Observation domain.} The grid is fully observed, and
    we set $g_{\bs}=1$ on the grid and $0$ otherwise. Additionally,
    we require the domain to be unbounded in all directions for
    asymptotic forms to hold.
    \item \emph{Higher-order homogeneity.} Joint moments of any order are
      finite and for any positive integer $L\geq 2$ and locations $\bs_1,
      \cdots, \bs_L\in\R^d$, for any $\bu\in\R^d$, $\cum\left[X_{\bs_1},
        \cdots, X_{\bs_L}\right]=\cum\left[X_{\bs_1+\bu}, \cdots,
        X_{\bs_L+\bu}\right]$.  If this assumption holds we define for
      $\bu_1, \cdots, \bu_{L-1}\in\R^d$,
      \begin{equation}
	c_L(\bu_1, \cdots, \bu_{L-1}) = \cum\left[X_{\bs_0},
          X_{\bs_0+\bu_1}, \cdots,
          X_{\bs_0+\bu_{L-1}}\right], \forall\bs_0\in\R^d.
      \end{equation}
      In particular $c_2(\cdot)$ is just the autocovariance function
      of the random field.
    \item \emph{Short-length memory.} For any positive integer $L\geq 2$,
      \begin{equation}
	\sum_{\bu_1, \cdots, \bu_{L-1}\in\R^d} 
	\left(1+\|\bu_j\|^d\right)
	\left| c_L(\bu_1, \cdots, \bu_{L-1}) \right| < \infty, \ j=1,\cdots,d.
      \end{equation}
    \end{enumerate}
  \end{assumption}
  \begin{proposition}
    \label{prop:asympNormalNonGaussian}
    Suppose Assumptions~\ref{ass:1} and~\ref{ass:asymptoticNormality2}
    hold. Let $\mathbf{w}_k(\cdot)$ be uniformely-bounded vector-valued
    functions from $\mathcal{T}^d$ to $\R^d$ such that
    $\{\mathbf{w}_k(\cdot)\}$ converges to $\mathbf{w}(\cdot)$ pointwise,
    where $\mathbf{w}(\cdot)$ is a Rieman-integrable function with values in
    $\R^d$.  Then, $|\bn|^{-1}\sum_{\bk\in\Omega_\bn}{
      \mathbf{w}_k(\bk)I_{\bn}(\bk) }$ is asymptotically
    jointly-normal. Additionally, suppose the grid grows to infinity in all
    directions, the asymptotic covariance structure of
    $|\bn|^{-1}\sum_{\bk\in\Omega_\bn}{ \mathbf{w}_k(\bk)I_{\bn}(\bk) }$ is
    then determined by
    \begin{align*}
      2^d\pi^d &
      |\bn|^{-1}\int_{\mathcal{T}^d}{
	\left(\mathbf{w}(\bk)+\mathbf{w}(-\bk)\right)
	\mathbf{w}(\bk)^T f_{X, \bdelta}(\bk)^2d\bk
      }\\
      &+ 
      (2\pi)^d|\bn|^{-1}\int_{\mathcal{T}^d}{
	\int_{\mathcal{T}^d}{
	  \mathbf{w}(\bomega_1)\mathbf{w}(\bomega_2)^T
	  f_{X,4,\delta}(\bomega_1, \bomega_2, -\bomega_1)
	  d\bomega_1
	}
	d\bomega_2,
      }
    \end{align*}
    where $f_{X,4,\delta}(\cdot, \cdot, \cdot)$ is the fourth-order
    cumulant spectral density, i.e.,
    \begin{equation*}
    	f_{X,4,\delta}(\bk_1, \bk_2, \bk_3)
    	=
    	\sum_{\bu_1,\bu_2,\bu_3}
    	{
    		c_4(\bu_1, \bu_2, \bu_3)
    		e^{-i(\bu_1\cdot\bk_1
    		+ \bu_2\cdot\bk_2
    		+\bu_2\cdot\bk_2	
    		)
    		}
    	},
    \end{equation*}
    and where $\mathbf{w}(-\bk)$ is obtained by $2\pi$ periodic extension of 
    $\mathbf{w}$ along all dimensions.
  \end{proposition}
  \begin{proof}
    Please see the Supplementary Material.
  \end{proof}
  Proposition~\ref{ass:asymptoticNormality2} is similar to
  Proposition~\ref{prop=normalityBoundedLinearCombinationsPeriodogram}.  The
  two differ in terms of the assumptions required to prove the result.
  Proposition~\ref{prop=normalityBoundedLinearCombinationsPeriodogram}
  requires the random field to be Gaussian while
  Proposition~\ref{ass:asymptoticNormality2} allows for non-Gaussian random
  fields at the expense of additional constraints on the memory of
  the random field.
\begin{theorem}
  \label{prop:asympvarianceest}
  Suppose Assumptions~\ref{ass:1},~\ref{assumption:rate},
  and~\ref{ass:asymptoticNormality2} hold. Then $\widehat{\btheta}$ is
  asymptotically normally distributed. Additionally, if the observed random
  field is Gaussian and the observation domain grows to infinity in all
  directions, $\widehat{\btheta}$ admits an asymptotic covariance structure
  determined by,
  \begin{equation*}
    2^{d+1}\pi^d|\bn|^{-1}
    \left[
      \int_{\mathcal{T}^d}{
	\nabla_{\btheta}\log f_{X,\bdelta}(\bk;\btheta)
	\nabla_{\btheta}\log f_{X,\bdelta}(\bk;\btheta)^T
	d\bk
      }
      \right]^{-1}.
  \end{equation*}
\end{theorem}
\begin{proof}
This results from combining Proposition~\ref{prop:asympNormalNonGaussian}
and the proof of Theorem~\ref{th=asymptNormality}.
\end{proof}
The asymptotic form of the covariance structure can also be determined for
the non-Gaussian case from Proposition~\ref{prop:asympNormalNonGaussian}.
Theorem~\ref{prop:asympvarianceest} is a generalization of a standard result
in time series analysis~\citep[Thm 10.8.2]{brockwell2009time}.  However, see
e.g.~\cite{simons2013maximum} for a practical large-sample example where the
asymptotic form has not been reached, but is instead dependent on the true
form of the expected periodogram as well as the sample size. This---in
addition to scenarios of incomplete grids---motivates the following section,
where we consider estimation of standard errors in the more general setting
where our asymptotic results do not hold.

\subsection{Estimating standard errors}

We now seek to derive how to estimate the standard error of
$\widehat{\btheta}$ for a given spatial sampling and model family.  Using
equations~\eqref{eq:21001405} and~\eqref{eq:12321806} from the Supplementary
Material, we obtain an approximation for the variance of $\widehat{\btheta}$
in the following proposition, where
$\mathcal{H}$ denotes the Fisher Information
matrix.
\begin{proposition}[Form of the variance]\label{propcov}
The covariance matrix of the quasi-likelihood estimator takes the form of
\begin{equation}\label{varhat}
  \var\left\{\widehat{\btheta} \right\}\approx {\bm{\mathcal{H}}}^{-1}(\btheta)
  \var\left\{\nabla\ell_M(\btheta)\right\}{\bm{\mathcal{H}}}^{-1}(\btheta),
\end{equation}
with the covariance matrix of the score taking the form of
\begin{align}\label{eq:stdformula1p}
  \cov\left\{
  \frac{\partial\ell_M(\bm{\theta})}{\partial\theta_p},
  \frac{\partial \ell_M(\bm{\theta})}{\partial\theta_q}
  \right\}&=\n^{-2}\sum_{\bk_1, \bk_2\in\Omega_{\bn}}
  \frac{\cov\left\{I_{\bn}(\bk_1),I_{\bn}(\bk_2)\right\}}
  {\overline{I}_{\bn}^2(\bk_1;\btheta)
  	\overline{I}_{\bn}^2(\bk_2;\btheta)}
  \frac{\partial\overline{I}_{\bn}(\bk_1;\btheta)}{\partial\theta_p}
  \frac{\partial\overline{I}_{\bn}(\bk_2;\btheta)}{\partial\theta_q}. 
\end{align}
\end{proposition}
The computation that appears in~\eqref{eq:stdformula1p} scales like
$|\mathbf{n}|^2$, i.e., not well for large grid sizes. We instead propose a
Monte Carlo implementation to speed this up. The dominant terms
in~\eqref{eq:stdformula1p} correspond to $\bomega_1 = \bomega_2$. We
approximate the sum over the rest of the terms, in the form
\begin{align*}
&\cov\left\{
  \frac{\partial\ell_M(\bm{\theta})}{\partial\theta_p},\frac{\partial \ell_M(\bm{\theta})}
       {\partial\theta_q}\right\}
  = |\bn|^{-2}\sum_{\bk_1\in\Omega_{\bn}}\left\{
  \frac{\partial\overline{I}_{\bn}(\bk_1;\btheta)}{\partial\theta_p}
  \frac{\var\left\{I_{\bn}(\bk_1)\right\}}{\overline{I}_{\bn}^4(\bk_1;\btheta)}
  \frac{\partial\overline{I}_{\bn}(\bk_1;\btheta)}{\partial\theta_q}\right\}\\ 
  &\quad \quad \quad \quad \quad \quad+ \frac{|\bn|^2-|\bn|}{M\n^2}
  \sum_{i = 1\ldots M}
  \frac{\partial\overline{I}_{\bn}(\bk_{1,i};\btheta)}{\partial\theta_p}
  \frac{\cov\left\{I_{\bn}(\bk_{1,i}),I_{\bn}(\bk_{2,i})\right\}}
{\overline{I}_{\bn}^2(\bk_{1,i};\btheta)\overline{I}_{\bn}^2(\bk_{2,i};\btheta)}  
\frac{\partial\overline{I}_{\bn}(\bk_{2,i};\btheta)}{\partial\theta_q},  
\end{align*}
where the $\bomega_{1,i}, \bomega_{2,i}, i = 1\ldots M$ are uniformly
and independently sampled from the set of Fourier frequencies $\Omega_\bn$ under
the requirement $\bomega_{1,i} \neq \bomega_{2,i}$. Note that if
tapering is used, one should consider a few coefficients near the main
diagonal in the above approximation, as tapering generates strong
short-range correlation in the frequency domain.

The covariances of the periodogram at two distinct Fourier frequencies
can be approximated by Riemann approximation of the two integrals that
appear in the expression below, before taking squared absolute values
and summing,
\begin{align*}
  \nonumber
  \cov\left\{I_{\bn}(\bk_{1,i}),I_{\bn}(\bk_{2,i})\right\}
  &=|\bn|^{-1}\left(\left|\int_{\mathcal{T}^d}{\widetilde{f}(\blambda)
    \mathcal{D}_{\bn}(\blambda-\bk_{1,i})
    \mathcal{D}^*_{\bn}(\blambda-\bk_{2,i})
    \,d\blambda}\right|^2\right.\\ 
  &+\left.\left|\int_{\mathcal{T}^d}{\widetilde{f}(\blambda)\mathcal{D}_{\bn}(\blambda-\bk_{1,i})
    \mathcal{D}^*_{\bn}(\blambda+\bk_{2,i})d\blambda}\right|^2\right), 
  \quad i=1,\ldots,M. 
\end{align*}
In the above, $\widetilde{f}$ is the following approximation to the
spectral density, which can be computed by a DFT,
\begin{equation*}
\widetilde{f}(\blambda) =
\sum_{\bu\in\prod_{i=1}^d{[-(n_i-1)\ldots(n_i-1)]}}{c_X(\bu;\btheta)\exp(-i\blambda\cdot\bu)}, 
\end{equation*}
and $\mathcal{D}_{\bn}(\lambda)$ is the non-centred \emph{modified} (due to
the modulation $g_\bs$) Dirichlet kernel of order~$\bn$ given by
\begin{equation*}
\mathcal{D}_{\bn}(\blambda) = \sum_{\bs\in\Jx}{g_\bs \exp(i\blambda\cdot\bs)},
\end{equation*}
where for clarity we omit the dependence on the modulation $g_\bs$ in the notation.
Finally we compute the derivatives of
$\overline{I}_{\bn}(\bk;\btheta)$ as follows,
\begin{equation}
\nabla_{\btheta} \overline{I}_{\bn}(\bk;\btheta) =
\sum_{u\in\Z^d}{\nabla_{\btheta}\overline{c}_X(\bu;\btheta)\exp(-i\bk\cdot\bu)}. 
\end{equation}

\section{Simulation studies and application to the study of planetary topography}
\label{sec:sim}

In this section we present simulation studies and an application to the
study of Venus' topography that demonstrate the performance of the Debiased
Spatial Whittle estimator. We also refer the reader to the Supplementary
Material which contains additional simulation studies. The simulations
presented in Section~\ref{sec:sim1} address the estimation of the range
parameter of a Mat\'ern process, whose slope parameter is known, observed
over a full rectangular grid. These simulations corroborate our theoretical
results on the optimal convergence rate of our estimator despite edge
effects, in contrast to the standard Whittle method. Our second simulation
study in Section~\ref{sec:missing} shows how our estimation procedure
extends the computational benefits of frequency-domain methods to
non-rectangular shapes of data, where we compare parameter estimates with
those of~\cite{guinness2017circulant} in the scenario of a circular shape of
observations. In Section~\ref{sec:strongmissing} we estimate the parameters
of a simulated Mat\'ern process sampled according to a real-world sampling
scheme of terrestrial ocean-floor topography \cite[]{GEBCO2019} with
approximately $72\%$ missing data. Finally, in Section~\ref{sec:venus} we
demonstrate the performance of the Debiased Spatial Whittle estimator when
applied to topographical datasets obtained from Venus \cite[]{Rappaport+99}.

\subsection{Estimation from a fully-observed rectangular grid of data}
\label{sec:sim1}

We simulate from the isotropic Mat\'ern model family, which
corresponds to the following covariance function,
\begin{equation}
\label{eq:materncovfunc}
c_X({\bu})=\sigma^2 \frac{2^{1-\nu}}{\Gamma(\nu)}
\left(\sqrt{2\nu}\frac{\|{\bu}\|}{\rho}\right)^{\nu} 
K_{\nu}\left(\sqrt{2\nu}\frac{\|{\bu}\|}{\rho}\right),
\end{equation}
where $K_{\nu}(x)$ is a Bessel function of the second kind. We
consider the problem of estimating the range parameter $\rho$, which
is fixed to 10 units, while the amplitude $\sigma^2=1$ and the slope
parameter $\nu\in\{\frac{1}{2}, \frac{3}{2}\}$ are fixed and
known. Inference is achieved from simulated data on two-dimensional
rectangular grids of increasing sizes, specifically $\{2^s: s = 4,
\cdots, 8\}$ in each dimension.
\begin{figure}[h!]
\centering
\includegraphics[width=0.8\textwidth]{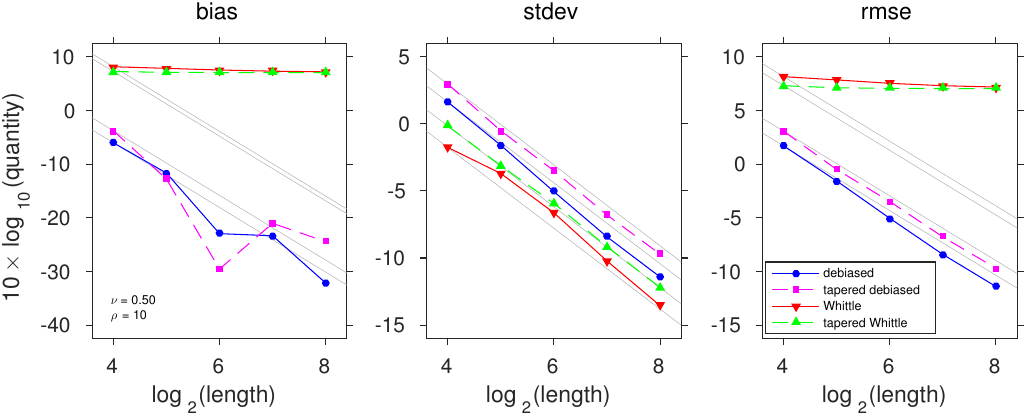}
\caption{\label{fig:rates1}Bias, standard deviation, and root
  mean-squared error of estimates of the range parameter $\rho = 10$
  of a Mat\'ern process~\eqref{eq:materncovfunc} with
  $\nu=1/2,\sigma^2=1$. The estimation method is identified by
  the line style, and gray lines functionally express the theoretical
  dependence on the square root of the sample size. The side length of
  the two-dimensional square grid is indicated by the horizontal axis,
  leading to a sample size of the length squared.}
\vspace{1cm}
\includegraphics[width=0.8\textwidth]{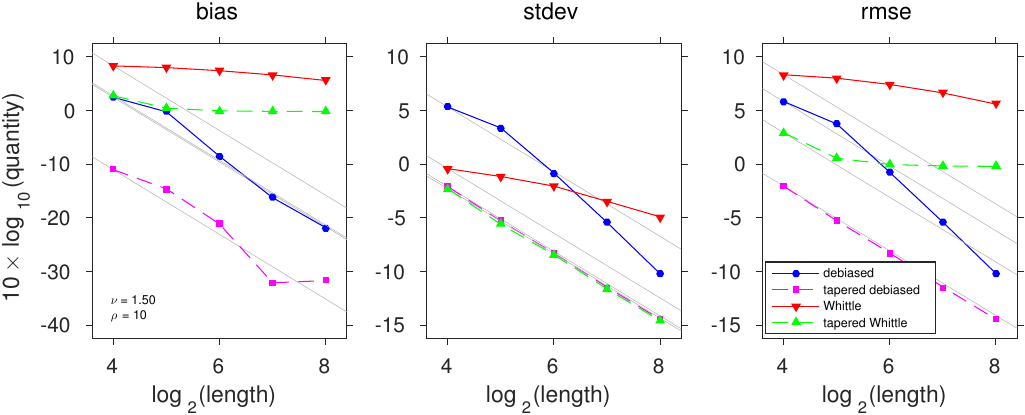}
\caption{\label{fig:rates2}The same simulation setup as in
  Figure~\ref{fig:rates1}, but with $\nu=3/2$. This higher slope parameter
  is associated with smoother realizations, resulting in worsened edge
  effects. This illustrates how our method effectively addresses the edge
  effect issues even in that setting.}
\end{figure}
\begin{figure}[h!]
\centering
\includegraphics[width=0.65\textwidth]{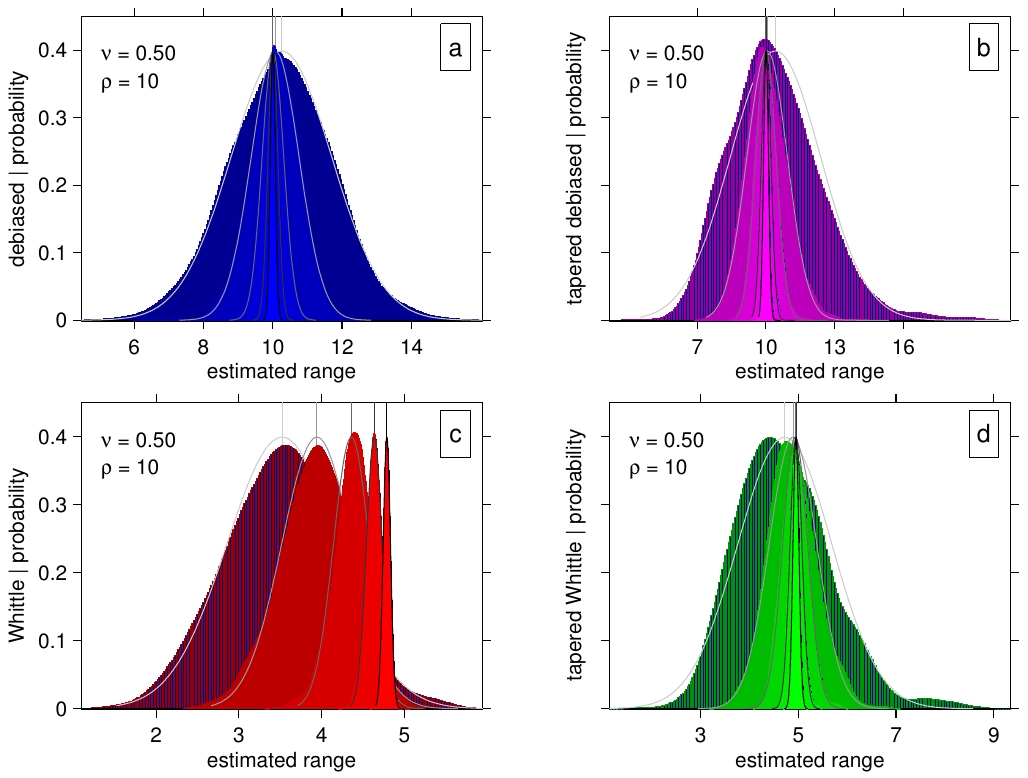}
\caption{\label{fig:histaW}Nonparametric density estimates $\hat{\rho}$ of
  the estimated range parameter $\hat{\rho}$ ($\rho = 10$) for a Mat\'ern
  random field~\eqref{eq:materncovfunc}, with $\sigma^2=1$ and
  $\nu=1/2$. The four subplots show different estimation methods of
  (a)~Debiased Spatial Whittle, (b)~Debiased Spatial Whittle with tapering,
  (c)~standard Whittle, and (d)~standard Whittle with tapering. The density
  estimate is shaded to reflect the size of the random field, with the
  darkest corresponding to total observations $|\mathbf{n}|=(2^4)^2$, and
  the shading incrementally taking a lighter colour for
  $|\mathbf{n}|=(2^5)^2, (2^6)^2, (2^7)^2,(2^8)^2$. Each density estimate is
  complemented by the best fitting Gaussian approximation as a solid black
  or fading gray line (black corresponds to $|\mathbf{n}|=(2^8)^2$ and the
  lightest gray to $|\mathbf{n}|=(2^4)^2$.)}
\end{figure} 
We implement four inference methods: 
\begin{enumerate}[label={(M\arabic*)}]
\item\label{DW} The Debiased Spatial Whittle method, i.e., the estimate
  derived from~\eqref{eq:DebiasedWhittleLKH};
\item\label{DWT} The Debiased Spatial Whittle method combined with a taper,
  specifically the estimate derived from~\eqref{eq:DebiasedWhittleLKH} with
  $g_\bs$ proportional to a Hanning taper;
\item\label{SWL} The standard Whittle likelihood, i.e., estimators obtained by
  replacing $\overline{I}_{\bn}(\bk;\btheta)$ with $f_X(\bk)$
  in~\eqref{eq:pseudolkh} and then
  minimizing~\eqref{eq:DebiasedWhittleLKH};
\item\label{SWLT} The standard Whittle likelihood combined with tapering using a
  Hanning taper, again derived from~\eqref{eq:DebiasedWhittleLKH}
  fitting to $f_X(\bk)$. 
\end{enumerate}
For each configuration of the slope parameter and grid size, we report
summary statistics corresponding to 1,000 independently realised
random fields. We report bias, standard deviation and root
mean-squared error for $\nu=1/2$ and $\nu=3/2$ in
Figures~\ref{fig:rates1} and~\ref{fig:rates2}, respectively.

We first observe that the rate of the Whittle likelihood~\ref{SWL} is very
poor, due to its large bias. It appears that tapering~\ref{SWLT} leads to
improved convergence rates when $\nu=3/2$, although bias remains. In
contrast, the rates of our proposed method~\ref{DW} and its tapered
version~\ref{DWT} do not curb down even with larger grid sizes. This concurs
with the theoretical results on the rate of convergence provided in
Section~\ref{sec:theory}. This example demonstrates that the Debiased
Spatial Whittle method balances the need for computational and statistical
efficiency with large data sets.

In Figure~\ref{fig:histaW} we report the empirical distribution of each
estimator obtained from the 1,000 independent inference procedures for
$\nu=1/2$. The four panels (a), (b), (c) and (d) show the distribution of
estimates from the four methods. The first two panels, (a) and (b), are
broadly unbiased with estimates centred on $\rho=10$ that converge
quickly. The standard Whittle method~(c) has issues with underestimation,
tending towards $\rho=5$. This asymptotic bias is in large part due to
aliasing not being accounted for, combined with the relatively small value
of $\nu=1/2$; these effects are still present in the tapered
estimates~(d). As would be expected, in all four subplots the variance is
decreasing with increasing sample size, at similar rates.  In the
Supplementary Material we present the same study where the Whittle and
tapered Whittle methods use an aliased version of the spectral density. This
largely reduces the bias of these methods. However some asymptotic bias
remains, even for the tapered Whittle method, due to our fixed approximation
to the aliased spectral density owing to computational constraints.

\subsection{Estimation from a circular set of observations}
\label{sec:missing}

In this section, we show how our Debiased Spatial Whittle method extends to
non-rectangular data. More specifically, we assume we only observe data
within a circle with diameter 97 units.  We consider the exponential
covariance kernel given by
\begin{equation}\label{eq:235302072019}
c_X(\bu) = \sigma^2\exp\left(-\frac{\|\bu\|}{\rho}\right), 
\qquad \bu\in\R^2,
\end{equation}
where $\sigma^2 = 1$ is fixed and known and we estimate the range
parameter $\rho$ whose true value is set to 5 units.
We note that the case of a growing circle satisfies SCC, according to
Lemma~\ref{lemma:scc2}, and hence leads to consistency of our estimator.  We
also expect optimal convergence rates, see Theorem~\ref{th=asymptNormality}.

A total number of 1,200 independent simulations are performed. As a
state-of-the-art baseline, we compare to a recent method proposed
by~\citet{guinness2017circulant}, which is an approximation of the circulant
embedding method developed by~\citet{stroud2017bayesian}. These authors
proposed an Expectation Maximization iterative procedure, where the observed
sample is embedded onto a larger grid that makes the covariance matrix
\emph{Block Circulant with Circulant Blocks} (BCCB), which can be
diagonalised fast through the FFT algorithm. \citet{guinness2017circulant}
point out that the size of the embedding grid is very large, making the
imputations costly and the convergence over the iterations slow. To address
this limitation they propose using a periodic approximation of the
covariance function on an embedding grid which is much smaller than that
required for the exact procedure. They show via simulations that using an
embedding grid ratio of $1.25$ along each axis leads to good approximations
of the covariance function on the observed grid.

To implement the method developed by~\citet{guinness2017circulant}, we use
the code provided by the authors. We set a grid ratio of $1.25$ to limit the
computational cost, and implement the method with two choices of the number
of imputations per iteration, $M=1$ and $M=20$. Each implementation is run
for a number of $30$ iterations for all samples.

Both our estimation method and that of~\citet{guinness2017circulant} are
initialised with the estimates provided by the method proposed
by~\citet{fuentes2007approximate}. We show in
Figure~\ref{fig:incompletetradeoff} (left) how the Debiased Spatial Whittle
method achieves computational and statistical efficiency. The 95 per cent
confidence interval of our estimate is similar to that obtained via the
method of~\citet{guinness2017circulant} ($M=1$), however our method, despite
also using an iterative maximization procedure, is significantly faster. As
shown in Figure~\ref{fig:incompletetradeoff} (right panel),
\citet{guinness2017circulant} ($M=20$) leads to lower root mean-squared
error but requires more computational time.

\begin{figure}[h]
\centering{\includegraphics[width=0.85\textwidth]{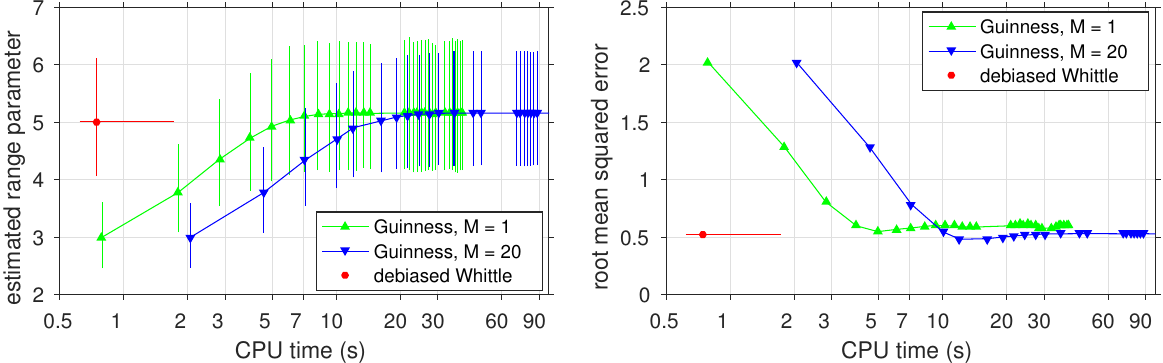}}
\caption{\label{fig:incompletetradeoff} Mean and 95 per cent confidence
  intervals (left) and root mean-squared error (right) of estimates of the
  range parameter $\rho=5$ of an exponential covariance
  model~\eqref{eq:235302072019}. Estimation is performed on a circular set
  of data with diameter 97 units. The converged estimates of the Debiased
  Spatial Whittle method are compared to the iterated estimates of two
  implementations of~\citet{guinness2017circulant}. The horizontal axis in
  both panels corresponds to the average computational time, as performed on
  an Intel(R) Core(TM) i7-7500U CPU 2.7--2.9~GHz processor.}
\end{figure}

\subsection{Application to a realistic sampling scheme of ocean-floor topography}
\label{sec:strongmissing}

In this simulation study we show that our estimator can address complex
lower-dimensional sampling substructure. We apply it to the estimation of a
Matérn process sampled on a real-world observation grid of ocean-bathymetry
soundings, characterised by a very large amount of missing data ($\approx
72\%$). We simulate two Mat\'ern processes, each with slope parameter $0.5$
and with range $20$ and $50$ units respectively. The initial grid is of size
$1081\times1081$. We select a subgrid of size $256\times256$ with similar
missingness properties to those of the whole grid. In
Figure~\ref{fig:frederik_grid1} we plot (left) a simulated Mat\'ern process
on that grid where missing observations have been replaced with zeros. We
note the large amount of missing observations within the bounding
rectangular grid, as well as its complex patterns (i.e. rather than a
uniform missingness scheme). For both these reasons the method proposed
by~\citet{fuentes2007approximate} fails, while our method is still able to
produce useful estimates, as shown in the right panel of
Figure~\ref{fig:frederik_grid1}.
\begin{figure}[h!]
  \centering
  \includegraphics[width=0.30\textwidth]{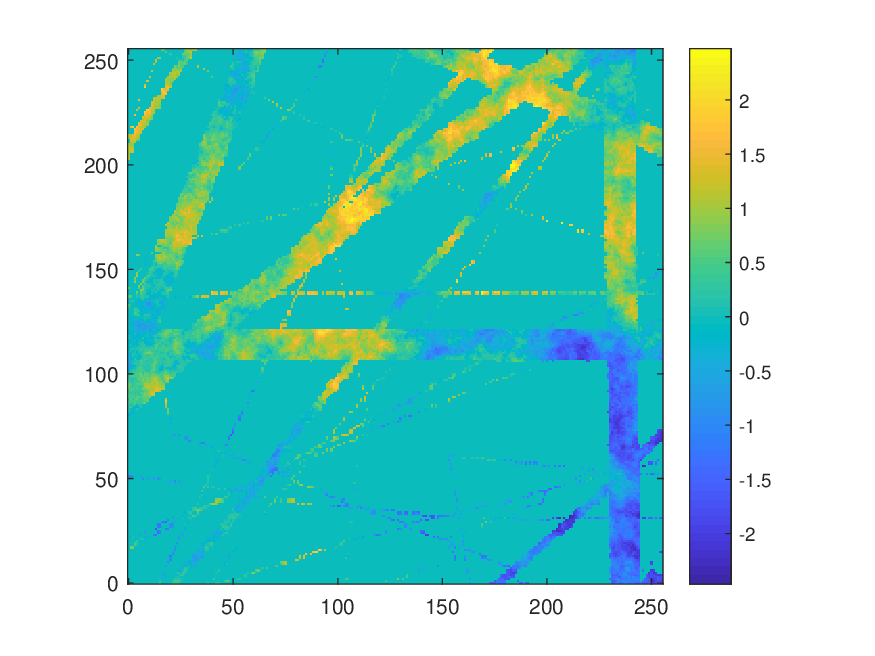}
  \includegraphics[width=0.59\textwidth]{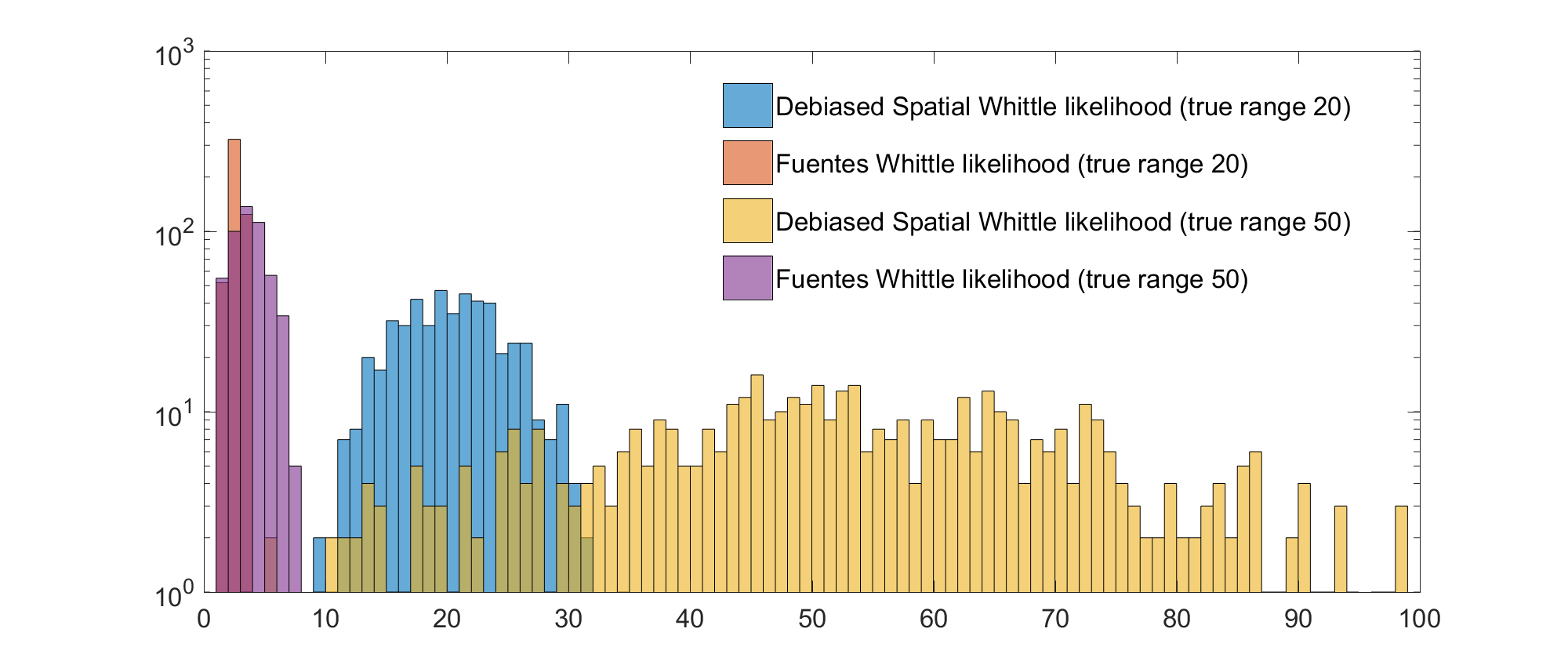}
  \caption{(Left) Simulated Mat\'ern process with slope parameter $0.5$ and
    range parameter $50$ units, on a real-world sampling grid, with missing
    observations replaced by zeros. (Right) Histogram of estimates of the
    range parameter of a simulated Mat\'ern process observed on the
    real-world grid shown in the left panel. We compare our proposed
    estimation method, the Debiased Spatial Whittle likelihood, to the
    method proposed by~\citet{fuentes2007approximate}. The true value of the
    range is fixed to 20 or 50. Despite an increased variance due to the
    complex missing data patterns, our method is still able to produce a
    useful estimate of the range parameter, in comparison to the estimates
    produced by the method proposed by~\citet{fuentes2007approximate}, which
    was not built to address such large and complex patterns of missing
    data.}
    \label{fig:frederik_grid1}
\end{figure}

\subsection{Application to the study of Venus' topography}
\label{sec:venus}

In this section we apply our Debiased Spatial Whittle method to the study of
Venus' topography. The motivation for modelling a planet's topography using
a parametric covariance model such as the Mat\'ern process is
multifaceted. For instance, we may expect that the combination of the slope
and range parameters will carry important information about the
geomorphological process or age of formation of the observed topography,
i.e., it is expected that those parameters will have an interpretable
physical meaning. The slope parameter can be related to the smoothness of
the topography, and the range parameter tells about the typical distance
over which two observed portions are uncorrelated.

\begin{table}
  \caption{\label{table:maternestimates}Estimates of the three parameters of a Mat\'ern
    process, see~\eqref{eq:materncovfunc}.
  }
  \centering
  \begin{tabular}{l|lll|lll|lll|lll}
    & \multicolumn{3}{c|}{\textbf{Patch~1}} & \multicolumn{3}{c|}{\textbf{Patch~2}} & \multicolumn{3}{c|}{\textbf{Patch~3}} & \multicolumn{3}{c}{\textbf{Patch~4}}\\ \hline
    Parameter: & $\sigma$     & $\nu$     & $\rho$     & $\sigma$     & $\nu$     & $\rho$     & $\sigma$     & $\nu$     & $\rho$     & $\sigma$     & $\nu$     & $\rho$\\ \hline
    \textbf{Debiased Spatial Whittle} & 1.2          & 0.5       & 17.7       & 1.2          & 0.7       & 6.8        & 2.1          & 0.5       & 36.5       & 1.5          & 0.6       & 15.0  \\
    \textbf{Standard Whittle} & 1.6          & 0.3       & 62.7       & 1.8          & 0.3       & 73.9       & 1.5          & 0.2       & 77.3       & 1.7          & 0.3       & 87.3  \\
    \textbf{Tapered Whittle}  & 2.0          & 0.4       & 52.0       & 1.7          & 0.2       & 80.6       & 1.2          & 0.2       & 88.1       & 1.9          & 0.4       & 83.7     
  \end{tabular}
\end{table}

Building on the work of \cite{Eggers2013}, we have selected four patches of
data (including that shown in Figure~\ref{fig:realizations} which
corresponds to Patch~3), each sampled regularly on a complete rectangular
grid. We compare three estimation procedures: the Debiased Spatial Whittle
method, the standard Whittle method, and the standard Whittle method with
tapering (again using a Hanning taper). Parameter estimates are reported in
Table~\ref{table:maternestimates}. We also compare the value of the exact
likelihood function taken at the estimated parameters for each estimation
method in Table~\ref{table:exactlkhvalues}. Specifically, if
$\hat{\btheta}_M$ and $\hat{\btheta}_W$ respectively denote the estimates
obtained via the Debiased Spatial Whittle and standard Whittle procedure, we
compare $l_E(\hat{\btheta}_M)$ and $l_E(\hat{\btheta}_W)$, with $l_E(\cdot)$
denoting the exact likelihood function (which is expensive to evaluate but
only needs to be done once for each analyzed method). The results in
Table~\ref{table:exactlkhvalues} show a much better fit of the model
corresponding to the parameters estimated via the Debiased Spatial Whittle
method, in comparison to the parameters estimated via either standard
Whittle or tapered Whittle. The parameter estimates in
Table~\ref{table:maternestimates} should be interpreted with care due to the
challenges inherent in joint estimation of all three parameters of a
Mat\'ern covariance function \cite[see,
  e.g.,][]{zhang2004inconsistent}. However in all four patches we observe
that the standard and tapered Whittle likelihood appear to overestimate the
range while underestimating the smoothness, consistent with results found
by~\cite{sykulski2019debiased} for oceanographic time series.

\begin{figure}[h!]
\centering
\includegraphics[width=0.70\columnwidth]{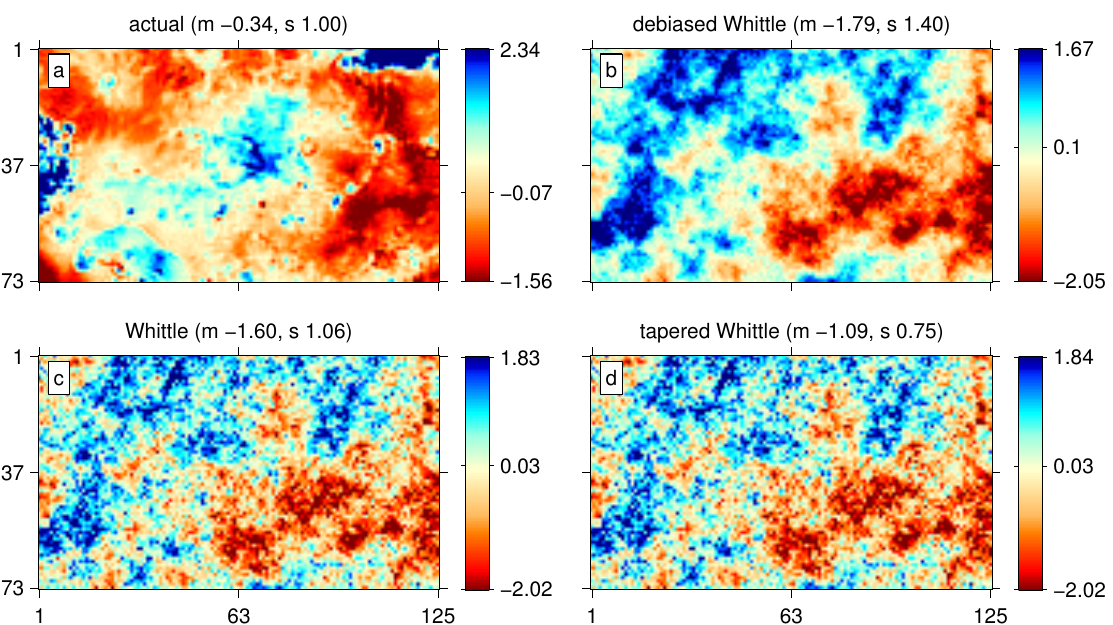}
\vspace{-3mm}
\caption{\label{fig:realizations} (a)~A realised random field from the
  topography of Venus; and simulated random fields from a Mat\'ern model
  with parameters estimated using (b)~Debiased Spatial Whittle estimation,
  (c)~standard Whittle estimation, and (d)~standard Whittle estimation using
  a Hanning taper. Simulated random fields were obtained using the same
  random seed to facilitate comparison.  Parameter values for each method
  are given in Table~\ref{table:maternestimates} (Patch~3) in
  Section~\ref{sec:venus}. Sample means (m) and standard deviations (s) are
  in the titles. Colour bars are marked at the 2.5th, 50th and 97.5th
  quantiles. Axis labels are in pixels.}
\end{figure}

Finally, Figure~\ref{fig:realizations} presents a
comparison of Patch~3 with three simulated samples, obtained using the
Mat\'ern model estimated using the Debiased Spatial Whittle, standard and
tapered Whittle methods, respectively. This analysis supports the conclusion
that the Debiased Spatial Whittle method is able to find more appropriate
parameter values for the model fit.

\begin{table}
\caption{\label{table:exactlkhvalues}Percentage of increase in the exact
  likelihood value at the estimated parameter values from
  Table~\ref{table:maternestimates} in comparison to the minimal value
  obtained among the three methods.}  \centering
\begin{tabular}{l|l|l|l|l}
  & \textbf{Patch~1} & \textbf{Patch~2} & \textbf{Patch~3} & \textbf{Patch~4} \\ \hline
\textbf{Debiased Spatial Whittle}  & 60.60             & 104.80          & 91.60             & 48.40         \\
\textbf{Standard Whittle} & 0                 & 16.10             & 0                 & 0             \\
\textbf{Tapered Whittle}  & 23.20             & 0                 & 53.90             & 25.20
\end{tabular}
\end{table}

\section{Discussion}
\label{sec:discussion}

In this paper we addressed the estimation of parametric covariance models
for Gaussian and non-Gaussian random fields using the discrete Fourier
transform. Key to understanding a random field is its spatial sampling; this
can range from a spatial point process, to regular sampling with an
irregular boundary, to observations missing at random on a grid, to a fully
sampled square regular grid. To maintain computational feasibility, this
paper addresses the analysis of a regularly sampled random field, with
potentially missing observations and an irregular (not cuboid) sampling
domain.

The Whittle likelihood uses the Fast Fourier Transform to achieve
computational efficiency. The approximation is based on results for Block
Toeplitz with Toeplitz Blocks
matrices~\citep{Tyrtyshnikov1998,kazeev2013multilevel}, on (growing-domain)
asymptotics, and on arguments that equate the Gaussian non-diagonal
quadratic form with another Gaussian, nearly diagonal, form. For time series
this argument is relatively straightforward, but is somewhat more complex
for spatial data in higher dimensions, where the bias becomes the dominant
term \citep{Guyon2008}, and the geometry of the sampling process leaves a
strong imprint.

The bias of the periodogram as an estimator of the spectral density (which
drives subsequent bias) decreases with rate
$\mathcal{O}\left(\n^{-1/d}\right)$~\citep{Guyon2008,dahlhaus1987edge} in
the ideal case of a fully-observed rectangular lattice in $d$~dimensions
that grows at the same rate along all directions.
\citet{dahlhaus1983spectral} proposed tapering to remedy this issue.  A more
general result by~\citet{Kent1996} shows that the approximation resulting
from replacing the exact likelihood with the Whittle likelihood in the case
of a full grid is driven by the size of the smallest side of the rectangular
lattice.  Tapering on its own cannot solve this issue.

To address bias in a general setting we proposed replacing the spectral
density by the true expectation of the periodogram.  From the notion of
Significant Correlation Contribution, we can understand the technical
underpinning of this bias removal process and draw a general framework of
sampling schemes and model families for which our estimator is statistically
efficient.

In addition, our Debiased Whittle procedure also explicitly
  accounts for aliasing in the computation of the expected periodogram, thus
  avoiding computationally--expensive wrapping operations to fold in higher
  unobserved frequencies into the likelihood. As would be expected, in
  simulations we found the bias correction from aliasing to be most
  important when the rate of decay in the spectral density in frequency is
  slow (e.g. a Mat\'ern process with small slope parameter). In contrast, we
  found that accounting for finite sampling and boundary effects to be most
  important when the rate of decay is high and the spectrum therefore has a
  large dynamic range (e.g. a Mat\'ern process with large slope
  parameter). Overall, our explicit handling for the effects of missing data
  provided further improvements for all processes studied, regardless of the
  specific form of the spectral density.

For random fields with missing observations, \citet{fuentes2007approximate}
suggested to replace the missing points of a rectangular lattice with zeros,
as we do in~\eqref{eq:definitionPeriodogram}, and correcting uniformly
across frequencies for the amplitude of the periodogram, based on the ratio
of the number of observed points to the total number of points in the
grid. This only partly corrects for the bias of the periodogram that results
from any non-trivial shape of the data, as frequencies are likely to not be
affected uniformly by the sampling scheme; in contrast to our estimation
procedure which directly encodes the observed data, and the observed
missingness pattern. Under relatively weak assumptions, and
  through the notion of Significant Correlation Contribution, we establish
  consistency and asymptotic normality in both Gaussian and non-Gaussian
  settings.

When studying non-Gaussian observations one can take two approaches; either
limiting the effects of the non-Gaussianity on the variance of the
estimator~\citep{giraitis1999whittle,sykulski2019debiased}, or even
permitting Whittle-type estimation based on higher order spectral moments,
see e.g.~\cite{anh2007minimum}. If infill asymptotics are
considered~\citep{bandyopadhyay2009asymptotic}, then the limiting
distribution of the Fourier transform need not be Gaussian. Note that the
aforementioned authors assumed completely random sampling of the fields,
which we do not, as such sampling leads to a ``nugget-effect'' at frequency
zero and beyond.


To treat more general multivariate processes, we defined a multivariate
sampling mechanism that is initially on the same grid, but where the
missingness pattern may be different between processes. To be able to arrive
at consistent estimators, we again use a version of the concept of
Significant Correlation Contribution, but now adapted to the multivariate
nature of the data. Under this assumption, which ensures we gain more
information as our sampling scheme diverges in cardinality, we do achieve
estimation consistency.

\citet{stroud2017bayesian} have proposed an approach that does not require
approximating the multi-level Toeplitz covariance matrix of the rectangular
lattice sample by a multi-level circulant matrix. Instead, their method
finds a larger lattice, termed an embedding, such that there exists a Block
Circulant with Circulant Blocks (BCCB) matrix that is the covariance matrix
of a Gaussian process on this extended lattice, and such that the covariance
matrix of the real process is a submatrix of this extended matrix. One can
then simulate efficiently the missing data on the extended lattice, and
estimate the parameters of the models. This process can be iterated until a
convergence criterion is met. This elegant method still suffers from
computational issues, as the size of the embedding might be quite large. A
solution suggested by~\citet{guinness2017circulant} is to use a circulant
approximation of the covariance on a smaller rectangular lattice. In that
case, the method is no longer exact, but~\citet{guinness2017circulant}
showed via simulations that using small embeddings can in some cases provide
a good compromise between statistical and computational efficiency.

In contrast, in this paper we revisited the root cause of why the
approximation of the likelihood may deteriorate, while continuing to require
that any proposed bias elimination result in a computationally competitive
method. Our method of bias elimination is ``built in'' by fitting the
periodogram to its expectation $\overline{I}_{\bn}(\bk;\btheta)$. This is in
contrast to estimating the bias and removing it, which typically increases
variance, and might lead to negative spectral density estimates.

We have thus proposed a bias elimination method that is data-driven, fully
automated, and computationally practical for a number of realistic spatial
sampling methods, in any dimension. Our methods are robust to huge volumes
of missing data, as backed up by our theoretical analysis, and evidenced by
our practical simulation examples. As a result, our methodology is not only
of great benefit for improved parameter estimation directly, but also has
knock-on benefits in, for example, the problem of prediction. Here a huge
number of methods exist and there is some debate as to which are most
practically useful~\citep{heaton2019case}. The broader point is that many of
these methods are based on Mat\'ern covariance kernels, and therefore our
methods, which we have shown greatly improve Mat\'ern parameter estimation,
can be naturally incorporated to improve the performance of such spatial
methods for prediction. Quantifying this benefit over a range of settings is
a natural line of further investigation.

Within parameter estimation, there are a number of large outstanding
challenges which are nontrivial extensions and merit further investigation
as stand-alone pieces of work: [1]~extensions to fully irregularly sampled
process on non-uniform grids; and [2]~extensions to multivariate processes
with complex sampling patterns. In each case the impact on the Fourier
Transform and the expected periodogram need to be carefully handled to
properly account for the bias of naively using basic Whittle-type
approximations.  We do, however, expect that large improvements are possible
both in terms of bias reduction (vs standard Whittle methods where edge
effect contamination will increase), and in terms of computational speed (vs
exact likelihood and other pseudo-likelihoods which will become increasingly
intractable as assumptions are relaxed).

\section{Acknowledgements}
\label{sec:acknowledgements}
The authors would like to thank the European Research Council under Grant
CoG 2015-682172NETS, within the Seventh European Union Framework Program.
We are extremely grateful to the anonymous referees and Associate Editor
for the Journal of the Royal Statistical Society (Series B) for
their constructive comments on the paper, which have ended up significantly
improving this manuscript. \\The code used in Section~\ref{sec:sim} is
available from https://github.com/arthurBarthe/JRSSB2022.
A python package called debiased-spatial-whittle is available
from Pypi.

\appendix

\makeatletter
\renewcommand{\@seccntformat}[1]{APPENDIX~{\csname the#1\endcsname}.\hspace*{1em}}
\makeatother

\bibliographystyle{rss}
\setlength{\bibsep}{3.2pt}
\bibliography{guillaumin}
\end{document}


\section*{Aliased Whittle likelihood comparison}
{\color{black} In this section, we provide simulation results in the same manner as
  those of Section~\ref{sec:sim1} in the main document, except for the fact
  that here both the Whittle and tapered Whittle estimator use a truncated
  approximation of the aliased spectral density of the sampled process, see 
	Figure~\ref{fig:aliased}. We
  limited the approximation to include the contribution of frequencies from
  $[-3\pi, 3\pi]^2$ to keep computational cost reasonable.
	The fact that we use a fixed approximation to the
  aliased spectral density explains why, despite largely reducing the bias
  for the Whittle and tapered Whittle, in comparison to the version in the
  main document, the efficiency of both the Whittle and tapered Whittle
  estimators appears to saturate for large grid sizes.
\begin{figure}[h!]
	\centering
		\includegraphics[width=\textwidth]{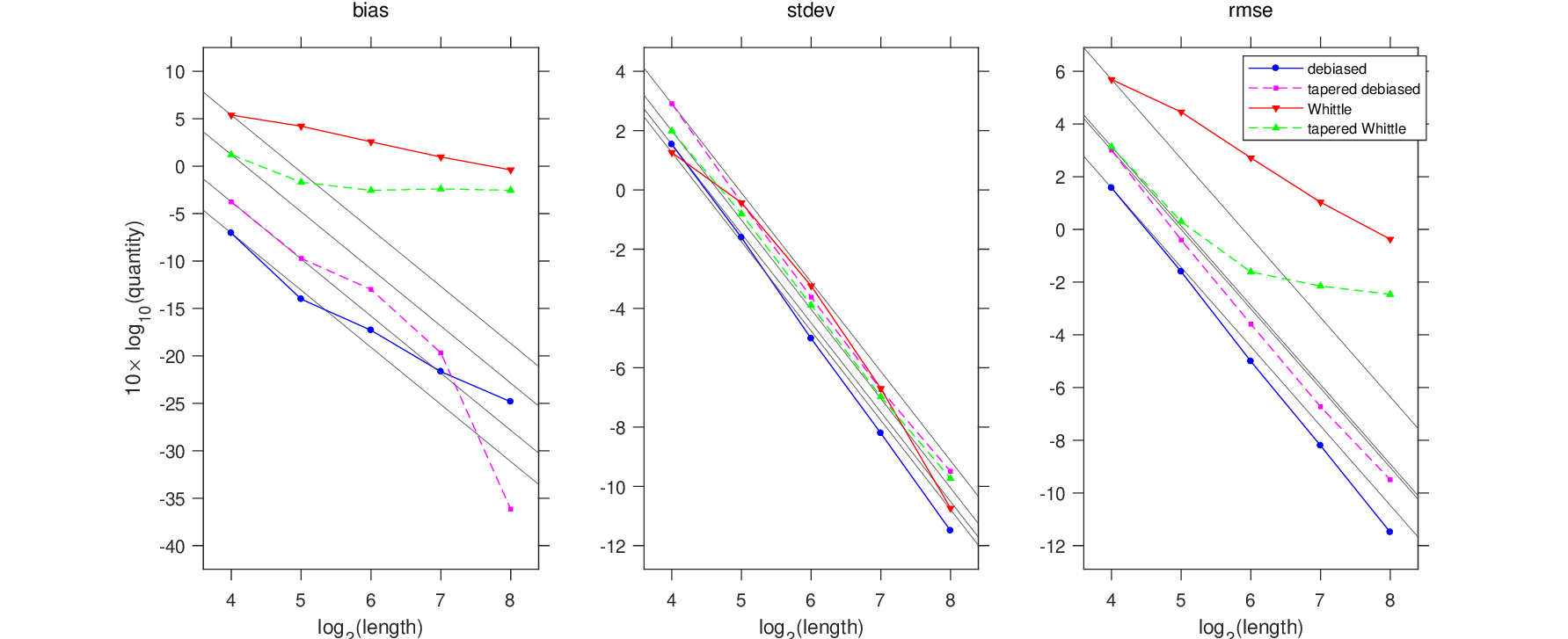}
	\caption{\label{fig:aliased}Bias, standard deviation, and root
  mean-squared error of estimates of the range parameter $\rho = 10$
  of a Mat\'ern process~\eqref{eq:materncovfunc} with
  $\nu=1/2,\sigma^2=1$. Compared to Figure~\ref{fig:rates1} in the
	main document, the Whittle and tapered Whittle estimation methods use an 
	approximation to the aliased spectral density function, by incorporating
	contributions from frequencies within the square domain $[-3\pi, 3\pi]^2$.}
\end{figure}
}
\newpage
\section*{Estimation for a discrete spatial model}
{\color{black} In this section we apply the Spatial Debiased Whittle to the
  estimation of a discrete parametric model. In comparison to continuous
  models, the estimation of the parameters of a discrete spatial model
	is not hindered by aliasing.  The model we consider is defined in the
  frequency domain according to,
\begin{equation}
\label{eq:discrete_spatial}
	f(\omega_1, \omega_2) = \left\{\begin{array}{lcr} \exp\left\{-\theta (|\omega_1| + |\omega_2|)\right\} & {\mathrm{if}} & \bm{\omega}\in(-\pi,\pi)^2\\
	0 & {\mathrm{o/w}} &
	\end{array}\right.,
\end{equation}
where $\theta\geq 0$.
The covariance function of this model is easily obtained analytically, 
and takes the form of,
\begin{equation}
	c_X(u_1, u_2) = 4 \mathcal{R}
	\left\{ 
	\frac{ 1 }{ i u_1 - \theta }
	\left(
	\exp\left[(i u_ 1 - \theta) \pi \right] - 1
	\right)
	\right\}
	\mathcal{R}
	\left\{ 
	\frac{ 1 }{ i u_2 - \theta }
	\left(
	\exp\left[(i u_ 2 - \theta) \pi \right] - 1
	\right)
	\right\},
\end{equation}
which is separable in \(u_1\) and \(u_2\), since the spectrum is
separable in \(\omega_1\) and \(\omega_2\). We display a simulated realization
in Figure~\ref{fig:sim_discrete}.
\begin{figure}
	\centering
		\includegraphics[width=0.8\textwidth]{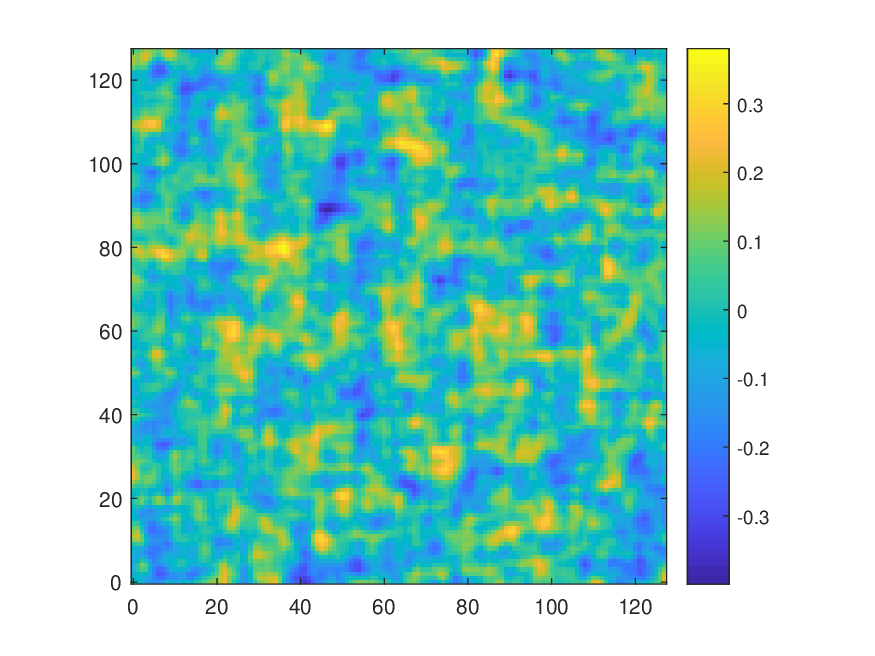}
		\caption{
		\label{fig:sim_discrete}
		Simulated sample from the discrete model defined by
		equation~\eqref{eq:discrete_spatial}, with $\theta=3$.}
\end{figure}

In our experiments we set $\theta = 3$ and initialize
estimates to $0.2$ for all estimation methods.  In a first experiment we
consider estimation on growing squares, see Figure~\ref{fig:fig_supp1}.
The tapered Whittle method performs very
well for this discrete model for large sizes, but suffers from bias for
smaller grid sizes. The tapered version of the Spatial Debiased Whittle
performs better than its non-tapered counter-part, due to remaining boundary
effects. However, it is notable that even without tapering, the Spatial
Debiased Whittle appears to perform at the expected square root $n$ rate.

\begin{figure}
\centering 
\includegraphics[width=\textwidth]{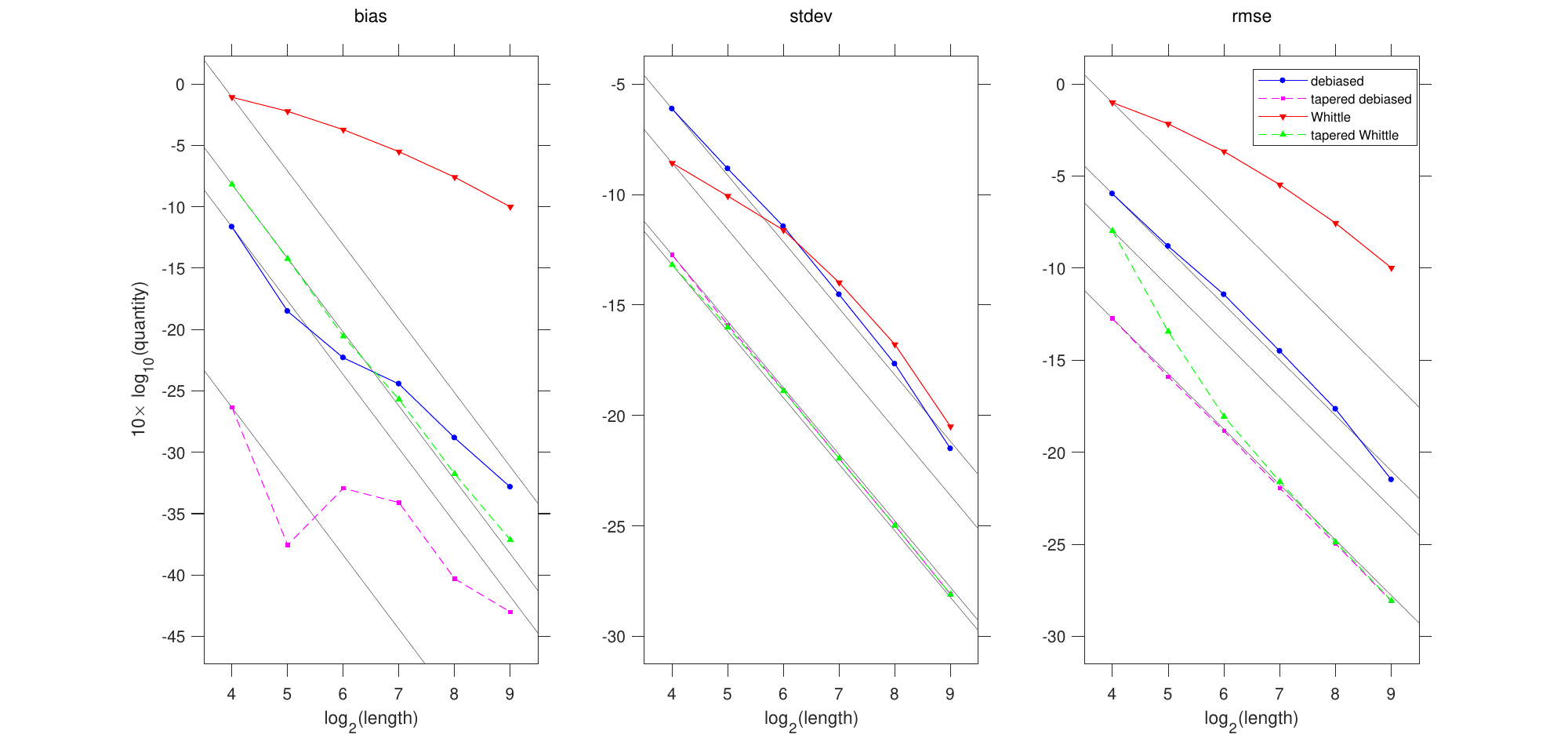}
\caption{Bias, standard deviation, and RMSE of estimates of $\theta=3$ for
  rectangular grids of size $N\times N$ where N increases in powers of $2$
  which are indicated by the values on the $x$-axis. All estimators are
  initialized to the value $0.2$.}
\label{fig:fig_supp1}
\end{figure}

In a second experiment, we demonstrate the ability of the Spatial Debiased
Whittle to perform well for rectangular but not square domains, see
Figure~\ref{fig:fig_discrete2}. We fix one
side of the domain to $16$ units, while the other side grows in powers of 4,
so that the sample sizes increase in the same way as in the previous
experiment.
In this configuration, the asymptotic bias of the tapered Whittle is
non-zero---this is because the expected periodogram never converges to the
spectral density, due to the bounded sample size along one dimension. In
contrast, the observed rate of the Spatial Debiased Whittle likelihood 
 remains of the order of square root the sample size.}

\begin{figure}[h!]
\centering
\includegraphics[width=\textwidth]{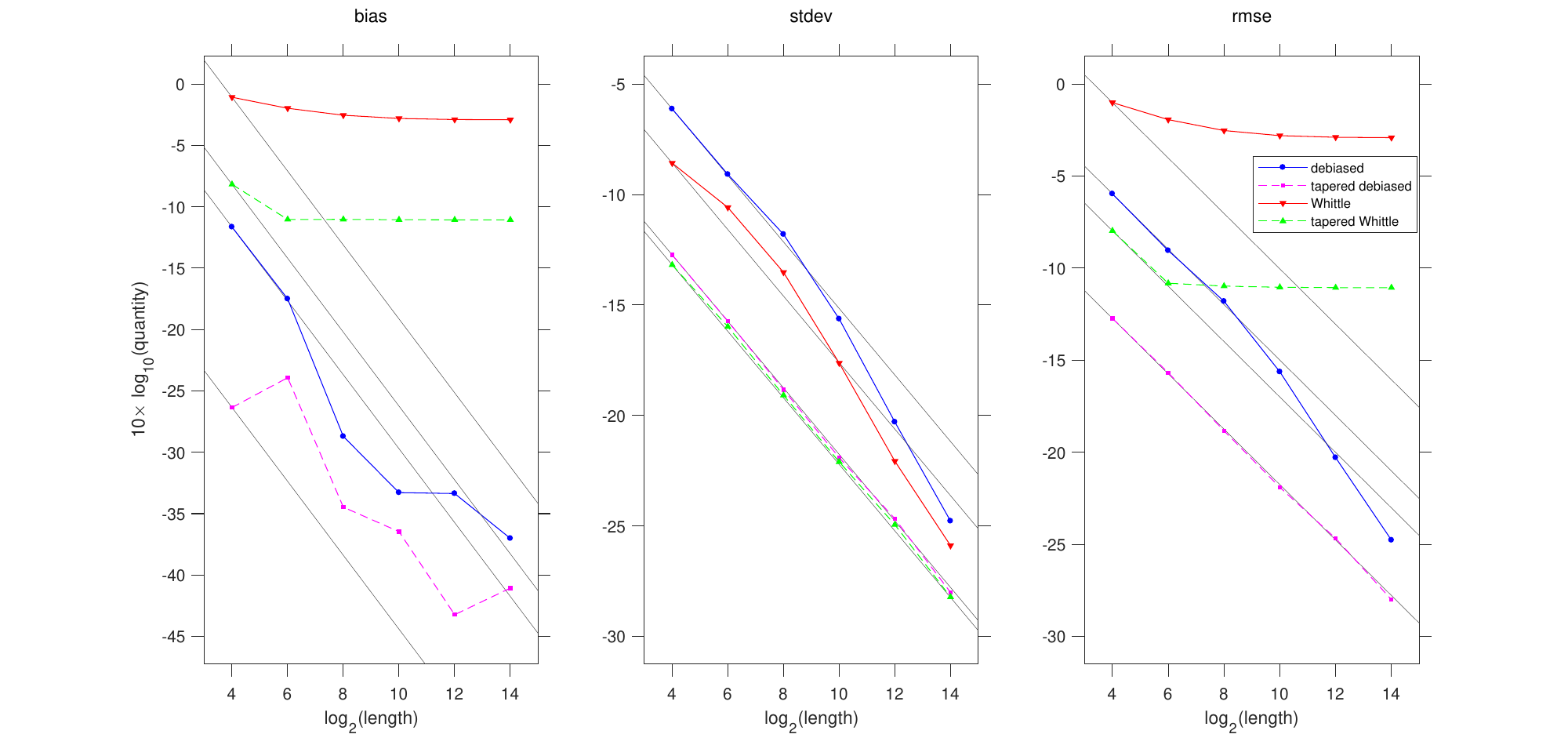}
\caption{Bias, standard deviation, and RMSE of estimates of $\theta=3$ for
  rectangular grids of size $16\times N$ where N increases in powers of $4$
  and is indicated on the $x$-axis. We observe how even for
  a simple discrete model, tapering has its limits and cannot fully account
  for the shape of the observational domain.}
\label{fig:fig_discrete2}
\end{figure}

{\color{black} \section*{Example of a violation of SCC}
\begin{figure}
\includegraphics[width=0.9\textwidth]{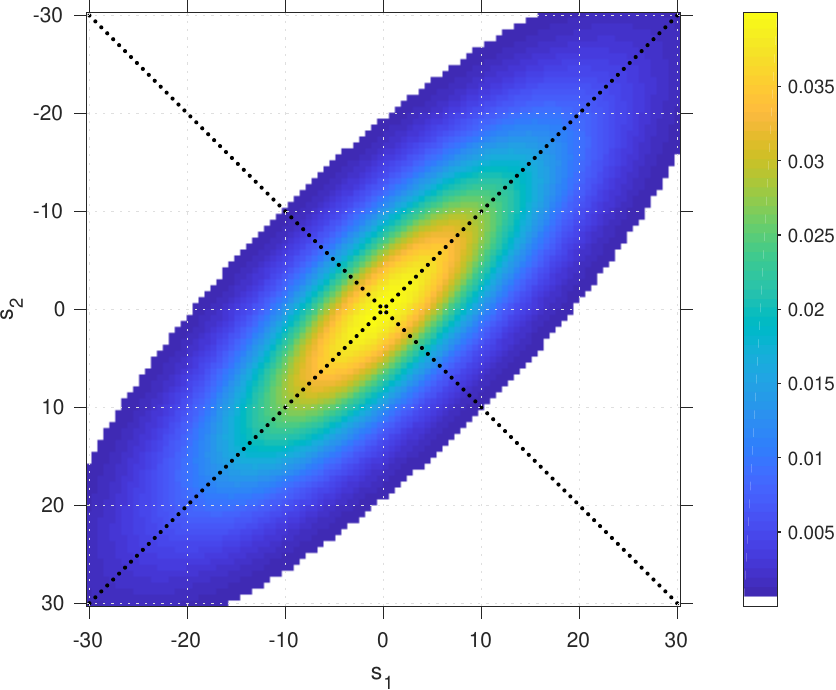}
\caption{This figure illustrates the geometric characteristics of
  Significant Correlation Contribution in 2D. We plot $c_X(\bu;\bm{\theta})$
  as a contour plot and superimpose $\mathbf{s}_{1,j}$ and
  $\mathbf{s}_{2,j}$, even if $\mathbf{s}_{1,j}$ and $\mathbf{s}_{2,j}$ are
  used to sample $X^{(r)}(\mathbf{s})$ to give $X^{(r)}_{\mathbf{s}}$ rather
  than sampling the covariance kernel.\label{fig:SCC}}
\end{figure}
We start by assuming that the autocovariance is
\begin{equation}
  \label{correlationfn}
  c_X(\bu;\bm{\theta})=0.04\times \exp\left\{-\frac{\theta_1}{2} (u_1+u_2)^2\right\}\exp\left\{-\frac{\theta_2}{2} (-u_1+u_2)^2\right\},
\end{equation}
and then we sample the process according to 
\[
\mathbf{s}_{1,j}=\begin{pmatrix} j& j\end{pmatrix},\quad
\mathbf{s}_{2,j}=\begin{pmatrix} -j& j\end{pmatrix}.\] It is fairly
straightforward, with either these line samplings, to convince oneself that
with one sampling we only learn about $\theta_1$ and with the other only
about $\theta_2$ as illustrated by Figure~\ref{fig:SCC}. Note that $c_X()$
is a valid auto--covariance, as the Fourier transform of Gaussians is
Gaussian and thus non--negative. Sampling the process $\mathbf{X}(\bs)$ with
$\mathbf{s}_{1,j}$ means that sums and differences of the sampling pattern
lives in the same linear subspace of ${\mathbb{R}}^2$. This means that we
only learn about one of the two functions in \eqref{correlationfn}.}

\section*{Proofs of lemmata, propositions and theorems}\label{sec:proof}
\subsection*{Proof of lemma 2}
\begin{proof}
    Let $\bkk=(k_0, \ldots, k_{d-1})\in\prod_{j=0}^{d-1}\left\{0,\ldots ,n_j-1\right\}$. We remind
    the reader that for $\bu\in\Z^d$, $\overline{c}_{\bn}(\bu) = c_{g,
      \bn}(\bu) c_X(\bu)$, where,
    \begin{equation*}
        c_{g,\bn}(\bu) = \frac{\sum_{\bs\in\Z^d}{g_\bs g_{\bs+\bu}}}
        {\sum_{\bs\in\Z^d}{g_\bs^2}}.
    \end{equation*}
    Using the fact that for any $\bq\in\left\{0,1\right\}^d$, 
    \begin{equation*}
        \overline{c}_{\bn}(\bu - \bq \circ \bn) 
        \exp\left(-i\sum_{j=0}^{d-1}{\frac{2k_j\pi}{n_j}(u_j-q_jn_j)}\right)
        = \overline{c}_{\bn}(\bu - \bq \circ \bn) 
        \exp\left(-i\sum_{j=0}^{d-1}{\frac{2k_j\pi}{n_j}u_j}\right),
    \end{equation*}
    where $\circ$ denotes the Hadamard product, i.e. element-wise
    multiplication,  and since $\overline{c}_{\bn}(\bu - \bq \circ \bn)$ is
    zero if any component of $\bu$ is zero and the corresponding component
    of $\bq$ is one (due to the definition of $c_{g, \bn_k}$), we obtain the
    proposed formula.  Indeed, any
    $\bu\in\prod_{j=0}^{d-1}\left\{-(n_j-1),\ldots ,n_j-1\right\}$ that
    contribute to the LHS of the proposed formula can be written as $\bu =
    \bu^+ - \bq \circ \bn$ for some unique
    $\bu^+\in\prod_{j=0}^{d-1}\left\{0,\ldots ,n_j-1\right\}$.  The extra
    terms in the RHS of the proposed formula take value zero according to
    the previous argument.
\end{proof}

\subsection*{Proof of Lemma~\ref{lemma:scc1}}

\begin{proof}
This comes as a consequence of the two following observations. First, two
continuous functions on $\T^d$ are equal if and only if their Fourier
coefficients are equal, see for instance~\citet{fourieranalysis}. Second,
for a sequence of full rectangular grids indexed by $k\in\N$
 that grow unbounded in all
directions, for any $\bu\in\Z^d$, we have $c_{g,n_k}(\bu)\rightarrow 1$ as
$k$ goes to infinity, see equation~\eqref{eq:cgcomplete} in the main body.
\end{proof}

\subsection*{Proof of Lemma~\ref{lemma:scc2}}

\begin{proof}
The argument is very similar to that of~Lemma~\ref{lemma:scc1}, with the
difference that for any $\bu\in\Z^d$ we have that $c_{g,n_k}(\bu)$ converges
to a positive constant (which might be strictly smaller than one) as $k$
goes to infinity.
\end{proof}

\subsection*{Proof of Theorem~\ref{th:consistency}}

\begin{proof}
We will show in Lemma~\ref{lemma:cvglkh} that $l_{\bn_k}(\cdot)$, 
as a random function, converges
uniformly to $\widetilde{l}_{\bn_k}(\cdot)$ in probability, i.e. their
difference converges uniformly to the zero function in probability. Hence
the difference $l_{\bn_k}(\hat{\btheta}_k) -
\widetilde{l}_{\bn_k}(\hat{\btheta}_k)$ converges to zero in
probability. Additionally, $l_{\bn_k}(\hat{\btheta}_k) -
\widetilde{l}_{\bn_k}(\btheta)$ converges to zero in probability. Indeed, by
definition, the parameter vector $\hat{\btheta}_k$ minimizes the function
$l_{\bn_k}(\cdot)$ over the parameter set $\Theta$, and according to
Lemma~\ref{lemma:minofexptedlkh}, the parameter vector $\btheta$ minimizes
the function $\widetilde{l}_{\bn_k}(\cdot)$. We therefore have, by the
triangle inequality,
\begin{equation*}
\left|\widetilde{l}_{\bn_k}(\btheta) -
\widetilde{l}_{\bn_k}(\hat{\btheta}_k)\right| \leq \left|
l_{\bn_k}(\hat{\btheta}_k) -
\widetilde{l}_{\bn_k}(\hat{\btheta}_k)\right| + \left |
l_{\bn_k}(\hat{\btheta}_k) - \widetilde{l}_{\bn_k}(\btheta) \right|,  
\end{equation*}
which converges to zero in probability. Making use of
Lemma~\ref{lemma:cvgtotheta} we conclude that $\hat{\btheta}_k$
converges in probability to $\btheta$.
\end{proof}

\subsection*{Proof of Proposition~\ref{prop=varianceLinearCombinations}}

\begin{proof}
We write the proof for the case $\Omega_{\bn}=\Omega_{\bn}^{(1)}$,
while the case $\Omega_{\bn}=\Omega_{\bn}^{(2)}$ is the same, up
to a constant factor.
We first write the proof of the proposition for the univariate Gaussian
case.  Let $a_\text{max}>0$ be a finite constant such that
$\left|a_\bn(\bk)\right|\leq a_\text{max}, \forall\bk\in\T,\forall
\bn\in\N^d$.  We first make the observation that the sum of the periodogram
values at the Fourier frequencies is the squared $L_2$ norm of the sample,
up to some multiplicative constant, since the Discrete Fourier Transform is
orthonormal, i.e.  \revision{
\begin{equation*}
	\sum_{\bk\in\Omega_{\bn}}{I_\bn(\bk)} = 
	\frac{\n}{(2\pi)^{d}\sum_{\bs\in\Z^d}{g_\bs^2}}
	\sum_{\bs\in\Z^d}{g_\bs^2 X_\bs^2}.
\end{equation*}
}
Therefore,
\revision{
\begin{align}
\label{eq:3006200832}
\var\left\{\n^{-1}\sum_{\bk\in\Omega_{\bn}}a_\bn(\bk)I_\bn(\bk)\right\}
&\nonumber\leq a_\text{max}^2 \var\left\{\n^{-1}\sum_{\bk\in\Omega_{\bn}}I_\bn(\bk)\right\} \\
& = 
\frac{a_\text{max}^2}{(2\pi)^{2d}\left(\sum{g_\bs^2}\right)^2}
\var\left\{
	\sum_{\bs\in\Z^d}{g_\bs^2 X_\bs^2}
\right\}.
\end{align}
} Note that the first inequality is valid since the covariance of the
periodogram at two Fourier frequencies $\bk, \bk'$ is non-negative for a
Gaussian process (as a consequence of Isserlis' theorem). Indeed, letting
\begin{equation*}
J(\bk) = \frac{(2\pi)^{-d/2}}{\sqrt{\sum_{\bs\in\Jx}{g_{\bs}^2}}}
\sum_{\bs\in\Jx}{g_\bs X_\bs
  \exp(-i\bk\cdot\bs}),
\end{equation*} we have, by Isserlis' theorem,
\begin{align*}
\cov\left\{I(\bk), I(\bk')\right\} 
&=\E\left\{J(\bk) J^*(\bk) J(\bk') J^*(\bk')\right\}-\E\left\{I(\bk)\right\} \E\left\{ I(\bk')\right\}\\\
&=\E\left\{J(\bk) J(\bk') \right\}\E\left\{J^*(\bk) J^*(\bk')\right\}+
  \E\left\{J(\bk) J^*(\bk') \right\}\E\left\{J^*(\bk)J(\bk') \right\}\\
&=\left| \E\left\{ J(\bk) J(\bk') \right\} \right|^2+
  \left|\E\left\{J^*(\bk) J(\bk') \right\} \right|^2,
\end{align*}
which is non-negative as the sum of two squares.  We study the term
\revision{$\var\left\{ \sum_{\bs\in\Z^d}{g_\bs^2 X_\bs^2} \right\}$}. We
have, again using Isserlis' theorem for Gaussian random variables,
\revision{
\begin{align}
\label{eq:3006200833}
\var\left\{
	\sum_{\bs\in\Z^d}{g_\bs^2 X_\bs^2}
\right\} &= \nonumber\E\left(\sum_{\bs\in\Jx}{g_\bs^2X_\bs
  ^2}\right)^2 - \left(\E\sum_{\bs\in\Jx}{g_\bs^2X_\bs ^2}\right)^2\\ 
& =\nonumber\sum_{\bs\in\Jx}{\sum_{\bs'\in\Jx}{\E\left\{g_\bs^2g_{\bs'}^2X_\bs^2X_{\bs'}^2\right\}
    - \E\left\{g_\bs^2X_\bs^2\right\}\E\left\{g_{\bs'}^2X_{\bs'}^2\right\}}}\\ 
	& = 2\sum_{\bs\in\Jx}{\sum_{\bs'\in\Jx}{
g_\bs^2 g_{\bs'}^2\left(\E\left\{X_\bs X_{\bs'}\right\}\right)^2}}
.
\end{align}
}
We now obtain, combining equations~\eqref{eq:3006200832}
and~\eqref{eq:3006200833},
\begin{align*}
\var\left\{|\bn_k|^{-1}\sum_{\bk\in\Omega_{\bn}}a_\bn(\bk)I_\bn(\bk)\right\} &\leq 
\frac{2a_\text{max}^2}{(2\pi)^{2d}\left(\sum{g_\bs^2}\right)^2}
\sum_{\bs\in\Jx}{\sum_{\bs'\in\Jx}{g_\bs^2 g_{\bs'}^2
    \left(\E\left\{X_\bs X_{\bs'}\right\}\right)^2}}\\
&\leq
\frac{2a_\text{max}^2}{(2\pi)^{2d}\left(\sum{g_\bs^2}\right)^2}
\sum_{\bu\in\Z^d}{c_X(\bu)^2\sum_{\bs\in\Jx}{g_\bs^2 g_{\bs+\bu}^2}
}\\
&\leq
\frac{2a_\text{max}^2}{(2\pi)^{2d}\left(\sum{g_\bs^2}\right)^2}
\sum_{\bu\in\Z^d}{c_X(\bu)^2\sum_{\bs\in\Jx}{g_\bs g_{\bs+\bu}}
}\\
& \leq
\frac{2a_\text{max}^2}{(2\pi)^{2d}\sum{g_\bs^2}}
\sum_{\bu\in\Z^d}{c_X(\bu)^2 c_g(\bu)},
\end{align*}
where we have made use of the assumption that $0\leq g_\bs\leq 1, \forall
\bs\in\Z^d$.
Therefore, we obtain the stated result, i.e.,
\begin{equation*}
\var\left\{|\bn_k|^{-1}\sum_{\bk\in\Omega_{\bn_k}} a_k(\bk)I_{\bn_k}(\bk)
\right\} = \mathcal{O}\left\{
\frac{\sum_{\bu\in\Z^d}{c_X(\bu)^2 c_{g,k}(\bu)}}
{\sum{g_\bs^2}}
\right\},
\end{equation*}
where the big $O$ is with respect to $k$ going to infinity. This concludes
the proof for the univariate Gaussian case.\qed

\subsection*{Proof of Corollary~\ref{cor=varianceLinearCombinations}}

We now treat the extension to the univariate but non-Gaussian case. This
requires defining the fourth-order cumulant according to,
\begin{align}
\nonumber
    \E\{X_{\mathbf{s}_1}X_{\mathbf{s}_2}X_{\mathbf{s}_3}X_{\mathbf{s}_4}\}&=c_4(\mathbf{s}_2-\mathbf{s}_1,\mathbf{s}_3-\mathbf{s}_1,\mathbf{s}_4-\mathbf{s}_1)+c_X(\mathbf{s}_3-\mathbf{s}_1)c_X(\mathbf{s}_4-\mathbf{s}_2)\\
    \nonumber
    &+
    c_X(\mathbf{s}_4-\mathbf{s}_1)c_X(\mathbf{s}_3-\mathbf{s}_2)+c_X(\mathbf{s}_2-\mathbf{s}_1)c_X(\mathbf{s}_4-\mathbf{s}_3).
\end{align}
Note that in the Gaussian case this equality holds with
$C_4(\mathbf{s}_2-\mathbf{s}_1,\mathbf{s}_3-\mathbf{s}_1,\mathbf{s}_4-\mathbf{s}_1)
= 0$ (trivially) always. With this definition, we can study the covariance
of the periodogram at two Fourier frequencies as follows,
\begin{align*}
\nonumber
\cov\{ I_{\mathbf{n}}(\bm{\omega}_1), I_{\mathbf{n}}(\bm{\omega}_2) \}
&=
\frac{(2\pi)^{-2d}}{(\sum_{\mathbf{s}}g^2_{\mathbf{s}})^2}
\cov\left\{
   \sum_{\bs_1, \bs_2} g_{\mathbf{s}_1}
   g_{\mathbf{s}_2}X_{\mathbf{s}_1}X_{\mathbf{s}_2}e^{-i\bm{\omega_1}^T(\mathbf{s}_1-\mathbf{s}_2)},
   \sum_{\bs_3, \bs_4} g_{\mathbf{s}_3}
   g_{\mathbf{s}_4}X_{\mathbf{s}_3}X_{\mathbf{s}_4} e^{-i\bm{\omega_2}^T(\mathbf{s}_3-\mathbf{s}_4)}\right\}\\
   &= \frac{(2\pi)^{-2d}}{(\sum_{\mathbf{s}}g^2_{\mathbf{s}})^2}\sum_{\bs_1, \bs_2, \bs_3, \bs_4} g_{\mathbf{s}_1}
   g_{\mathbf{s}_2}g_{\mathbf{s}_3}
   g_{\mathbf{s}_4}\cov\{X_{\mathbf{s}_1}X_{\mathbf{s}_2},X_{\mathbf{s}_3}X_{\mathbf{s}_4}\}
   e^{-i\bm{\omega_1}^T(\mathbf{s}_1-\mathbf{s}_2)}
   e^{-i\bm{\omega_2}^T(\mathbf{s}_3-\mathbf{s}_4)}.
\end{align*}
We write $C_k=\frac{(2\pi)^{-d}}{\sum_{\mathbf{s}}g^2_{\mathbf{s}}}$,
where the dependence on $k$ comes from the implicit dependence of 
$\{g_\bs\}$ on $k$.
We note that we can determine directly that
\begin{align}
\label{eq:cum4}
\nonumber
\cov\{X_{\mathbf{s}_1}X_{\mathbf{s}_2},X_{\mathbf{s}_3}X_{\mathbf{s}_4} \}&=
\E\{X_{\mathbf{s}_1}X_{\mathbf{s}_2}X_{\mathbf{s}_3}X_{\mathbf{s}_4}\}-\E\{X_{\mathbf{s}_1}X_{\mathbf{s}_2}\}\E\{X_{\mathbf{s}_3}X_{\mathbf{s}_4}\}\\
\nonumber
&=c_4(\mathbf{s}_2-\mathbf{s}_1,\mathbf{s}_3-\mathbf{s}_1,\mathbf{s}_4-\mathbf{s}_1)+c_X(\mathbf{s}_3-\mathbf{s}_1)c_X(\mathbf{s}_4-\mathbf{s}_2)\\
&+c_X(\mathbf{s}_4-\mathbf{s}_1)c_X(\mathbf{s}_3-\mathbf{s}_2).
\end{align}
We additionally define,
\begin{equation*}
\mathcal{G}(\bs_1, \bs_2, \bs_3) = c_X(\mathbf{s}_3-\mathbf{s}_1)c_X(\mathbf{s}_4-\mathbf{s}_2)
+c_X(\mathbf{s}_4-\mathbf{s}_1)c_X(\mathbf{s}_3-\mathbf{s}_2),
\end{equation*}
where the choice of the letter $\mathcal{G}$ comes from the fact that in the
Gaussian case
$\cov\{X_{\mathbf{s}_1}X_{\mathbf{s}_2},X_{\mathbf{s}_3}X_{\mathbf{s}_4} \}$
simplifies to this quantity.  Now summing over $2$-combinations of Fourier
frequencies, we can apply the triangular inequality,
{\tiny \begin{align*}
\nonumber
	&\left|
		\sum_{\bm{\omega_1}, \bm{\omega_2}}
		{
			a_{\bm{\omega_1}} a_{\bm{\omega_2}}
			\cov\{ I_{\mathbf{n}}(\bm{\omega}_1), I_{\mathbf{n}}(\bm{\omega}_2) \}
		}
	\right|
	=\\
	\nonumber
		&\left|
			\sum_{\bm{\omega_1}, \bm{\omega_2}}
				a_{\bm{\omega_1}} a_{\bm{\omega_2}}
				C_k^2\sum_{\mathbf{s}_1, \bs_2, \bs_3, \bs_4} 
		  			 g_{\mathbf{s}_1}g_{\mathbf{s}_2}g_{\mathbf{s}_3}g_{\mathbf{s}_4}
		  			 \left\{
		  			 c_4(\mathbf{s}_2-\mathbf{s}_1,\mathbf{s}_3-\mathbf{s}_1,\mathbf{s}_4-\mathbf{s}_1)
		  			 + \mathcal{G}(\bs_1, \bs_2, \bs_3) 
		  			 \right\}
		  			 e^{-i\bm{\omega_1}^T(\mathbf{s}_1-\mathbf{s}_2)}
   					e^{-i\bm{\omega_2}^T(\mathbf{s}_3-\mathbf{s}_4)}
		\right|\\
		\nonumber
	&\leq
	\left|
		\sum_{\bm{\omega_1}, \bm{\omega_2}}
		{
			a_{\bm{\omega_1}} a_{\bm{\omega_2}}
			C_k^2\sum_{\bs_1, \bs_2, \bs_3, \bs_4}
	  			 g_{\mathbf{s}_1}g_{\mathbf{s}_2}g_{\mathbf{s}_3}g_{\mathbf{s}_4}
	  			 c_4(\mathbf{s}_2-\mathbf{s}_1,\mathbf{s}_3-\mathbf{s}_1,\mathbf{s}_4-\mathbf{s}_1)
	  			 e^{-i\bm{\omega_1}^T(\mathbf{s}_1-\mathbf{s}_2)}
   				e^{-i\bm{\omega_2}^T(\mathbf{s}_3-\mathbf{s}_4)}
		}
	\right|\\ 
	&+
	\left|
			\sum_{\bm{\omega_1}, \bm{\omega_2}}
				a_{\bm{\omega_1}} a_{\bm{\omega_2}}
				C_k^2\sum_{\bs_1, \bs_2, \bs_3, \bs_4}
		  			 g_{\mathbf{s}_1}g_{\mathbf{s}_2}g_{\mathbf{s}_3}g_{\mathbf{s}_4}
		  			 \mathcal{G}(\bs_1, \bs_2, \bs_3) 
		  			 e^{-i\bm{\omega_1}^T(\mathbf{s}_1-\mathbf{s}_2)}
   					e^{-i\bm{\omega_2}^T(\mathbf{s}_3-\mathbf{s}_4)}
	\right|
\end{align*}}
The second term in the sum has already been studied in
the proof of Proposition~\ref{prop=varianceLinearCombinations} 
where we assumed Gaussianity. As
for the first term, again using the triangular inequality,
we may deduce that
{\small \begin{align*}
\nonumber
&\left|
		\sum_{\bm{\omega_1}, \bm{\omega_2}}
			a_{\bm{\omega_1}} a_{\bm{\omega_2}}
			C_k^2\sum_{\mathbf{s}_1} \sum_{\mathbf{s}_2}
  			 \sum_{\mathbf{s}_3} \sum_{\mathbf{s}_4}
	  			 g_{\mathbf{s}_1}g_{\mathbf{s}_2}g_{\mathbf{s}_3}g_{\mathbf{s}_4}
	  			 c_4(\mathbf{s}_2-\mathbf{s}_1,\mathbf{s}_3-\mathbf{s}_1,\mathbf{s}_4-\mathbf{s}_1)
	  			 e^{-i\bm{\omega_1}^T(\mathbf{s}_1-\mathbf{s}_2)}
				   e^{-i\bm{\omega_2}^T(\mathbf{s}_3-\mathbf{s}_4)}
\right|\\
\nonumber
&\leq
\sum_{\bm{\omega_1}, \bm{\omega_2}}
			a_{\bm{\omega_1}} a_{\bm{\omega_2}}
			C_k^2\sum_{\mathbf{s}_1} \sum_{\mathbf{s}_2}
  			 \sum_{\mathbf{s}_3} \sum_{\mathbf{s}_4}
	  			 g_{\mathbf{s}_1}g_{\mathbf{s}_2}g_{\mathbf{s}_3}g_{\mathbf{s}_4}
	  			 \left|
	  			 	c_4(\mathbf{s}_2-\mathbf{s}_1,\mathbf{s}_3-\mathbf{s}_1,\mathbf{s}_4-\mathbf{s}_1)
	  			 	e^{-i\bm{\omega_1}^T(\mathbf{s}_1-\mathbf{s}_2)}
   					e^{-i\bm{\omega_2}^T(\mathbf{s}_3-\mathbf{s}_4)}
   				\right|\\
&\leq
\sum_{\bm{\omega_1}, \bm{\omega_2}}
			a_{\bm{\omega_1}} a_{\bm{\omega_2}}
			\left|\bn\right|
			C_k^2\sum_{\bm{\tau}_1} \sum_{\bm{\tau}_2}\sum_{\bm{\tau}_3}
	  			 \left|
	  			 	c_4(\bm{\tau}_1,\bm{\tau}_2,\bm{\tau}_3)
	  			 \right|.
	  			 \label{eqn-allterms}
\end{align*}}
We now make use of our assumption of absolute summability of fourth-order
cumulants. Defining the positive finite constant 
$K_4=\sum_{\bm{\tau}_1=0}^\infty \sum_{\bm{\tau}_2=0}^\infty
\sum_{\bm{\tau}_3=0}^\infty
	  			 \left|
	  			 	c_4(\bm{\tau}_1,\bm{\tau}_2,\bm{\tau}_3)
	  			 \right|
<\infty$, we obtain,
\begin{equation*}
\frac{1}{|\bn_k|^2}
	\left|
		\sum_{\bm{\omega_1}, \bm{\omega_2}}
		{
			a_{\bm{\omega_1}} a_{\bm{\omega_2}}
			\cov\{ I_{\mathbf{n}}(\bm{\omega}_1), I_{\mathbf{n}}(\bm{\omega}_2) \}
		}
	\right|
\leq
\frac{\sum_{\bu\in\Z^d}{c_X(\bu)^2 c_{g,k}(\bu)}}
{\sum{g_\bs^2}}
+
|\bn_k| C_k^2 K_4,
\end{equation*}
where the first term is the one obtained also for Gaussian random fields.
This allows us to conclude, under our assumption of absolute summability of
fourth-order cumulants, that in the non-Gaussian case,
\begin{equation}
\var\left\{|\bn_k|^{-1}\sum_{\bk\in\Omega_{\bn_k}} a_k(\bk)I_{\bn_k}(\bk)
\right\} = \mathcal{O}\left\{
\frac{\sum_{\bu\in\Z^d}{c_X(\bu)^2 c_{g,k}(\bu)}}
{\sum{g_\bs^2}}
+ 
\frac
{
\left|\bn_k\right|
}
{
\left(\sum g_{\bs}^2\right)^2
}
\right\}.
\label{eqn:nonlin-var}
\end{equation}\qed

\subsection*{Proof of Corollary~\ref{cor=varianceSesLinearCombinations}}
\begin{proof}
For a multivariate random field we proceed much in the same way as the proof of
 Proposition~\ref{prop=varianceLinearCombinations}. We study
the variance of the quadratic form
\begin{equation}
\label{eq:03030107}
|\bn_k|^{-1}\sum_{\bk\in\Omega_{\bn_k}}
{\bJ_{\bn_k}^*(\bk) A_k(\bk) \bJ_{\bn_k}(\bk) }.
\end{equation}
For all $\bk\in\Omega_{\bn_k}$ we perform an orthonormal eigendecomposition
of $A_k(\bk)$,
\begin{equation*}
	A_k(\bk) = \sum_{j=1}^p{\lambda_j(\bk) \be_j(\bk) \be_j(\bk)^H},
\end{equation*}
where we do not indicate the dependence on $k$ to avoid complicating the
notation.  We then define the complex-valued scalars,
\begin{equation*}
	Z_j = \bJ^H(\bk) \be_j(\bk), \quad j= 1,\ldots,p,
\end{equation*}
and note that, due to the orthonormality of the basis $\be_1, \ldots,
\be_p$,
\begin{equation*}
	\bJ_{\bn_k}(\bk) = \sum_{j=1}^p{Z_j(\bk) \be_j(\bk)}.
\end{equation*}
We have,
\begin{align*}
\var
\left\{
|\bn_k|^{-1}\sum_{\bk\in\Omega_{\bn_k}}
{ \bJ_{\bn_k}^H(\bk) A_k(\bk) \bJ_{\bn_k} (\bk)}
\right\}&=\var
\left\{|\bn_k|^{-1}\sum_{\bk} \sum_{j=1}^p \lambda_{j}(\bk)|Z_{j}(\bk)|^2
\right\}\\
&=|\bn_k|^{-2}\sum_{\bk_1, \bk_2} \sum_{j_1, j_2} \lambda_{j_1}(\bk_1)\lambda_{j_2}(\bk_2)\cov\{ |Z_{j_1}(\bk_1)|^2,|Z_{j_2}(\bk_2)|^2\}.
\end{align*}
Using Isserliss' theorem we deduce, for any $\bk_1, \bk_2\in\Omega_{\bn_k}^2$, $j_1, j_2 = 1, \ldots, p$, 
\begin{align*}
\nonumber
\cov\{ |Z_{j_1}(\bk_1)|^2,|Z_{j_2}(\bk_2)|^2\}&=\E\{Z_{j_1}(\bk_1)Z_{j_2}(\bk_2)\}\E\{Z_{j_1}^\ast(\bk_1)Z_{j_2}^\ast(\bk_2)\}\\
&+\E\{Z_{j_1}(\bk_1)Z_{j_2}^\ast(\bk_2)\}\E\{Z_{j_1}^\ast(\bk_1)Z_{j_2}(\bk_2)\}\geq 0.
\label{isser}
\end{align*}
Therefore it follows that 
\begin{align*}
\var
\left\{
|\bn_k|^{-1}\sum_{\bk\in\Omega_{\bn_k}}
{ \bJ_{\bn_k}^H(\bk) A_k(\bk) \bJ_{\bn_k} (\bk)}
\right\}
&\leq \lambda_{\max}^2|\bn_k|^{-2}\sum_{\bk_1} \sum_{\bk_2} \sum_{j_1} \sum_{j_2} \cov\{ |Z_{j_1}(\bk_1)|^2,|Z_{j_2}(\bk_2)|^2\}
.
\end{align*}
Besides,
\begin{align*}
\var
\left\{
|\bn_k|^{-1}\sum_{\bk\in\Omega_{\bn_k}}
{ \bJ_{\bn_k}^H(\bk) \bJ_{\bn_k} (\bk)}
\right\}
&=
\var
\left\{
|\bn_k|^{-1}\sum_{\bk\in\Omega_{\bn_k}}
{ 
\left(
\sum_{j=1}^p{Z_j^* \be_j^H(\bk)}
\right)
\left(
\sum_{j=1}^p{Z_j \be_j(\bk)}
\right)
}
\right\}\\
&=
\var
\left\{
|\bn_k|^{-1}\sum_{\bk\in\Omega_{\bn_k}}
{ 
\left(
\sum_{j_1, j_2=1}^p{Z_{j_1}^* Z_{j_2} \be_{j_1}^H(\bk) \be_{j_2}(\bk)}
\right)
}
\right\}\\
&=
\var
\left\{
|\bn_k|^{-1}\sum_{\bk\in\Omega_{\bn_k}}
{ 
\sum_{j=1}^p{\left|Z_{j}\right|^2}
}
\right\},
\end{align*}
after we again use the orthonormality of the basis $\be_1, \ldots, \be_p$.
Hence we deduce that,
\begin{equation*}
\var
\left\{
|\bn_k|^{-1}\sum_{\bk\in\Omega_{\bn_k}}
{ \bJ_{\bn_k}^H(\bk) A_k(\bk) \bJ_{\bn_k} (\bk)}
\right\}
\leq \lambda_{\max}^2
\var
\left\{
|\bn_k|^{-1}\sum_{\bk\in\Omega_{\bn_k}}
{ \bJ_{\bn_k}^H(\bk) \bJ_{\bn_k} (\bk)}.
\right\}
\end{equation*}
As in the univariate case, we use the isometry property of the discrete
Fourier transform to write this in the form of,
\begin{align*}
\var
\left\{
|\bn_k|^{-1}\sum_{\bk\in\Omega_{\bn_k}}
{ \bJ_{\bn_k}(\bk)^H A_k(\bk) \bJ_{\bn_k}(\bk) }
\right\}
&\leq
\lambda_\mathrm{max}^2
\var
\left\{
\sum_{q=1,\ldots,p}
\left\{
\frac
{
	(2\pi)^{-d}
}
{
	{
		\sum{{g_{\bs'}^\pq}^2}
	}
}
\sum
_{\bs}
{
	{g_{\bs}^\pq}^2
	{X_{\bs}^\pq}^2
}
\right\}
\right\}.
\end{align*}
By applying the Isserlis theorem, we obtain the following upper-bound,
\begin{align*}
&\var
\left\{
|\bn_k|^{-1}\sum_{\bk\in\Omega_{\bn_k}}
{ {\bJ_{\bn_k}(\bk)}^H A_k(\bk) \bJ_{\bn_k}(\bk) }
\right\}
\\&\leq
2 \lambda_\mathrm{max}^2
\sum_{q,r}
\frac
{
	(2\pi)^{-2d}
}
{
	{
		\sum{{g_{\bs'}^\pq}^2}
		\sum{{g_{\bs'}^\pr}^2}
	}
}
\left\{
	\sum_{s,s'}
	{
		{g_\bs^\pq}^2
		{g_{\bs'}^\pr}^2
		\left(
			\E
			\left[
				X_\bs^\pr
				X_{\bs'}^\pq
			\right]
		\right)^2
	}
\right\}.
\end{align*}
By a manipulation similar to the one we used earlier for the univariate case,
we obtain,
\begin{equation*}
	\var
\left\{
|\bn_k|^{-1}\sum_{\bk\in\Omega_{\bn_k}}
{ \bJ_{\bn_k}(\bk)^* A_k(\bk) \bJ_{\bn_k}(\bk) }
\right\}
=
\mathcal{O}
\left\{
\sum_{q,r}
\frac{
	\sum_{\bu}
	{	
		c_X^\pqr(\bu)^2 c_g^\pqr(\bu)
	}
}
{
\sqrt{
\sum{{g_\bs^\pq}^2}
\sum{{g_\bs^\pr}^2}
}
}
\right\},
\end{equation*}
which determines the order of the variance of such quadratic forms.\qed
\end{proof}

\if1\BiasStuff{
\subsection*{Proof of Proposition~\ref{propcov}}\label{proof:varianceprop}

\begin{proof}
Let us assume that the deterministic matrix ${\bm
  {\mathcal{H}}}(\bm{\theta})$ is invertible. Following on from
Theorem~1 of~\citep{sykulski2016biased} this result follows directly
from brute force calculations. Then it follows from~\eqref{eqn:score}
\begin{align*}
  \widehat{\bm{\theta}}-\bm{\theta}&\approx
  -{\bm   {\mathcal{H}}}^{-1}(\bm{\theta})
  \nabla  \ell_M(\bm{\theta}).
\end{align*}
It is clear from this equation that for large random fields the
estimate $\widehat{\bm{\theta}}$ is unbiased. We then note that
\begin{equation*}
  \var\{  \widehat{\bm{\theta}}\}\approx
  - {\bm   {\mathcal{H}}}^{-1}(\bm{\theta})\var \{ \nabla  \ell_M(\bm{\theta})\}
  (-  {\bm   {\mathcal{H}}}^{-T}(\bm{\theta})).
\end{equation*}
The $\approx$ sign above signifies that as the random field becomes
unbounded in spread then the ratio of the entries on the left hand
side, by those of the right hand side eventually equates to 1. We can
estimate the Hessian at the value of the MLE by taking the observed
Hessian:
\begin{equation}
  \widehat{  {\bm   {\mathcal{H}}}}(\bm{\theta})
  = {\mathbf{H}}(\widehat{\bm{\theta}}).
\end{equation}
First we recall that the $p$th entry of $\nabla  \ell_M(\bm{\theta})\}$ is
\begin{equation}
  \frac{\partial \ell_M(\bm{\theta})}{\partial
    \theta_p}=
  \n^{-1}\sum_{\bk\in\Omega_\bn}\left\{\frac{\frac{\partial
      \overline{I}_{\bn}(\bk;\btheta)}{\partial \theta_p}
  }{\overline{I}_{\bn}(\bk;\btheta)} - 
  \frac{\partial \overline{I}_{\bn}(\bk;\btheta)}{\partial
    \theta_p}\frac{I_{\bn}(\bk)}{\overline{I}_{\bn}^2(\bk;\btheta)}\right\}. 
\end{equation}
This takes the form of a linear combination of the periodogram with a
deterministic weight. We can note that the variance of this object
takes the form of
\begin{align}\label{eq:stdformula1}
  \var\left\{ \frac{\partial \ell_M(\bm{\theta})}{\partial
    \theta_p}\right\}&=\n^{-2}\sum_{\bk_1\in\Omega_\bn}
  \sum_{\bk_2\in\Omega_\bn}
  \frac{\partial \overline{I}_{\bn}(\bk_1;\btheta)}{\partial
    \theta_p}\frac{\partial
    \overline{I}_{\bn}(\bk_2;\btheta)}{\partial
    \theta_p}\frac{\cov\left\{I_{\bn}(\bk_1),I_{\bn}(\bk_2)\right\}}
{\overline{I}_{\bn}^2(\bk_1;\btheta)\overline{I}_{\bn}^2(\bk_2;\btheta)}. 
\end{align}
The above equation helps us compute the standard errors but we need to
estimate\\ $\cov\left\{I_{\bn}(\bk_1),I_{\bn}(\bk_2)\right\}$. We
start by writing down the spectral representation
\begin{equation}
  X_j=\int dZ(\bk)\exp(\bi\bk\cdot\bs_{j}),\quad
  \bk\in[-\frac{\pi}{\Delta},\frac{\pi}{\Delta}]^d. 
\end{equation}
We then have
\begin{align*}
  \cov\left\{J_\bn(\bk),J_\bn(\bk')\right\} &= \n^{-1}\sum_{j_1=0}^{n-1}
  \sum_{j_2=0}^{n-1}{g_{j_1}g_{j_2}^\ast \cov\{X_{j_1},X_{j_2}\}e^{-\bi(\bk\cdot\bs_{j_1}-\bk'\cdot
      \bs_{j_2})}}\\
  &= \n^{-1}\sum_{j_1=0}^{n-1}
  \sum_{j_2=0}^{n-1}g_{j_1}g_{j_2}^\ast \cov\{\int
  dZ(\bk_1)e^{\bi\bk_1\cdot\bs_{j_1}},\int
  dZ(\bk_2)e^{\bi\bk_2\cdot\bs_{j_2}}\}\\ 
  &\times e^{-\bi(\bk\cdot\bs_{j_1}-\bk'\cdot
    \bs_{j_2})}.
\end{align*}
We now define 
\begin{equation}
  \mathcal{D}_{\bn}(\lambda) = \sum_{\bs\in\Jx}{g_\bs e^{i\lambda\cdot\bs}}.
\end{equation}
This permits us to write
\begin{align*}
  \cov\left\{	J_\bn(\bk),J_\bn(\bk')\right\} &=\n^{-1}\iint 
  \mathcal{D}_{\bn}\left(\bk_1-\bk\right)\mathcal{D}_{\bn}^\ast \left(\bk_2-\bk'\right)
  \cov\left\{dZ(\bk_1),dZ(\bk_2) \right\}\\
  &=\iint 
  \mathcal{D}_{\bn}\left(\bk_1-\bk\right)\mathcal{D}_{\bn}^\ast \left(\bk_2-\bk'\right)
  S(\bk_1)\delta(\bk_1-\bk_2)d\bk_1d\bk_2.
\end{align*}
We assume Gaussianity and using Isserlis' theorem we have
\begin{equation}
  \label{isserlis}
  \cov\left\{I_{\bn}(\bk_1),I_{\bn}(\bk_2)\right\}=\left|
  \cov\left\{J_\bn(\bk_1),J_\bn(\bk_2)\right\}\right|^2. 
\end{equation}
\end{proof}
}\fi
\end{proof}

\subsection*{Proof of Lemma~\ref{lemma:minofexptedlkh}}

\begin{proof}
The difference between the expected likelihood function at the true
parameter vector and any parameter vector $\bgamma\in\Theta$ takes the form
\begin{equation*}
  \widetilde{l}_\bn(\bgamma)-\widetilde{l}_\bn(\btheta) =
  |\bn|^{-1}\sum_{\bk\in\Omega_\bn}{\phi\left(\frac{\Ink{\btheta}}{\Ink{\bgamma}}\right)}, 
  \end{equation*}
with $\phi: x\mapsto x - \log x - 1$. This function is non-negative and
attains it minimum uniquely at $x=1$.

The proof in the multivariate case requires a bit more care than the
univariate case but follows the same pattern.  Following Taniguchi (1979)
and Guillaumin et al (2017) for 1-d and the multivariate version provided
in~\cite{hosoya1982central} we define the function
\[ \widetilde{l}_\bn(\bgamma)= |\mathbf{n}|^{-1}\sum_{\bm{\omega}}
    \left\{\log {\mathrm{det}}\{ \overline{\mathbf{I}}(\bm{\omega};\bm{\gamma})\} +{\mathrm{trace}}\left\{\overline{\mathbf{I}}^{-1}(\bm{\omega};\bm{\gamma})\overline{\mathbf{I}}(\bm{\omega};\bm{\theta}\right\}\right\}
.
\]
We now note that
\begin{align}
\nonumber
\widetilde{l}_\bn(\bgamma)-\widetilde{l}_\bn(\btheta) &=|\mathbf{n}|^{-1}\sum_{\bm{\omega}}
    \left\{{\mathrm{trace}}\left\{\overline{\mathbf{I}}^{-1}(\bm{\omega};\bm{\gamma})\overline{\mathbf{I}}(\bm{\omega};\bm{\theta})\right\}-
    \log\frac
    {\mathrm{det}\overline{\mathbf{I}}(\bm{\omega};\bm{\theta})}
    {\mathrm{det}\overline{\mathbf{I}}(\bm{\omega};\bm{\gamma})}
    -p
    \right\}.
\end{align}
We define
$\mathbf{B}_{\bm{\omega}}(\bm{\theta},\bm{\gamma})=\overline{\mathbf{I}}(\bm{\omega};\bm{\theta})\overline{\mathbf{I}}^{-1}(\bm{\omega};\bm{\gamma})$,
and assume this matrix has positive eigenvalues
$\{\beta_i(\omega)\}_{i=1}^p$. We then obtain,
\begin{align}
\nonumber
\widetilde{l}_\bn(\bgamma)-\widetilde{l}_\bn(\btheta) &=|\mathbf{n}|^{-1}\sum_{\bm{\omega}}
   \sum_j \{\beta_j -\log \beta_j -1\}.
\end{align}
From here, like in the univariate case we make use of the fact that $\phi:
x\mapsto x - \log x - 1$ is non-negative and attains it minimum uniquely at
$x=1$, which corresponds to
$\mathbf{B}_{\bm{\omega}}(\bm{\theta},\bm{\gamma})$ being the identity
matrix.\qed
\end{proof}

\subsection*{Proof of Lemma~\ref{lemma:bounds_periodogram}}

\begin{proof}
By combining equations~\eqref{eq:definitionPeriodogram} and
\eqref{eq:cgdef} in the main body then the periodogram can be
expressed as
\begin{equation*}
I_{\bn}(\bk) = \frac{(2\pi)^{-d}}{\sum{g_\bs^2}}\left|
\sum_{\bs\in\Jx}{g_\bs X_\bs \exp(-i\bk\cdot\bs)} \right|^2,
\qquad\bk\in\T^d.
\end{equation*}
Making use of equation~\eqref{barI} of the main body,
we therefore have,
\begin{align*}
\In{\bgamma} = 
\int_{\T^d}{f_{\delta,X}(\bk-\blambda;\bgamma)\mathcal{F}_{\bn}(\blambda)
d\blambda}.
\end{align*}
Also,
\begin{align*}
 \int_{\T^d}{\mathcal{F}_\bn(\bk)d\bk} &=
\frac{(2\pi)^{-d}}{\sum{g_\bs^2}}\int_{\T^d}{\left|\sum_{\bs\in\Jx}{g_\bs
    \exp(i\bk\cdot\bs}\right|^2d\,\bk}\\ \nonumber &=
\frac{(2\pi)^{-d}}{\sum{g_\bs^2}}\int_{\T^d}{\sum_{\bs\in\Jx}{\sum_{\bs'\in\Jx}{g_\bs g_{\bs'}
      \exp\{i\bk\cdot(\bs'-\bs)\}}}\,d\bk}\\ \nonumber &=
\frac{(2\pi)^{-d}}{\sum{g_\bs^2}}\sum_{\bs\in\Jx}{\sum_{\bs'\in\Jx}{\int_{\T^d}{g_\bs g_{\bs'}
      \exp\{i\bk\cdot(\bs'-\bs)\}}}\,d\bk}\\ \nonumber &=
\frac{1}{\sum{g_\bs^2}}\sum_{\bs\in\Jx}{\sum_{\bs'\in\Jx}{g_\bs g_{\bs'}
    \delta_{\bs,\bs'}}}\\ &= 1,
\end{align*}
which is a direct adaptation of a standard result for the F\'ejer
kernel. Hence,
\begin{align*}
\left| \In{\bgamma} \right|  &\leq 
\int_{\T^d}{\left|
  f_{\delta,X}(\bk-\blambda;\bgamma)\mathcal{F}_{\bn}(\blambda)
  \right|
  d\blambda}\\
&\leq 
f_{\delta,\text{max}}
\int_{\T^d}{\left|
  \mathcal{F}_{\bn}(\blambda)
  \right|
  d\blambda}\\
&\leq f_{\delta,\text{max}}.
\end{align*}
Similarly, we obtain the other inequality, i.e.
$
	\In{\bgamma} \geq f_{\delta,\text{min}},
$
which concludes the proof.\qed
\end{proof}

\subsection*{Proof of Lemma~\ref{lemma=061120182}}
We shall need the following intermediary result in our proof.
\begin{lemma}
	\label{lemma:intermediaryresult}
	We have, for a growing domain,
	\begin{equation*}
	\nk{k}^{-1}\sum_{\bk\in\Omega_{\bn_k}}\left\{
	\Ink{\btheta} - \Ink{\bgamma}
	\right\}^2=
	\sum_{\bu\in\Z^d}\left\{\overline{c}_{\bn_k}(\bu;\btheta)-\overline{c}_{\bn_k}(\bu;\bgamma)\right\}^2 + o(1).
	\end{equation*}
\end{lemma}
\begin{proof}
	We distinguish two cases:
	\begin{enumerate}
		\item In the case where the domain is unbounded,
		we have set $\Omega_{\bn}=\Omega_{\bn}^{(1)}$, see the
		discussion following~\eqref{eq:pseudolkh} in the main document.
		 Then the
		result in obtained by application of Parseval's equality,
		according to which,
		\begin{equation*}
			\sum_{\bu\in\Z^d}\left\{\overline{c}_{\bn_k}(\bu;\btheta)-\overline{c}_{\bn_k}(\bu;\bgamma)\right\}^2=
			\int_{\T^d}\left\{
			\Ink{\btheta} - \Ink{\bgamma}
			\right\}^2d\bk,
		\end{equation*}
		and application of the Dominated Convergence Theorem.
		\item In the case where one or more dimensions of the domain
		are bounded, we use $\Omega_{\bn}=\Omega_{\bn}^{(2)}$, 
		see the
		discussion following~\eqref{eq:pseudolkh} in the main document,
		 and
		in that case we have exactly,
		\begin{equation*}
			\nk{k}^{-1}\sum_{\bk\in\Omega_{\bn_k}}\left\{
			\Ink{\btheta} - \Ink{\bgamma}
			\right\}^2=
			\sum_{\bu\in\Z^d}\left\{\overline{c}_{\bn_k}(\bu;\btheta)-\overline{c}_{\bn_k}(\bu;\bgamma)\right\}^2.
		\end{equation*}
		This can be established by direct calculations using the expression of
		the expected periodogram as a Fourier series provided in
		 Lemma~\ref{eq:computeExpectedPeriodogram}.
	\end{enumerate}
\end{proof}

We can now establish the proof for Lemma~\ref{lemma=061120182}.
\begin{proof}
We start by providing a proof in the scalar case (the non--Gaussian but
scalar case requires no adjustment).  We first observe, given
equation~\eqref{eq:formulaexpectedlikelihood} of the main body, that
\begin{equation*}
  \widetilde{l}_{\bn_{k}}(\bgamma) - \widetilde{l}_{\bn_{k}}(\btheta)
  = \nk{k}^{-1}\sum_{\bk\in\Omega_{\bn_k}}
  \left\{
  \frac{\Ink{\btheta}}{\Ink{\bgamma}}
  - \log\frac{\Ink{\btheta}}{\Ink{\bgamma}}
  -1
  \right\}.
\end{equation*}
As before, denoting $\phi:x\mapsto x - \log x - 1, \ x > 0$, and
$g_\bn(\bomega)$ the piece-wise continuous function that maps any
frequency of $\T^d$ to the closest smaller Fourier frequency
corresponding to the grid $\mathcal{J_\bn}$, we have
\begin{equation*}
  \widetilde{l}_{\bn_{k}}(\bgamma) - \widetilde{l}_{\bn_{k}}(\btheta)
  = (2\pi)^{-d}\int_{\T^d}
  {
    \phi\left(
    \frac{\Inkb{g(\bk)}{\btheta}}{\Inkb{g(\bk)}{\bgamma}}
    \right)
  }d\bk.
\end{equation*}
A Taylor expansion of $\phi(\cdot)$ around 1 gives, with
$\psi(x)=(x-1)^2$,
\begin{equation*}
\label{eq:Taylorexpansionphi}
\phi(x) = \psi(x)(1+\epsilon(x)),
\end{equation*}
where $\epsilon(x)\rightarrow 0$ as $x\rightarrow 1$. Therefore for
any $\delta>0$ there exists $\mu>0$ such that for all $x$ such that
$|x-1|\leq \mu$, $|\epsilon(x)| < \delta$. Now let, for all $k\in\N$,
\begin{equation*}
  \Pi_k = \left\{
  \bk\in \mathcal{T}^d: \left|\frac{\Inkb{g(\bk)}{\btheta}}{\Inkb{g(\bk)}{\bgamma}}-1\right| \leq \mu
  \right\}.
\end{equation*}
We distinguish two cases:
\begin{enumerate}
\item If for some $\delta > 0$, the Lebesgue measure of $\Pi_k$ does not
  converge to $(2\pi)^d$, equation~\eqref{eq:05112018} of the main body
  holds.
\item Otherwise, if for any $\delta>0$ the Lebesgue measure of $\Pi_k$
  does converge to $(2\pi)^d$, we then have 
  \begin{equation*}
    \left|\widetilde{l}_{\bn_{k}}(\bgamma) - \widetilde{l}_{\bn_{k}}(\btheta)\right| 
    = \int_{\Pi_k\cup\Pi_k^C}
    {
      \psi\left(
      \frac{\Inkb{g(\bk)}{\btheta}}{\Inkb{g(\bk)}{\bgamma}}
      \right)\left\{1+\epsilon\left(\frac{\Inkb{g(\bk)}{\btheta}}{\Inkb{g(\bk)}{\bgamma}}
      \right)\right\}
    }d\bk,
  \end{equation*}
where $\Pi_k^C$ denotes the complementary of $\Pi_k$ as a subset of
$\mathcal{T}^d$ and where the function $\epsilon(\cdot)$ was defined in
equation~\eqref{eq:Taylorexpansionphi}. Denoting
$h(\bk;\btheta,\bgamma) =
\frac{\Inkb{g(\bk)}{\btheta}}{\Inkb{g(\bk)}{\bgamma}}$ (note that this
quantity also depends on $k$),
\begin{align*}
\widetilde{l}_{\bn_{k}}(\bgamma) - \widetilde{l}_{\bn_{k}}(\btheta)
&= \int_{\mathcal{T}^d}
  {
    \psi\left(
    h(\bk;\btheta,\bgamma)
    \right)
  }d\bk\\
  &+ 
  \int_{\Pi_k}
      {
	\psi\left(
	h(\bk;\btheta,\bgamma)\right)\epsilon\left(h(\bk;\btheta,\bgamma)
	\right)
      }d\bk\\
      &+
      \int_{\Pi_k^C}
	  {
	    \psi(h(\bk;\btheta,\bgamma)\epsilon(h(\bk;\btheta,\bgamma))d\bk
	  }.
\end{align*}
We shall now show that the two last terms of the right-hand side of this
equation are asymptotically vanishing, so that we can limit our study to the
first term, which will turn out to take a simple form in relation to our
definition of significant correlation contribution (SCC) in the main
body. Given the definition of $\Pi_k$ we have,
\begin{equation*}
\left|
\int_{\Pi_k}
{
\psi\left(
h(\bk;\btheta,\bgamma)\right)\epsilon\left(h(\bk;\btheta,\bgamma)
\right)
}d\bk
\right|
\leq
\delta\int_{\Pi_k}
{
  \psi\left(
  h(\bk;\btheta,\bgamma)\right)
}d\bk
\leq
\delta\int_{\mathcal{T}^d}
  {
    \psi\left(
    h(\bk;\btheta,\bgamma)\right)
  }d\bk,
\end{equation*}
where the two inequalities come from the fact that the function
$\psi(\cdot)$ is non-negative. We also have
\begin{equation*}
\left|
\int_{\Pi_k^C}
{
  \psi(h(\bk;\btheta,\bgamma)\epsilon(h(\bk;\btheta,\bgamma))d\bk
}
\right|
= o(1),
\end{equation*}
since the integrand is upper-bounded given Assumption~1.\ref{Ass:sdf}
and since the measure of the set $\Pi_k^C$ goes to zero. Hence we
obtain, by the triangle inequality,
\begin{equation*}
  \left|
  \widetilde{l}_{\bn_{k}}(\bgamma) - \widetilde{l}_{\bn_{k}}(\btheta)
  \right|
  \geq
  \left(
  \int_{T^d}
      {
	\psi\left(
	h(\bk;\btheta,\bgamma)
	\right)
      }d\bk
      \right)
      \left(1-\delta\right)
      +o(1).
\end{equation*}
We now study the term $(2\pi)^{-d}\int_{\T^d}
{
  \psi\left(
  h(\bk;\btheta,\bgamma)
  \right)
}d\bk=
\nk{k}^{-1} \sum_{\bk\in\Omega_{\bn_k}}\left\{
\frac{\Ink{\btheta}}{\Ink{\bgamma}}-1
\right\}^2.$
We observe that
\begin{eqnarray*}
  \nk{k}^{-1}\sum_{\bk\in\Omega_{\bn_k}}\left\{
  \Ink{\btheta} - \Ink{\bgamma}
  \right\}^2 &=& 
  \nk{k}^{-1}\sum_{\bk\in\Omega_{\bn_k}}\Ink{\bgamma}^2\left\{
  \frac{\Ink{\btheta}}{\Ink{\bgamma}}-1
  \right\}^2\\
  &\leq& \nk{k}^{-1}  f_{\text{max}, \delta}^2 
  \sum_{\bk\in\Omega_{\bn_k}}\left\{
  \frac{\Ink{\btheta}}{\Ink{\bgamma}}-1
  \right\}^2.
\end{eqnarray*}
Additionally, according to Lemma~\ref{lemma:intermediaryresult},
\begin{align*}
\nk{k}^{-1} \sum_{\bk\in\Omega_{\bn_k}}\left\{
	\Ink{\btheta} - \Ink{\bgamma}
	\right\}^2&=
 \sum_{\bu\in\Z^d}\left\{\overline{c}_{\bn_k}(\bu;\btheta)-\overline{c}_{\bn_k}(\bu;\bgamma)\right\}^2 + o(1)\\
	&=\sum_{\bu\in\Z^d}{c_{g,\bn_k}(\bu)^2\left\{c_X(\bu;\btheta)-c_X(\bu;\bgamma)\right\}^2} + o(1)\\
	&\geq \frac{1}{2}\underline{\lim}_{k\rightarrow\infty}S_k(\btheta, \bgamma) + o(1),
\end{align*}
where the last inequality holds for $k$ sufficiently large, given the SCC
assumption, see Definition~\ref{def=significantCorrelation}.  Therefore we
obtain for $k$ sufficiently large,
\begin{equation*}
\left|
\widetilde{l}_{\bn_{k}}(\bgamma) - \widetilde{l}_{\bn_{k}}(\btheta)
\right|
\geq
\frac{1}{2
  f_{\text{max},\delta}^2}\left(1-\delta\right)\underline{\lim}_{k\rightarrow\infty}S_k(\btheta,
\bgamma) + o(1). 
\end{equation*}
Choosing $\delta=1/2$, we obtain the inequality stated in
equation~\eqref{eq:05112018} of the main body. This concludes the proof in
the univariate case, as we have shown the absolute difference of the
expected log-likelihood is lower bounded by the assumption of SCC.
\end{enumerate}

We now extend the proof of Lemma~\ref{lemma=061120182} to the multivariate
case. In the multivariate case, we first observe that we may write the
difference of the expected log--likelihood for different parameter values as
\begin{align}
\nonumber
    \widetilde{\ell}_{\mathbf{n}}(\bm{\gamma})- \widetilde{\ell}_{\mathbf{n}}(\bm{\theta})&=
    |\mathbf{n}|^{-1}\sum_{\bm{\omega}}
    \left\{\log {\mathrm{det}}\{ \overline{\mathbf{I}}(\bm{\omega};\bm{\gamma})\} +{\mathrm{trace}}\left\{\overline{\mathbf{I}}^{-1}(\bm{\omega};\bm{\gamma})\overline{\mathbf{I}}(\bm{\omega};\bm{\theta}\right\}\right\}\\
    \nonumber
    &-|\mathbf{n}|^{-1}\sum_{\bm{\omega}}
    \left\{\log {\mathrm{det}}\{ \overline{\mathbf{I}}(\bm{\omega};\bm{\theta})\} +{\mathrm{trace}}\left\{\overline{\mathbf{I}}^{-1}(\bm{\omega};\bm{\theta})\overline{\mathbf{I}}(\bm{\omega};\bm{\theta}\right\}\right\}\\
    &=|\mathbf{n}|^{-1}\sum_{\bm{\omega}}
    \left\{-\log {\mathrm{det}}\{ \overline{\mathbf{I}}^{-1}(\bm{\omega};\bm{\gamma})\overline{\mathbf{I}}(\bm{\omega};\bm{\theta}\} +{\mathrm{trace}}\left\{\overline{\mathbf{I}}^{-1}(\bm{\omega};\bm{\gamma})\overline{\mathbf{I}}(\bm{\omega};\bm{\theta}\right\}-p\right\}.
\end{align}
We define
$\tilde{\mathbf{B}}_{\bm{\omega}}(\bm{\theta},\bm{\gamma})=\overline{\mathbf{I}}^{-1}(\bm{\omega};\bm{\gamma})\overline{\mathbf{I}}(\bm{\omega};\bm{\theta})$,
and assume this matrix has positive eigenvalues
$\{\tilde{\beta}_i(\omega)\}_{i=1}^p$. Rewriting this expression in terms of
the eigenvalues we get
\begin{align}
\nonumber
    \widetilde{\ell}_{\mathbf{n}}(\bm{\gamma})- \widetilde{\ell}_{\mathbf{n}}(\bm{\theta})&=
    |\mathbf{n}|^{-1}\sum_{\bm{\omega}}
    \left\{-\sum_j\log\tilde{\beta}_j(\omega)+\sum_j\tilde{\beta}_j(\omega)-p
\right\}\\
&=|\mathbf{n}|^{-1}\sum_{\bm{\omega}}
\sum_j \phi\{\tilde{\beta}_j(\bm{\omega})\}.
\end{align}
We define $g_\mathbf{n}(\bm{\omega})$ as the piece-wise continuous function
that maps any frequency of ${\cal T}^d$ to the closest smaller Fourier
frequency corresponding to the grid of ${\cal I}_n,$ we have
\begin{align}
   \widetilde{\ell}_{\mathbf{n}}(\bm{\gamma})- \widetilde{\ell}_{\mathbf{n}}(\bm{\theta})&=(2\pi)^{-d}
   \int_{{\cal T}^d} \sum_{j=1}^p\phi\left( \tilde{\beta}_j(g_\mathbf{n}(\bm{\omega}))\right)\; d  \bm{\omega}.
\end{align}
A Taylor expansion of $\phi(\cdot)$ around 1 gives with
$\phi(x)=\psi(x)(1+\epsilon(x))$ with $\psi(x)=(x-1)^2$, where $\epsilon(x)$
is going to zero as $x\rightarrow 1$. Most of this proceeds exactly like for
the univariate case, but we shall now proceed to study what SCC means in
this context.  Unlike the univariate case we now have to propose a new
approximation that works also in this case. Given we have
\begin{align}
\nonumber
   \widetilde{\ell}_{\mathbf{n}}(\bm{\gamma})- \widetilde{\ell}_{\mathbf{n}}(\bm{\theta})&=\sum_{\bm{\omega}} \sum_j \left( \tilde{\beta}_j(g_\mathbf{n}(\bm{\omega}))-1\right)^2\\
   \nonumber
   \label{eqn:misfit}
   &=\sum_{\bk}{
   \mathrm{trace}
   \left[
   	\tilde{\mathbf{B}}_{\bm{\omega}}(\bm{\theta},\bm{\gamma}) - I_p
   \right]^2
   }\\
   \nonumber
   &=\sum_{\bk}{
   \mathrm{trace}
   \left[
   	\overline{\mathbf{I}}^{-1}
   	(\bm{\omega};\bm{\gamma})
   	\left(
   	\overline{\mathbf{I}}(\bm{\omega};\bm{\theta}) -\overline{\mathbf{I}}(\bm{\omega};\bm{\gamma})
   	\right)
   \right]^2
   }\\
      &\geq f_{\mathrm{max}, \delta}^{-2} \sum_{\bk}{
   \mathrm{trace}
   \left[
   	\overline{\mathbf{I}}(\bm{\omega};\bm{\theta}) - \overline{\mathbf{I}}(\bm{\omega};\bm{\gamma})
   \right]^2
   }\\
   &=\sum_{\bk}{
   	\sum_{q=1}^p{
   		\sum_{r=1}^p{
   			\left|
   				\overline{I}^\pqr(\bk;\btheta)
   				-
   				\overline{I}^\pqr(\bk;\bgamma)
			\right|^2
   		}
   	}
   },
\end{align}
where the inequality results from Lemma~\ref{lemma=traceinequality}.  We can
now relate the above quantity to the multivariate version of SCC via the use
of Parseval's identity, just like we did in the univariate case.\qed
\end{proof}

\subsection*{Proof of Lemma~\ref{lemma=traceinequality}}
\begin{proof}
	Since $H_1$ is Hermitian positive definite it admits $p$ real positive eigenvalues
	$0 < \lambda_1\leq \ldots \leq \lambda_p$ and there exits a unitary matrix $U$
	such that $H_1 = U^* D U$, where $D$ is the diagonal matrix with 
	elements $\lambda_1, \ldots, \lambda_p$ on the diagonal.
	We then have,
	\begin{align*}
		\mathrm{trace}\left[H_1 H_2 \right]^2 & 
		= \mathrm{trace} \left[ U^* D U H_2 U^* D U H_2 \right]
		= \mathrm{trace} \left[ D U H_2 U^* D U H_2 U^* \right]
		=  \mathrm{trace} \left[ D Z D Z \right],
	\end{align*}
	where $Z = U H_2 U^*$, which is Hermitian positive definite just like
	$H_2$ is.
	Therefore,
	\begin{align*}
		\mathrm{trace}\left[H_1 H_2 \right]^2
		&= \sum_{j=1}^p
		{
			\sum_{k=1}^p
			{
				\lambda_{j} Z_{j, k}
				\lambda_{k} Z_{k, j}
			}
		}
		=
		 \sum_{j=1}^p
		{
			\sum_{k=1}^p
			{
				\lambda_{j}
				\lambda_{k}
				 \left|Z_{j, k} \right|^2
			}
		}\\
		&\geq
		\lambda_1^2
		 \sum_{j=1}^p
		{
			\sum_{k=1}^p
			{
				 \left|Z_{j, k} \right|^2
			}
		}
		=
		\lambda_1^2
		 \mathrm{trace}{\ Z}^2
		 = \lambda_1^2
		 \mathrm{trace}{\ H_2}^2.
	\end{align*}
	This concludes the proof.\qed
\end{proof}

\subsection*{Proof of Lemma~\ref{lemma=06112018}}

\begin{proof}
First we observe that for any fixed $\bk\in\T^d$, 
 $\Ink{\bgamma_k}$ converges to $\Ink{\bgamma}$ as
$k$ goes to infinity.
This comes from Assumption 1.\ref{Ass:sdf}, where we have assumed an 
upper-bound on the derivative of the spectral density with respect to
the parameter vector. In that case,
\begin{align*}
\left| \Ink{\bgamma_k} - \Ink{\bgamma} \right| &\leq
\left| (2\pi)^{-d}\int_{\mathcal{T}^d}{
\left\{
  f_{X,\bdelta}(\bk-\bk';\bgamma_k)
	-f_{X,\bdelta}(\bk-\bk';\bgamma)
\right\}
	\mathcal{F}_{\bn}(\bk')d\bk'}
\right|\\
&\leq
(2\pi)^{-d}\int_{\T^d}{
\left|
  f_{X,\bdelta}(\bk-\bk';\bgamma_k)
	-f_{X,\bdelta}(\bk-\bk';\bgamma)
\right|
	\mathcal{F}_{\bn}(\bk')d\bk'}\\
&\leq
(2\pi)^{-d}\int_{\T^d}{
M_{\partial_\theta}\|\bgamma_k-\bgamma\|_2
	\mathcal{F}_{\bn}(\bk')d\bk'}\\
&\leq
M_{\partial_\theta}\|\bgamma_k-\bgamma\|_2
\end{align*}
which converges to zero as $\|\bgamma_k-\bgamma\|_2$ converges to zero by 
assumption.

Now using equation~\eqref{eq:formulaexpectedlikelihood}, we can 
apply the Dominated Convergence Theorem to
$(\widetilde{l}_{\bn_k}(\bgamma_k)-\widetilde{l}_{\bn_k}(\bgamma))_{k\in\N}$,
using the bounds established in Lemma~\ref{lemma:bounds_periodogram},
and the $\bk$-pointwise convergence of
$\left|\overline{I}_{\bn_k}(\bk;\bgamma_k) -
\overline{I}_{\bn_k}(\bk;\bgamma_k)\right|$ to zero. Hence
$(\widetilde{l}_{\bn_k}(\bgamma_k)-\widetilde{l}_{\bn_k}(\bgamma))_{k\in\N}$
converges to zero, which concludes the proof.\qed
\end{proof}

\subsection*{Proof of Lemma~\ref{lemma:cvgtotheta}}
\begin{proof}
Assume, with the intent to reach a contradiction, that $(\bgamma_k)$
does not converge to $\btheta$. By compactness of $\Theta$, there
exists $\bgamma\in\Theta$ distinct from $\btheta$ and
$(\bgamma_{j_k})$ a subsequence of $(\bgamma_k)$ such that
$\bgamma_{j_k}$ converges to $\bgamma$. We then have, using the
inverse triangle inequality,
\begin{align*}
  |\widetilde{l}_{\bn_{j_k}}(\bgamma_{j_k})-\widetilde{l}_{\bn_{j_k}}(\btheta)|
	&\geq  \left| \widetilde{l}_{\bn_{j_k}}(\bgamma)-\widetilde{l}_{\bn_{j_k}}(\btheta)\right|
	- \left|\widetilde{l}_{\bn_{j_k}}(\bgamma_{j_k})-\widetilde{l}_{\bn_{j_k}}(\bgamma)\right|.
\end{align*}
The second term on the right-hand side of the above equation converges
to zero according to Lemma~\ref{lemma=06112018} whereas the first term
is asymptotically lower bounded according to
Lemma~\ref{lemma=061120182}. Therefore the quantity
$|\widetilde{l}_{\bn_{j_k}}(\bgamma_{j_k})-\widetilde{l}_{\bn_{j_k}}(\btheta)|$
is asymptotically lower bounded, which contradicts the initial
assumption that
$\widetilde{l}_{\bn_k}(\bgamma_k)-\widetilde{l}_{\bn_k}(\btheta)$
converges to zero. This concludes the proof, by obtaining a
contradiction.\qed
\end{proof}

\subsection*{Proof of Lemma~\ref{lemma:cvglkh}}
\begin{proof}
We have, for $\bgamma\in\Theta$,
\begin{align*}
  \widetilde{l}_{\bn_k}(\bgamma) - l_{\bn_k}(\bgamma) &=
  |\bn_k|^{-1}\sum_{\bk\in\Omega_\bn}\left\{\log{\Ink{\bgamma}}+
\frac{\Ink{\btheta}}{\Ink{\bgamma}}-\log{\Ink{\bgamma}}-\frac{I_{\bn_k}(\bk)}{\Ink{\bgamma}}\right\}\\ 
  &=|\bn_k|^{-1}\sum_{\bk\in\Omega_\bn}
  {\frac{\Ink{\btheta}-I_{\bn_k}(\bk)}{\Ink{\bgamma}}}.
\end{align*}
In order to show that 
$\widetilde{l}_{\bn_k}(\bgamma) - l_{\bn_k}(\bgamma)$ converges uniformly
in probability to the zero function over $\Theta$,
we need to show that,
\begin{equation}
	\sup_{\bgamma\in\Theta}
	\left|
		\widetilde{l}_{\bn_k}(\bgamma) - l_{\bn_k}(\bgamma)
	\right|
	=
	o_p(1),
\end{equation}
as $k$ goes to infinity.

We first observe that, 
given that the quantity $\Ink{\bgamma}^{-1}$
 is deterministic and upper-bounded
independently of $\bgamma$ by $f_{\text{min}, \delta}^{-1}$, 
we can use Proposition~\ref{prop=varianceLinearCombinations} 
to write that
\revision{
	\begin{equation*}
	\var\left\{\widetilde{l}_{\bn_k}(\bgamma) - l_{\bn_k}(\bgamma)
	\right\} =
	\mathcal{O}\left\{
	\frac{\sum_{\bu\in\Z^d}{c_X(\bu)^2 c_g(\bu)}}
	{\sum{g_\bs^2}}
	\right\}, 
	\end{equation*}
	where the big $\mathcal{O}$ does not depend on $\bgamma$. 
	Thus using Chebychev's inequality
	\[	\widetilde{l}_{\bn_k}(\bgamma) - l_{\bn_k}(\bgamma) 
	=\mathcal{O}_P\left\{
	\left(
	\frac{\sum_{\bu\in\Z^d}{c_X(\bu)^2 c_g(\bu)}}
	{\sum{g_\bs^2}}
	\right)^{1/2}
	\right\}
	.
	\]

This holds for any fixed $\bgamma\in\Theta$. In order to establish uniform 
convergence in probability we shall also use smoothness properties
of the expected periodogram.
Let $\epsilon>0$ and $\eta>0$. 
Define the events,
\begin{equation*}
	A_k = 
	\left(
		\sup_{\bgamma\in\Theta}
		\left|
		\widetilde{l}_{\bn_k}(\bgamma) - l_{\bn_k}(\bgamma)
		\right|
		\geq \epsilon
	\right), \quad\forall k\in\N.
\end{equation*}
We wish to show that there exists $k_A\in\N$ such that for all
integer $k\geq k_A, P(A_k)\leq\eta$.
We note that,
\begin{equation*}
A_k = 
\bigcup_{\bgamma\in\Theta}
\left(
\left|
\widetilde{l}_{\bn_k}(\bgamma) - l_{\bn_k}(\bgamma)
\right|
\geq \epsilon
\right), \quad\forall k\in\N.
\end{equation*}
Indeed, inclusion $\supset$ is obvious, whereas 
inclusion $\subset$ follows from the sup being reached
due to the continuity w.r.t $\bgamma$ and the compacity of
$\Theta$.
Let
\begin{equation*}\Delta_{\bn_k}(\bgamma, \bgamma')= 
\widetilde{l}_{\bn_k}(\bgamma) - l_{\bn_k}(\bgamma)
-
(\widetilde{l}_{\bn_k}(\bgamma') - l_{\bn_k}(\bgamma')).
\end{equation*}
We have, by Taylor-expension, 
\begin{align*}
	\Delta_{\bn_k}(\bgamma, \bgamma')
	&=
	|\bn_k|^{-1}\sum_{\bk\in\Omega_\bn}
	{
	\left(
	\frac{1}{\Ink{\bgamma}}
	-
	\frac{1}{\Ink{\bgamma'}}\right)
	\left(\Ink{\btheta}-I_{\bn_k}(\bk)\right)
	}	
	\\
	&=
	|\bn_k|^{-1}\sum_{\bk\in\Omega_\bn}
	\left\{
	\frac{1}{\Ink{\bgamma}^2}\left(\bgamma'-\bgamma\right)^T
	\nabla_\theta\Ink{\widetilde{\bgamma}_{\bk}}
	\left(\Ink{\btheta}-I_{\bn_k}(\bk)\right)
	\right\},
	\\
\end{align*}
where each $\widetilde{\bgamma}_{\bk}$ depends on $\bk$.
Hence, by the triangle inequality,
\begin{align*}
\left|
\Delta_{\bn_k}(\bgamma, \bgamma')
\right|
&
\leq
|\bn_k|^{-1}
\left|
\sum_{\bk\in\Omega_\bn}
{
\frac{1}{\Ink{\bgamma}^2}
\left(\bgamma'-\bgamma\right)^T
\nabla_\theta\Ink{\widetilde{\bgamma}_{\bk}}
I_{\bn_k}(\bk)
}
\right|
\\
&+
|\bn_k|^{-1}
\left|\sum_{\bk\in\Omega_\bn}
{
	\frac{1}{\Ink{\bgamma}^2}
	\left(\bgamma'-\bgamma\right)^T
	\nabla_\theta\Ink{\widetilde{\bgamma}_{\bk}}
	\Ink{\btheta}
}
\right|.
\end{align*}
Using the upper-bound for the norm of the
derivative of the spectral density with respect to the parameter vector,
as well as the lower bound for the spectral density,
we obtain,
\begin{align*}
\left|
\Delta_{\bn_k}(\bgamma, \bgamma')
\right|
&\leq
|\bn_k|^{-1}
f_{\delta, \text{min}}^{-2}M_{\partial\theta}\|\gamma' - \gamma\|
\sum_{\bk\in\Omega_\bn}
\left\{
	I_{\bn_k}(\bk)+ \Ink{\btheta}
\right\}
\\
&=
|\bn_k|^{-1}
f_{\delta, \text{min}}^{-2}M_{\partial\theta}\|\gamma' - \gamma\|
\left(
2\sum_{\bk\in\Omega_\bn}
{
	\Ink{\btheta}
}
+ o_P(1)
\right),
\end{align*}
according to Proposition~\ref{prop=varianceLinearCombinations},
and using SCC.
This implies that we can choose $\delta>0$ small enough such
that there exists a natural integer $k_C$ such that,
\begin{equation*}
	\forall k\geq k_C, \ \forall \bgamma, \bgamma'\in\Theta, \
	\|\bgamma' - \bgamma\| \leq \delta \Longrightarrow
	P\left(
	\left|
	\Delta_{\bn_k}(\bgamma, \bgamma')
	\right|
	 \geq \frac{\epsilon}{2}\right) 
	\leq 
	\frac{\eta}{2.}.
\end{equation*}
Now, 
let $\left\{\balpha_j^\delta\right\}_{j=1, \ldots, J}$ be a finite family
of elements of $\Theta$ such that,
\begin{equation*}
	\bigcup_{j=1}^J B(\balpha_j^\delta, \delta) \supset \Theta,
\end{equation*}
with $B(\balpha_j^\delta, \delta)$ denoting the
ball centered on $\balpha_j^\delta$ with radius $\delta$.
Existence here follows from the compacity of $\Theta$, and the 
positiveness of $\delta$.
Define the events,
\begin{equation*}
B_k = 
\bigcup_{j=1}^J
\left(
|\bn_k|^{-1}
\left|
\sum_{\bk\in\Omega_\bn}
{\frac{\Ink{\btheta}-I_{\bn_k}(\bk)}{\Ink{\balpha_j}}}
\right|
\geq \frac{\epsilon}{2}
\right), \quad\forall k\in\N,
\end{equation*}
and $C_k=A_k\setminus B_k$.
Clearly $B_k\subset A_k$ so that $A_k = B_k \cup C_k$, and 
therefore $P(A_k) \leq P(B_k) + P(C_k)$.
Again by Proposition~\ref{prop=varianceLinearCombinations},
and because $J$ is finite,
there exists $k_B$ such that for any integer
$k\geq k_B$, $P(B_k)\leq\frac{\eta}{2}$.
Finally, for an outcome in $C_k$, there exists $\bgamma'\in\Theta$
such that 
$\left|\widetilde{l}_{\bn_k}(\bgamma') - l_{\bn_k}(\bgamma')\right|
 \geq \epsilon$. By construction, there exists
 $j\in\left\{1,\ldots, J\right\}$ such that
 $\|\balpha_j - \bgamma'\|\leq \delta$, but at the same time
 we have 
 $
 \left|\widetilde{l}_{\bn_k}(\balpha_j) - l_{\bn_k}(\balpha_j)\right|
 \leq \frac{\epsilon}{2}.
 $
 By inverse triangle inequality, we therefore have,
 $\Delta_{\bn}(\balpha_j, \gamma')\geq \frac{\epsilon}{2}$.
 Hence for integer $k\geq k_C$, $P(C_k)\leq \frac{\eta}{2}$.
 We conclude that, with $k_A=\max(k_B, k_c)$,
  for $k\geq k_A, P(A_k)\leq \eta$. Since this can be 
 achieved for any choice of $\eta$, this concludes the proof.

The extension to univariate non-Gaussian random fields follows from
Corollary~\ref{cor=varianceLinearCombinations}.
Similarily, for a Gaussian multivariate random field,
\begin{align*}
\widetilde{l}_{\bn_k}(\bgamma) - l_{\bn_k}(\bgamma) &=
|\bn_k|^{-1}\sum_{\bk\in\Omega_\bn}\left\{
\mathrm{trace}\left[\Ink{\btheta}\Ink{\bgamma}^{-1}\right] - J^H(\bk) {\Ink{\bgamma}}^{-1} J(\bk)\right\}
\end{align*}
and we use Corollary~\ref{cor=varianceSesLinearCombinations}.
}\qed
\end{proof}

\subsection*{Proof of Lemma~\ref{proposition=upperboundnormcovmat}}

\begin{proof}
The proof is adapted from the one-dimensional case,
see~\cite{guillaumin2017analysis} and \cite{sykulski2019debiased}. We
first define the following isomorphism from $\prod_{i=1}^d\{1,\ldots, n_i\}$
to $\{1,\ldots,|\bn|\}$, that will be used for a change of variable:
\begin{equation*}
	j(j_1, \ldots, j_d) = \sum_{k=1}^d\left\{(j_k-1)\prod_{j=1}^{k-1}{n_j}\right\},
\end{equation*}
and $j_1(j), \ldots, j_d(j)$ the component functions of its
inverse. This isomorphism gives the index in the column vector $\bX$
of the observation at location $(j_1, \ldots, j_d)$ on the grid, given
our choice of ordering.

Let $\balpha$ be any complex-valued vector of $\C^n$, and denote $\balpha^*$ its Hermitian transpose. 
We then have, using the above isomorphism for a change of variables,
\begin{align*}
  \balpha^* \Cx \balpha &= \sum_{j,k=1}^{|\bn|}{\balpha_{j}^*(C_{\bX})_{j,k} \balpha_{k}}\\
  &= \sum_{j_1=0}^{n_1-1}\ldots \sum_{j_d=1}^{n_d-1}{
    \sum_{k_1=0}^{n_1-1}\ldots \sum_{k_d=1}^{n_d-1}{
      \balpha_{j(j_1,\ldots, j_d)}^*(C_{\bX})_{j(j_1,\ldots, j_d),k(k_1,\ldots, k_d)}\balpha_{k(k_1,\ldots, k_d)}.
    }
  }
\end{align*}
Here we use the fact that
\begin{equation*}
  (C_{\bX})_{j(j_1,\ldots, j_d),k(k_1,\ldots, k_d)} = c_\bX(k_1-j_1, \ldots, k_d-j_d),
\end{equation*}
so that
\begin{align*}
  \balpha^* \Cx \balpha &= \sum_{j_1=0}^{n_1-1}\ldots \sum_{j_d=1}^{n_d-1}{
    \sum_{k_1=0}^{n_1-1}\ldots \sum_{k_d=1}^{n_d-1}{
      \balpha_{j(j_1,\ldots, j_d)}^*\balpha_{k(k_1,\ldots, k_d)} \int_{\T^d}{
	f_{X,\delta}(\bomega)e^{i\left((k_1-j_1)\omega_1+\ldots+(k_d-j_d)\omega_d\right)}
      }d\bk
    }
  }\\
  &= 
  \int_{\T^d}{
    f_{X,\delta}(\bomega)
    \sum_{j_1=0}^{n_1-1}\ldots \sum_{j_d=1}^{n_d-1}{
      \sum_{k_1=0}^{n_1-1}\ldots \sum_{k_d=1}^{n_d-1}{
	\balpha_{j(j_1,\ldots, j_d)}^*\balpha_{k(k_1,\ldots, k_d)}
	e^{i\left((k_1-j_1)\omega_1+\ldots+(k_d-j_d)\omega_d\right)}
    }}
  }d\bk\\
  &=
  \int_{\T^d}{
    f_{X,\delta}(\bomega)
    \left|
    \sum_{j_1=0}^{n_1-1}\ldots \sum_{j_d=1}^{n_d-1}{
      \balpha_{j(j_1,\ldots, j_d)}
      e^{i\left(j_1\omega_1+\ldots+j_d\omega_d\right)}
    }
    \right|^2
  }d\bk\\
  &\leq f_{\delta,\text{max}}
  \int_{\T^d}{
    \left|
    \sum_{j_1=0}^{n_1-1}\ldots \sum_{j_d=1}^{n_d-1}{
      \balpha_{j(j_1,\ldots, j_d)}
      e^{i\left(j_1\omega_1+\ldots+j_d\omega_d\right)}
    }
    \right|^2
  }d\bk.
\end{align*}
By Parseval's equality, we obtain,
\begin{equation*}
  0\leq \balpha^* \Cx \balpha \leq f_{\delta,\text{max}} \|\balpha\|_2^2,
\end{equation*}
where $\|\balpha\|_2$ is the $l_2$ vector norm of the vector $\balpha$.
This concludes the proof of the upper bound.
The lower bound can be derived in the same way, which concludes the
proof.\qed
\end{proof}

\subsection*{Proof of
  Proposition~\ref{prop=normalityBoundedLinearCombinationsPeriodogram}} 
\label{proof=normalityBoundedLinearCombinationsPeriodogram}

\begin{proof}
We only treat the scenario where $g_\bs = 1,
\forall\bs\in\mathcal{J}_{n_k}$, i.e., we do not consider the situation of
missing observations for this proposition. The proof is adapted
from~\citet[p.~217]{grenander1958toeplitz}. We write it for the 
case of $\Omega_{\bn}=\Omega_{\bn}^{(1)}$, the case
$\Omega_{\bn}=\Omega_{\bn}^{(2)}$ being almost identical. Define
\begin{equation*}
L_k = |\bn_k|^{-1}\sum_{\bk\in\Omega_{\bn_k}}
{
  w_k(\bk)I_{\bn_k}(\bk),
}
\end{equation*}
as a weighted sum of periodogram values, and $U_{\bn_k}$ the
multi-dimensional Fourier matrix corresponding to $\Jx$. We have
\begin{equation*}
L_k = |\bn_k|^{-1}\bX^* U_{\bn_k}^*\text{diag}(w_k(\bk_0), \ldots, w_k(\bk_{|\bn_k|-1}))  
U_{\bn_k} \bX.
\end{equation*}
Writing $W_k = |\bn_k|^{-1} U_{\bn_k}^*\text{diag}(w_k(\bk_0),
\ldots, w_k(\bk_{|\bn_k|-1})) U_{\bn_k}$, we then have
\begin{equation*}
L_k = \bX^*  W_k \bX,
\end{equation*}
which we regard as a quadratic form in the vector $\bX$.
Following~\citet[p.~134]{Cramer1946}, in particular his formula 11.12.2, 
the characteristic
function of the random variable $L_k$ therefore takes the form of
{\small \begin{align*}
\phi_{L_k}(\alpha) &= \E\left\{\exp(i\alpha L_k)\right\}\\ 
& = (2\pi)^{-n/2}
\left|C_X(\btheta)\right|^{-1/2}
\int_{-\infty}^\infty{\ldots
  \int_{-\infty}^\infty{
    \exp\left\{-x^*\left(-i\alpha  W_k 
    +\frac{1}{2}
    C^{-1}_X(\btheta)
    \right)
    x
    \right\}dx_1\ldots dx_n
  }
},
\end{align*}}
where for a square matrix $A$, $|A|$ denotes its determinant.  Using a
known result~\citep{matrixanalysis} for complex-valued symmetric
matrices, there exists a diagonal matrix $D_k$ and a unitary matrix
$V_k$ such that
\begin{equation}
  \label{eq:100620191609}
  -i\alpha W_k
  +\frac{1}{2}
  C^{-1}_X(\btheta) = VD_kV^T.
\end{equation}
By posing the change of variables $y = V^Tx$ we obtain,
\begin{equation*}
\phi_{L_k}(\alpha) = 
(2\pi)^{-n/2}
\left|C_X(\btheta)\right|^{-1/2}
\prod_{j=1}^n{
  \int_{-\infty}^\infty{
    \exp\left\{
    -y^2 d_{j,k}
    \right\}	
    dy,
  }
}
\end{equation*}
where the $d_{j,k}, j=1,\ldots,n$ are the complex-valued elements of
the diagonal matrix $D_k$ from equation~\eqref{eq:100620191609}, and
where we remind the reader that $|V|=1$ since $V$ is unitary.  As we
recognize integrals of the form $\int_{-\infty}^\infty{\exp(-y^2)dy}$
we obtain,
\begin{align*}
\phi_{L_k}(\alpha) &= 
2^{-n/2}
\left|C_X(\btheta)\right|^{-1/2} 
\left|
-i\alpha  W_k
+\frac{1}{2}
C_X(\btheta)^{-1}
\right|^{-1/2}\\
&=
\left|
-2i\alpha C_X(\btheta)W_k
+
I_{\n}
\right|^{-1/2}
\end{align*}
Hence,
\begin{equation*}
\log \phi_{L_k}(\alpha) = -\frac{1}{2}\log \left|I_{|\bn_k|}-2i\alpha C_X(\btheta) 
W_k \right|.
\end{equation*}
Denoting with $\nu_{1,k},\ldots,\nu_{|\bn_k|,k}$ the eigenvalues of 
$C_X(\btheta)  W_k$, we therefore have
\begin{equation*}
\log \phi_{L_k}(\alpha) = -\frac{1}{2}\sum_{j=1}^{|\bn_k|}{
  \log\left(1 - 2i\alpha\nu_{j,k}\right).
}
\end{equation*}
According to Proposition~\ref{proposition=upperboundnormcovmat} the
spectral norm of $C_X$, the covariance matrix of $\bX$, is
upper-bounded by $f_{\text{max},\delta}$. The spectral norm of $W_k$
is clearly upper-bounded by $|\bn_k|^{-1} M_W$, as from the
definition of $W_k$ its eigenvalues are exactly
\begin{equation*}
	|\bn_k|^{-1} w_k(\bk_0), |\bn_k|^{-1} w_k(\bk_1), \ldots, 
|\bn_k|^{-1} w_k(\bk_{|\bn_k|-1}).
\end{equation*}
By property of the spectral norm of a product of matrices, we obtain,
\begin{equation}
\label{eq:1006192035}
|\bn_k|^{-1} m_W f_{\text{min},\delta} \leq |\nu_{j,k}| \leq
|\bn_k|^{-1} M_W f_{\text{max},\delta}, \ \ \forall
j=1,\ldots,|\bn_k|, k\in\N. 
\end{equation}
The variance of $L_k$ is given by
\begin{equation*}
	\sigma_k^2 = \var\left\{L_k\right\} = 
	2\sum_{j=1}^{|\bn_k|}{\nu_{j,k}^2},
\end{equation*}
and therefore satisfies
\begin{equation}
\label{eq:1006192036}
2|\bn_k|^{-1} (m_W f_{\text{min}})^2 \leq \sigma_k^2 \leq 2
|\bn_k|^{-1} (M_W f_{\text{max},\delta})^2. 
\end{equation}
We also observe that
\begin{equation*}
\frac{\nu_{j,k}}{\sigma_k}\rightarrow 0, \ (k\rightarrow\infty),
\end{equation*}
uniformly, given the bounds determined in
equations~\eqref{eq:1006192035} and~\eqref{eq:1006192036}. Denote
$\Ll_k$ the standardized quantity $\left(L_k - \E\{L_k\}\right) /
\sigma_k$. After Taylor expansion of the logarithm terms to third
order, its characteristic function takes the form of
\begin{align}
  \nonumber\log \phi_{\Ll_k}(\alpha) &= -\frac{1}{2}\sum_{j=1}^{|\bn_k|}{
    \log\left(1 - \frac{2i\alpha\nu_{j,k}}{\sigma_k}\right)
  } 
  - i\frac{\alpha \sum_{j=1}^{|\bn_k|}{\nu_{j,k}}}{\sigma_k}\\
  & = -\frac{1}{2}\alpha^2 + 
  \sum_{j=1}^{|\bn_k|}\left[
    \frac{4}{3}\left(\frac{i\alpha\nu_{j,k}}{\sigma_k}\right)^3
    +o\left\{\left(\frac{i\alpha\nu_{j,k}}{\sigma_k}\right)^3\right\}
    \right]
  \label{eq:1006192039},
\end{align}
where the small o is uniform and is denoted $\epsilon_k$ in what
follows, to make it clear that it does not depend on $j$. The second
term in equation~\eqref{eq:1006192039} can be shown to become
negligible as $k$ goes to infinity, since
\begin{align*}
\left|
\sum_{j=1}^{|\bn_k|}\left[
  \frac{4}{3}\left(\frac{i\alpha\nu_{j,k}}{\sigma_k}\right)^3
  +o\left\{\left(\frac{i\alpha\nu_{j,k}}{\sigma_k}\right)^3\right\}
  \right]
\right|
&\leq \alpha^3 \sigma_k^{-3}\left(\frac{4}{3}+\epsilon_k\right)
\sum_{j=1}^{|\bn_k|}
    {
      \left|\nu_{j,k}\right|^3
    }\\
    &\leq 
    \alpha^3 \left(\frac{4}{3}+\epsilon_k\right)
    \frac{|\bn_k|^{-2}M_W^3 f_{\text{max}}^3}
	 {|\bn_k|^{-3/2}m_W^3 f_{\text{min}}^3}\\
	 &=\mathcal{O}(|\bn_k|^{-1/2}).
\end{align*}
We conclude that $\phi_{\Ll_k}(\alpha)$ converges to
$\exp(-\frac{1}{2}\alpha)$, and therefore $L_k$ is asymptotically
standard normally distributed after appropriate normalization.\qed
\end{proof}

\subsection*{Proof of Theorem~\ref{th=asymptNormality}}
\begin{proof}
Direct calculations show that the gradient of our quasi-likelihood function
at the true parameter
vector is given by,
\begin{equation}
\label{eq:gradientoflkh}
  \nabla_\theta l_{\bn_k}(\btheta) = 
  |\bn_k|^{-1}\sum_{\bk\in\Omega_{\bn_k}}{
    \overline{I}_{\bn_k}(\bk;\btheta)^{-2}
    \left(
    \overline{I}_{\bn_k}(\bk;\btheta)
    -
    I(\bk)
    \right)
    \nabla_\theta \overline{I}_{\bn_k}(\bk;\btheta)
  }.
\end{equation}
By expanding this gradient function at the
true parameter value, and noting that
$\nabla_\theta l_{\bn_k}(\bk;\hat{\btheta})=0$ by definition of $\hat{\btheta}$
and given
Assumption~\ref{assumption:rate}.\ref{ass:interior}, we obtain
\begin{equation*}
  \nabla_\theta l_{\bn_k}(\bk;\btheta) = 
  H(\btheta_k')(\btheta-\hat{\btheta}_k),
\end{equation*}
where $H(\cdot)$ is the Hessian of $l_{\bn_k}(\cdot)$ and $\btheta_k'$
is a parameter vector that converges in probability to the true
parameter vector, since $\hat{\btheta}_k$ is consistent as per
Theorem~\ref{th:consistency}. Therefore,
\begin{equation}
  \label{eq:21001405}
  \hat{\btheta}_k - \btheta = -H^{-1}(\btheta_k')
  \nabla_\theta l_{\bn_k}(\bk;\btheta).
\end{equation}
We now study the expected Hessian of the likelihood
function taken at the true parameter vector, 
$\mathcal{H}(\btheta)$.
Direct calculations lead to
\begin{equation*}
  \mathcal{H}(\btheta) = |\bn_k|^{-1}\sum_{\bk\in\Omega_{\bn_k}}{
    \overline{I}_{\bn_k}(\bk;\btheta)^{-2}
    \nabla_\theta \overline{I}_{\bn_k}(\bk;\btheta)
    \nabla_\theta \overline{I}_{\bn_k}(\bk;\btheta)^T.
  }
\end{equation*}

It can be shown, see \citet[p.~17 of their supplementary
  document]{sykulski2019debiased} for instance, that in
equation~\eqref{eq:21001405} the quantity $H(\btheta_k')$ satisfies, if
Assumption~\ref{assumption:rate}.\ref{ass:rate2} holds,
\begin{equation*}
H(\btheta_k') = \mathcal{H}(\btheta) + \mathcal{O}_P(r_k) + o_P(1).
\end{equation*}
Hence we have, asymptotically,
\begin{equation}
\label{eq:12321806}
H^{-1}(\btheta_k') = \mathcal{H}^{-1}(\btheta) + o_P(1).
\end{equation}
Since equation~\eqref{eq:gradientoflkh} follows the conditions required for
Proposition~\ref{prop=varianceLinearCombinations} to apply, the gradient at
the true parameter vector $\nabla_\theta l_{\bn_k}(\bk;\btheta)$ is itself
$\mathcal{O}_P(r_k)$.  Further more, Lemma~\ref{lemma:mineigenvalueHessian}
tells us that the minimum eigenvalue of $\mathcal{H}$ is lower-bounded by
$S(\btheta)$, independently of $k$. We finally obtain the stated result,
\begin{equation*}
\hat{\btheta}_k - \btheta= \mathcal{O}_P(r_k).
\end{equation*}
In the case of a sequence of full grids, 
$|\bn_k|^{1/2}\nabla_\theta l_{\bn_k}(\bk;\btheta)$
is additionally shown to follow a standard normal distribution via
Proposition~\ref{prop=normalityBoundedLinearCombinationsPeriodogram}, and
we conclude to the asymptotic normality of our estimator.\qed
\end{proof}

\subsection*{Definitions, notation and lemmatas required for the proof of Proposition~\ref{prop:asympNormalNonGaussian}}
	First we introduce some notation for cumulants and remind 
	the reader about 
	their basic properties.
	For integer $L\geq1$ and random variables $Y_1, \ldots, Y_L$,
	all having finite $L$-th order moments,
	the cumulant of $Y_1, \ldots, Y_L$ is defined by,
	\begin{equation*}
	\cum\left[Y_1, \ldots, Y_L\right]
	= \sum_{\nu\in \mathcal{P}\left\{1, \ldots, L\right\}}{
		(-1)^{\#\nu - 1}\left(\#\nu - 1\right)! \prod_{\mathcal{S}\in\nu}\E\left[\prod_{j\in\mathcal{S}}Y_j\right]
	},
	\end{equation*}
	where $ \mathcal{P}\left\{1, \ldots, L\right\}$ denotes the set of partitions of $\left\{1, \ldots, L\right\}$, 
	and $\#\nu$ denotes the cardinality of the partition $\nu$, i.e. the number of sets it contains.
	The cases $L=1$ and $L=2$ correspond to expectation and covariance respectively. Higher-order
	cumulants vanish for multivariate normal $Y_1, \ldots, Y_L$.
	For a given random variable $Y$, we denote $\cum_L\left\{ Y \right\}$ its
	$L$-th order cumulant,
	i.e. $\cum_L\left\{ Y \right\} = \cum(Y, \ldots, Y)$ with $Y$ repeated 
	$L$ times. In our proof we shall make use of the two following lemmas, that
	can be found in \cite{brillinger2001time}.
	\begin{lemma}[Basic properties of cumulants]
		\label{lemma:cumulantsbasics}
		Let $L$ be a positive integer, $Z$, $Y_1, \ldots Y_{L}$ be
		random variables all having finite $L$-th order
		moments, and $a\in\R$. We have the following properties;
		\begin{enumerate}
			\item \emph{Symmetry.} The cumulant $\cum\left\{
			Y_1, \ldots, Y_{L}
			\right\}
			$ does not depend on the order of the variables.
			\item \emph{Multi-linearity.} The cumulant is linear with 
			respect to each of its variables, i.e.
			\begin{equation*}
			\cum\left\{
			aZ + Y_1, Y_2, \ldots, Y_L
			\right\}
			=
			a \ \cum\left\{
			Z, Y_2 \ldots, Y_L
			\right\}
			+
			\cum\left\{
			Y_1, \ldots, Y_L
			\right\}.
			\end{equation*}
		\end{enumerate}
	\end{lemma}
	\begin{lemma}[Cumulant of products of random variables]
		\label{lemma:cumofproducts}
		Let $L$ be a positive integer. \\ Let $Y_1, \ldots, Y_{2L}$ be random variables, all having finite $2L$-th order moments. We have,
		\begin{equation*}
		\cum\left[
		Y_{1}Y_{L+1}, Y_{2}Y_{L+2}, \ldots, Y_{L}Y_{2L}
		\right]
		=	\sum_{\nu}
		\cum
		\left[
		Y_{j}: j\in\nu_1
		\right]
		\ldots
		\cum
		\left[
		Y_{j}: j\in\nu_p
		\right],
		\end{equation*}
		where the left-hand side is the cumulant of $L$ products of
		pairs of random variables, and
		where the summation on the right-hand side
		 is over indecomposable---as defined by~\cite{brillinger2001time}---partitions $\nu=(\nu_1, \ldots,\nu_p)$ of the $L\times 2$ table below.
		\begin{align}
		\label{eq:table}
		\nonumber 1 &\ \  & L+1 \\
		\nonumber 2 &\ \  & L+2 \\
		\ldots& \  \ & \ldots \\
		\nonumber L-1 &  \ \ & 2L-1 \\
		\nonumber L &  \ \ & 2L
		\end{align}
		A partition $\nu$ of the above $L\times 2$ table is indecomposable if and only if any two elements of the table can be joined by a path where two consecutive elements on said path are either within a same set $\mathcal{S}\in\nu$ or on 
		the same row. We give an example of an indecomposable partition in
		Figure~\ref{fig:indecomposable}.
	\end{lemma}

	\begin{figure}
		\centering
		\includegraphics[width=0.25\linewidth]{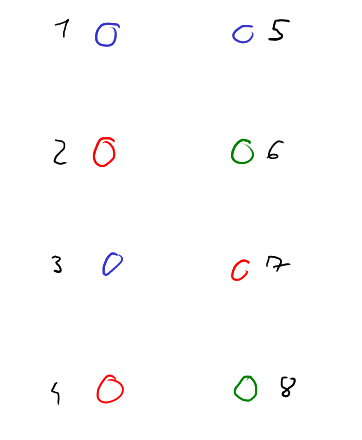}
		\caption{Example of an indecomposable partition
			of the $L\times 2$ table~\eqref{eq:table} in the case $L=4$, with sets of the
			partition indicated by colors red, green and blue.
			The partition is indecomposable because any
			 two elements of the table can be joined by a path where two consecutive elements on said path are either within a same set $\mathcal{S}\in\nu$ or on 
			the same row. For instance, here, such a path between $1$ and $8$ is
			$1 \rightarrow 3 \rightarrow 7 \rightarrow 4 \rightarrow 8$.}
		\label{fig:indecomposable}
	\end{figure}

	To establish the proof of Proposition~\ref{prop:asympNormalNonGaussian}
	we shall follow the line of proof from~\cite{brillinger2001time}
	for the analysis of time series. We introduce some
	additional notation in order to extend to random fields of
	any dimensionality $d$.
	Let $\Delta_\bn(\bk)$ denote the $d$-dimensional separable
	Dirichlet kernel, i.e.,
	\begin{equation*}
	\Delta_\bn(\bk) = \prod_{j=1}^d{
		\left(
		\sum_{t=0}^{n_j-1}e^{i \omega_j t}
		\right)
		=\prod_{j=1}^d{\Delta_{n_j}(\omega_j)}
		,
	}
	\end{equation*}
	where for a positive integer $n$, and scalar $\omega\in\R$,
	\begin{equation*}
	\Delta_n(\omega) = 
	\sum_{t=0}^{n-1}e^{i \omega t},
	\end{equation*}
	is the usual Dirichlet kernel.
	We define, for $0\leq q \leq d$, and 
	for any $l_1, \ldots, l_q\in\{1, \ldots, d\}$,
	\begin{equation}
	\label{eq:deltak}
	\Delta_\bn^{(l_1, \ldots, l_q)}(\bk) = 
	\prod_{\substack{j=1\\j\neq l_1, \ldots, l_q}}^d{
		\Delta_{n_j}(\omega_j).
	}
	\end{equation}
	Note that when $\bk\in\Omega_{\bn}$, i.e. is a Fourier frequency,
	$\Delta_\bn(\bk) = \Delta_\bn^{(k)}(\bk) = 0$ except if
	$\omega_j \equiv 0 \ [2\pi]$, $j=1,\ldots,d$ (where we write
	$a\equiv b \ [c]$ for real numbers $a,b,c$ if there exists an
	integer $k$ such that $a - b = kc$) in which case $\Delta(\bk) = |\bn|$,
	and except if 
	$\omega_j \equiv 0 \ [2\pi]$, $j=1,\ldots,d, \ j\neq k$
	in which case
	$\Delta_\bn^{(k)}(\bk) = \prod_{\substack{j=1\\j\neq k}}^d{n_j}$.
	
	Let $\widetilde{J}_\bn(\bomega) = |\bn|^{\frac{1}{2}}J_\bn(\bomega)$.
	The following lemma is an adaptation of~\citet[Lemma P4.1]{brillinger2001time} to higher dimensions.
	
	\begin{lemma}
		\label{lemma:ineqDelta}
		Let $d\geq 1$ be an integer. Let $\bn\in\left(\N\backslash\{0\}\right)^d$, 
		$\bu\in\N^d$ and 
		$\blambda\in\R^d$. Let $\left\{g_{\bs}\right\}_{\bs\in\Z^d}$ take value $1$ for $\bs\in\mathcal{J}_\bn$,
		and value $0$ otherwise.
		We have the following inequality,
		\begin{equation*}
		\left|
		\sum_{\bs\in\Z^{d}}
		{
			g_{\bs}g_{\bs+\bu}
			e^{-i\bs\cdot\blambda}
		}
		-
		\Delta_{\bn}(\blambda)
		\right|
		\leq
		\sum_{j=1}^{d}{u_j\left|\Delta_\bn^{(j)}(\blambda)\right|}
		+
		\sum_{\substack{j,k=1 \\ k > j}}^{d}{u_j u_k
			\left|\Delta_\bn^{(j, k)}(\blambda)\right|}
		+
		\ldots
		+u_1\cdots u_{d}\left|\Delta_\bn^{(1, \ldots, d)}(\blambda)\right|,
		\end{equation*}
	where we note that $\Delta_\bn^{(1, \ldots, d)}(\blambda)=1, \
	 \forall\blambda\in\R^d$.
	\end{lemma}
	\begin{proof}
		We write,
		\begin{equation*}
		\left|
		\sum_{\bs\in\Z^{d}}
		{
			g_{\bs}g_{\bs+\bu}
			e^{-i\bs\cdot\blambda}
		}
		-
		\Delta_{\bn}(\blambda)
		\right|
		=
		\left|
		\sum_{\bs\in\Z^d}
		g_{\bs}g_{\bs+\bu}
		e^{-i\blambda \cdot \bs}
		- \sum_{\bs\in\Z^d}
		g_{\bs}
		e^{-i\blambda \cdot \bs}
		\right|
		=
		\left|
		\sum_{\bs\in\Z^d}
		g_{\bs}
		(1 - g_{\bs+\bu})
		e^{-i\blambda \cdot \bs}
		\right|.
		\end{equation*}
		We first consider the cases $d=1$ and $d=2$ as examples, before proving 
		the result for any dimensionality $d\geq 1$ by induction.
		For $d=1$, we have $u\in\N$  and, by applying the triangle inequality,
		\begin{align*}
		\left|
		\sum_{s}
		g_{s}(1 - g_{s+u})
		e^{-i\lambda s}
		\right|
		\leq
		\sum_{s}
		\left|g_{s}
		(1 - g_{s+u})
		e^{-i\lambda s}
		\right|
		=\sum_{s}
		|g_{s}|
		\left|
		(1 - g_{s+u})
		\right|
		= u.
		\end{align*}
		The last equality holds because each term $|g_{s}|
		\left|
		(1 - g_{s+u})
		\right|$ is non-zero if and only if $s$ is a point on the grid
		but $s+u$ is not, which occurs for a total number of
		$u$ locations --- more specifically for $s\in\Z$ such
		that
		$n_1-u\leq s \leq n_1 - 1$.
		In dimension $d=2$, we \emph{split} the problem along both
		dimensions.
		\begin{align*}
		&\left|
		\sum_{\bs\in\Z^2}
		g_{\bs}g_{\bs+\bu}
		e^{-i\blambda \cdot\bs}
		-
		\Delta_n\left(\blambda\right)
		\right|
		= \left|
		\sum_{\bs}
		g_{\bs}
		(1 - g_{\bs+\bu})
		e^{-i\blambda \cdot \bs}
		\right|
		\\
		&=
		\left|
		\sum_{s_1=0}^{n_1-1}{
			\sum_{s_2=0}^{n_2-1}
			(1 - 
			g_{s_1+u_1, s_2+u_2})
			e^{-i(\lambda_1 s_1 + \lambda_2 s_2)}
		}
		\right|
		\\
		&=
		\left|
		\sum_{s_1=0}^{n_1-1}\sum_{s_2=n_2-u_2}^{n_2-1}{
			e^{-i(\lambda_1 s_1 + \lambda_2 s_2)}
		}
		+
		\sum_{s_2=0}^{n_2}\sum_{s_1=n_1-u_1}^{n_1-1}{
			e^{-i(\lambda_1 s_1 + \lambda_2 s_2)}
		}
		-
		\sum_{s_2=n_2-u_2}^{n_2-1}\sum_{s_1=n_1-u_1}^{n_1-1}{
			e^{-i(\lambda_1 s_1 + \lambda_2 s_2)}
		}
		\right|,
		\end{align*}
		where we split the sum over non-zero terms, using the 
		fact that,
		\begin{align*}
		&\left\{(s_1, s_2)\in\Z^2:
		g_{s_1, s_2}
		(1 - 
		g_{s_1+u_1, s_2+u_2})=1
		\right\}
		=\\& \left\{
		(s_1, s_2)\in\Z^2: 0\leq s_1 < n_1
		\right\}
		\cap
		\left\{
		(s_1, s_2)\in\Z^2: 0\leq s_2 < n_2
		\right\}
		\cap\\&
		\left(
		\left\{
		(s_1, s_2)\in\Z^2: s_1\geq n_1-u_1 
		\right\}
		\cup
		\left\{
		(s_1, s_2)\in\Z^2: s_2\geq n_2-u_2 
		\right\}
		\right),
		\end{align*} and that 
		\begin{equation*}
		\sum_{A\cup B} = \sum_A + \sum_B - \sum_{A\cap B}.
		\end{equation*}
		Then by
		applying the triangle inequality we obtain,
		\begin{align*}
		&\left|
		\sum_{s}
		g_{\bs}g_{\bs+\bu}
		e^{-i\blambda \cdot\bs}
		-
		\Delta_n\left(\blambda\right)
		\right|
		\\
		&
		\leq
		\left|
		\sum_{s_1=0}^{n_1-1}\sum_{s_2=n_2-u_2}^{n_2-1}{
			e^{-i(\lambda_1 s_1 + \lambda_2 s_2)}
		}
		\right|
		+
		\left|
		\sum_{s_2=0}^{n_2-1}\sum_{s_1=n_1-u_1}^{n_1-1}{
			e^{-i(\lambda_1 s_1 + \lambda_2 s_2)}
		}
		\right|
		+
		\left|
		\sum_{s_2=n_2-u_2}^{n_2-1}\sum_{s_1=n_1-u_1}^{n_1}{
			e^{-i(\lambda_1 s_1 + \lambda_2 s_2)}
		}
		\right|
		\\
		&
		=
		\left|
		\sum_{s_2=n_2-u_2}^{n_2-1}{
			e^{-i \lambda_2 s_2}
			\sum_{s_1=0}^{n_1-1}
			e^{-i\lambda_1 s_1 }
		}
		\right|
		+
		\left|
		\sum_{s_1=n_1-u_1}^{n_1-1}{
			e^{-i\lambda_1 s_1}
			\sum_{s_2=0}^{n_2-1}
			e^{-i \lambda_2 s_2}
		}
		\right|
		+
		\left|
		\sum_{s_2=n_2-u_2}^{n_2-1}\sum_{s_1=n_1-u_1}^{n_1}{
			e^{-i(\lambda_1 s_1 + \lambda_2 s_2)}
		}
		\right|
		\\
		&
		=
		\left|
		\Delta_{n_1}(\lambda_1)
		\sum_{s_2=n_2-u_2}^{n_2-1}{
			e^{-i \lambda_2 s_2}
		}
		\right|
		+
		\left|
		\Delta_{n_2}(\lambda_2)
		\sum_{s_1=n_1-u_1}^{n_1-1}{
			e^{-i\lambda_1 s_1}
		}
		\right|
		+
		\left|
		\sum_{s_2=n_2-u_2}^{n_2-1}\sum_{s_1=n_1-u_1}^{n_1}{
			e^{-i(\lambda_1 s_1 + \lambda_2 s_2)}
		}
		\right|
		\\
		&
		\leq
		\left|\Delta_{n_1}(\lambda_1)\right|
		\sum_{s_2=n_2-u_2}^{n_2-1}{
			\left|
			e^{-i \lambda_2 s_2}
			\right|
		}
		+
		\left|\Delta_{n_2}(\lambda_2)\right|
		\sum_{s_1=n_1-u_1}^{n_1-1}{
			\left|
			e^{-i\lambda_1 s_1}
			\right|
		}
		+
		\sum_{s_2=n_2-u_2}^{n_2-1}\sum_{s_1=n_1-u_1}^{n_1}{
			\left|
			e^{-i(\lambda_1 s_1 + \lambda_2 s_2)}
			\right|
		}
		\\
		&
		=
		u_2
		\left|\Delta_{n_1}\left(\lambda_1\right)\right|
		+
		u_1
		\left|\Delta_{n_2}\left(\lambda_2\right)\right|
		+
		u_1u_2
		=
		u_1\left|\Delta_{\bn}^{(1)}(\blambda)\right|+
		u_2\left|\Delta_{\bn}^{(2)}(\blambda)\right|+
		u_1u_2.
		\end{align*}

		We now prove the result for any dimensionality $d\geq 1$ by induction on
		$d$.
		\begin{itemize}
			\item We already proved the result for the case $d=1$.
			\item Assume the property holds up to a given $d\geq 1$.
			Let $\bu\in\N^{d+1}$.
			Given any $\mathbf{v}\in\Z^{d+1}$, we denote $\mathbf{v}^{(d+1)}\in\Z^d$ the vector
			with components $v_1, \ldots, v_d$.
			We will make use of this notation for several vectors in the rest of 
			the proof.
			 We observe that,
			\begin{align*}
			&\left\{
			\bs\in\Z^{d+1}:
			g_{\bs}(1 - g_{\bs+\bu}) = 1
			\right\}
			=\\
			&
			\left\{
			\bs\in\Z^{d+1}:
			g_{\bs} = 1
			\right\}
			\cap \left(
			\left\{
			\bs\in\Z^{d+1}:
			g_{s_{d+1}+u_{d+1}} = 0
			\right\}
			\cup
			\left\{
			\bs\in\Z^{d+1}:
			g_{\bs^{(d+1)}+\bu^{(d+1)}} = 0
			\right\}
			\right)=\\
			&
			\left(
			\left\{
			\bs\in\Z^{d+1}:
			g_{\bs} = 1
			\right\}
			\cap 
			\left\{
			\bs\in\Z^{d+1}:
			g_{s_{d+1}+u_{d+1}} = 0
			\right\}
			\right)
			\cup\\&
			\left(
			\left\{
			\bs\in\Z^{d+1}:
			g_{\bs} = 1
			\right\}
			\cap
			\left\{
			\bs\in\Z^{d+1}:
			g_{\bs^{(d+1)}+\bu^{(d+1)}} = 0
			\right\}
			\right).
			\end{align*}
			Let 
			\begin{align*}
			A &= \left\{
			\bs\in\Z^{d+1}:
			g_{\bs}(1 - g_{\bs+\bu}) = 1
			\right\},\\
			B &= \left\{
			\bs\in\Z^{d+1}:
			g_{\bs} = 1
			\right\}
			\cap 
			\left\{
			\bs\in\Z^{d+1}:
			g_{s_{d+1}+u_{d+1}} = 0
			\right\},\\
			C &=
			\left\{
			\bs\in\Z^{d+1}:
			g_{\bs} = 1
			\right\}
			\cap
			\left\{
			\bs\in\Z^{d+1}:
			g_{\bs^{(d+1)}+\bu^{(d+1)}} = 0
			\right\}.
			\end{align*}
			The idea here is that we \emph{split} the problem
			between the last dimension (set $B$) and the 
			$d$ first dimensions taken altogether
			(set $C$).
			We then have, since $A=B\cup C$,
			\begin{equation*}
			\sum_{\bs\in A}{e^{i\bs\cdot\blambda}}
			=
			\sum_{\bs\in B}{e^{i\bs\cdot\blambda}}
			+
			\sum_{\bs\in C}{e^{i\bs\cdot\blambda}}
			-
			\sum_{\bs\in B\cap C}{e^{i\bs\cdot\blambda}},
			\end{equation*}
			and by the triangle inequality,
			\begin{equation*}
			\label{eq:dlficsdfd}
			\left|\sum_{\bs\in A}{e^{i\bs\cdot\blambda}}\right|
			\leq
			\left|\sum_{\bs\in B}{e^{i\bs\cdot\blambda}}\right|
			+
			\left|\sum_{\bs\in C}{e^{i\bs\cdot\blambda}}\right|
			+
			\left|\sum_{\bs\in B\cap C}{e^{i\bs\cdot\blambda}}\right|.
			\end{equation*}
			We consider each term separately.
			Firstly,
			\begin{align*}
			\left|\sum_{\bs\in B}{e^{i\bs\cdot\blambda}}\right|
			&=
			\left|
				\sum_{s_1=0}^{n_1-1}
				\ldots
				\sum_{s_d=0}^{n_d-1}
				\sum_{s_{d+1}=n_{d+1}- u_{d+1}}^{n_{d+1}-1}
				e^{i\sum_{j=1}^{d+1}{s_j\lambda_j}}
			\right|\\
			&=
			\left|
			\left(\sum_{s_{d+1}=n_{d+1}- u_{d+1}}^{n_{d+1}-1}
			e^{i{s_{d+1}\lambda_{d+1}}}
			\right)
			\left(
			\sum_{s_1=0}^{n_1-1}
			\ldots
			\sum_{s_d=0}^{n_d-1}
			e^{i\sum_{j=1}^{d}{s_j\lambda_j}}
			\right)
			\right|\\
			&=
			\left|
			\sum_{s_{d+1}=n_{d+1}- u_{d+1}}^{n_{d+1}-1}
			e^{i{s_{d+1}\lambda_{d+1}}}
			\right|
			\left|
			\sum_{s_1=0}^{n_1-1}
			\ldots
			\sum_{s_d=0}^{n_d-1}
			e^{i\sum_{j=1}^{d}{s_j\lambda_j}}
			\right|\\
			&\leq
			u_{d+1}\left|\Delta_\bn^{(d+1)}(\blambda)\right|.
			\end{align*}
			Secondly, using the fact that the property holds up to
			dimensionality
			$d$,
			\begin{align*}
			\left|\sum_{\bs\in C}{e^{i\bs\cdot\blambda}}\right|
			&=
			\left|\sum_{\bs\in\Z^{d+1}}{\mathbbm{1}_C(\bs) e^{i\bs\cdot\blambda}}\right|\\
			&=
			\left|
				\sum_{\bs\in\Z^{d+1}}
				{
				g_{\bs}(1 - g_{\bs^{(d+1)} + \bu^{(d+1)}})
				}
				e^{i\sum_{j=1}^{d+1}{s_j\lambda_j}}
			\right|
			\\
			&=
			\left|
			\left(
			\sum_{s_{d+1}=0}^{n_{d+1}}e^{i{s_{d+1}\lambda_{d+1}}}
			\right)
			\left(
			\sum_{\bs\in\Z^{d}}
			{
				g_{\bs}(1 - g_{\bs + \bu^{(d+1)}})
			}
			e^{i\sum_{j=1}^{d}{s_j\lambda_j}}
			\right)
			\right|
			\\
			&\leq
			\left|
			\sum_{s_{d+1}=0}^{n_{d+1}-1}e^{i s_{d+1} \lambda_{d+1}}
			\right|
			\left(
			\sum_{j=1}^d{u_j\left|\Delta_{\bn^{(d+1)}}^{(j)}(\blambda)\right|}
			+\sum_{\substack{j,k=1 \\ k > j}}^d{u_j u_k
				\left|\Delta_{\bn^{(d+1)}}^{(j, k)}(\blambda)\right|}
			+\ldots
			+u_1\cdots u_d
			\right)\\
			&=
			\sum_{j=1}^d{u_j\left|\Delta_{\bn}^{(j)}(\blambda)\right|}
			+\sum_{\substack{j,k=1 \\ k > j}}^d{u_j u_k\left|\Delta_{\bn}^{(j, k)}(\blambda)\right|}
			+\ldots
			+u_1\cdots u_d \left|\Delta_\bn^{(1, \ldots, d)}(\blambda)\right|,
			\end{align*}
			where in the last equality we used the fact that
			$	\left|
			\sum_{s_{d+1}=0}^{n_{d+1}-1}e^{i s_{d+1} \lambda_{d+1}}
			\right|\left|\Delta_{\bn^{(d+1)}}^{(j)}(\blambda)\right|
			=
			\left|\Delta_{\bn}^{(j)}(\blambda)\right|$.
			Thirdly, again using the fact that the property holds up to
			dimensionality
			$d$, 
			\begin{align*}
			\left|\sum_{\bs\in B\cap C}{e^{i\bs\cdot\blambda}}\right|
			&\leq
			\left|
			\sum_{s_{d+1}=n_{d+1} - u_{d+1}}^{n_{d+1}-1}e^{i s_{d+1} \lambda_{d+1}}
			\right|\\
			&\times
			\left(
			\sum_{j=1}^d{u_j\left|\Delta_{\bn}^{(j)}(\blambda)\right|}
			+\sum_{\substack{j,k=1 \\ k > j}}^d{u_j u_k
				\left|\Delta_{\bn^{(d+1)}}^{(j, k)}(\blambda)\right|}
			+\ldots
			+u_1\cdots u_d
			\right)\\
			&\leq
			u_{d+1}
			\left(
			\sum_{j=1}^d{u_j\left|\Delta_{\bn^{(d+1)}}^{(j)}(\blambda)\right|}
			+\sum_{\substack{j,k=1 \\ k > j}}^d{u_j u_k
				\left|\Delta_{\bn^{(d+1)}}^{(j, k)}(\blambda)\right|}
			+\ldots
			+u_1\cdots u_d
			\right)\\
			&=
			u_{d+1}
			\left(
			\sum_{j=1}^d{u_j\left|\Delta_{\bn}^{(j, d+1)}(\blambda)\right|}
			+\sum_{\substack{j,k=1 \\ k > j}}^d{u_j u_k\left|\Delta_{\bn}^{(j, k, d+1)}(\blambda)\right|}
			+\ldots
			+u_1\cdots u_d
			\right)
			\\
			&=
			\sum_{j=1}^d{u_ju_{d+1}\left|\Delta_{\bn}^{(j, d+1)}(\blambda)\right|}
			+\sum_{\substack{j,k=1 \\ k > j}}^d{u_j u_ku_{d+1}
				\left|\Delta_{\bn}^{(j, k, d+1)}(\blambda)\right|}
			+\ldots
			+u_1\cdots u_{d+1}.
			\end{align*}
			Substituting these expressions into~\eqref{eq:dlficsdfd}, we obtain,
			\begin{equation*}
			\left|
			\sum_{\bs\in\Z^{d+1}}
			{
				g_{\bs}g_{\bs+\bu}
				e^{i\bs\cdot\blambda}
			}
			-
			\Delta_{\bn}(\blambda)
			\right|
			\leq
			\sum_{j=1}^{d+1}{u_j\left|\Delta_\bn^{(j)}(\blambda)\right|}
			+
			\sum_{\substack{j,k=1 \\ k > j}}^{d+1}{u_j u_k
				\left|\Delta_\bn^{(j, k)}(\blambda)\right|}
			+
			\ldots
			+u_1\cdots u_{d+1},
			\end{equation*}
			which is exactly the desired property for dimensionality $d+1$.
		\end{itemize}
		By induction, we conclude that the property holds for any dimensionality $d$.
		\qed
	\end{proof}
	As an example, in dimension $d=3$, the inequality takes the following form,
	\begin{align*}
	\left|
	\sum_{\bs\in\Z^{3}}
	{
		g_{\bs}g_{\bs+\bu}
		e^{i\bs\cdot\blambda}
	}
	-
	\Delta_{\bn}(\blambda)
	\right|&\leq
	u_1\left|\Delta_{\bn}^{(1)}(\blambda)\right|
	+
	u_2\left|\Delta_{\bn}^{(2)}(\blambda)\right|
	+
	u_3\left|\Delta_{\bn}^{(3)}(\blambda)\right|\\&
	+
	u_2u_3 \left|\Delta_{\bn}^{(2, 3)}(\blambda)\right|
	+
	u_1u_3 \left|\Delta_{\bn}^{(1, 3)}(\blambda)\right|
	+
	u_1u_2 \left|\Delta_{\bn}^{(1, 2)}(\blambda)\right|
	\\&
	+u_1 u_2 u_3.
	\end{align*}
	
	We now use this result to approximate the $L$-th order cumulant of 
	the multi-dimensional DFT.
	
	\begin{lemma}[$L$-th order cumulants of the DFT]
		\label{lemma:cumulantsDFT}
		Suppose Assumption~\ref{ass:asymptoticNormality2} holds.
		For an integer $L\geq 2$, and
		$\bk_1, \ldots, \bk_L\in\R^d$, we have,
		\begin{align*}
		\cum_L\left\{
		\widetilde{J}_\bn(\bk_1),
		\ldots,
		\widetilde{J}_\bn(\bk_L)
		\right\}
		=
		\Delta_\bn\left(
		\sum_{j=1}^L{
			\bk_j
		}
		\right)
		f_L(\bk_1, \ldots, \bk_{L-1})
			+
			\mathcal{O}
			\left(
				\Lambda\left(\sum_{j=1}^L{
					\bk_j
				}\right)
			\right),
		\end{align*}
		where $f_L$ is the $L$-th cumulant spectral density and where
		we have defined,
		\begin{equation}
			\Lambda(\blambda) = 
			\sum_{j=1}^d{
				\left|\Delta_{\bn}^{(j)}(\blambda)\right|
			}
			+\sum_{\substack{j,k=1 \\ k > j}}^d{
				\left|\Delta_{\bn}^{(j, k)}(\blambda)\right|
			}
			+\ldots
			+1,
		\end{equation}
		and where the $\mathcal{O}(\cdot)$ does not depend on
		$\bk_1, \ldots, \bk_L$.
	\end{lemma}
	\begin{proof}
		By properties of cumulants, see Lemma~\ref{lemma:cumulantsbasics},
		 direct calculations give,
		\begin{align}
		\nonumber&\cum_L\left\{
		\widetilde{J}_\bn(\bk_1),
		\ldots,
		\widetilde{J}_\bn(\bk_L)
		\right\}
		=
		\sum_{\bs_1, \ldots, \bs_L\in\Z^d}{
			\cum(X_{\bs_1}, \ldots, X_{\bs_L})
			g_{\bs_1}\ldots g_{\bs_L}
			e^{-i\sum_{j=1}^L{\bk_j \cdot\bs_j}}
		}\\
		\nonumber&=
		\sum_{\bs_1}\sum_{\bu_1, \ldots, \bu_{L-1}}{
			c_L(\bu_1, \ldots, \bu_{L-1})
			g_{\bs_1}g_{\bs_1+\bu_1}\ldots g_{\bs_1+\bu_{L-1}}
			e^{-i\sum_{j=1}^{L-1}{\bk_j \cdot\bu_j}}
			e^{-i\sum_{j=1}^L{\bk_j \cdot\bs_1}}
		}\\
		&=
		\label{eq:formcumulantJ}
		\sum_{\bu_1, \ldots, \bu_{L-1}}{
			c_L(\bu_1, \ldots, \bu_{L-1})
			e^{-i\sum_{j=1}^{L-1}{\bk_j \cdot\bu_j}}
			\sum_{\bs_1}
			g_{\bs_1}g_{\bs_1+\bu_1}\ldots g_{\bs_1+\bu_{L-1}}
			e^{-i\sum_{j=1}^L{\bk_j \cdot\bs_1}}
		}.
		\end{align}
		Suppose for convenience that $\bu_1, \ldots, \bu_{L-1}$
		all have non-negative components. 
		The general case can be treated similarly, please
		see our comment on this at the end of this
		proof.
		Additionally, denote $\tilde{\bu}\in\N^d$ as the vector defined
		by 
		\begin{equation}
			\label{eq:utildadef}
			\tilde{u}_j=\max\{\bu_k \cdot \be_j: k=1, \ldots, L-1\},
		\end{equation}
		where $\be_j, j=1,\ldots,d$ denotes the $d$-vector with all components
		set to zero except for the $j$-th component which is set to $1$, 
		such that
		$\bu_k\cdot \be_j$ is the $j$-th component of $\bu_k$.
		The right-most term of~\eqref{eq:formcumulantJ}
		can be approximated
		using the fact that, for $\blambda\in\R^d$,
		\begin{align*}
		&\left|
		\sum_{\bs_1\in\Z^d}
		g_{\bs_1}g_{\bs_1+\bu_1}\ldots g_{\bs_1+\bu_{L-1}}
		e^{-i\blambda \cdot\bs_1}
		-
		\Delta_\bn\left(\blambda\right)
		\right|
		=
		\left|
		\sum_{\bs_1}
		g_{\bs_1}
		\left(
		g_{\bs_1+\bu_1}\ldots g_{\bs_1+\bu_{L-1}}
		-1
		\right)
		e^{-i\blambda \cdot\bs_1}
		\right|\\
		=
		&\left|
		\sum_{\bs_1}
		g_{\bs_1}
		\left(g_{\bs_1+\tilde{\bu}}
		-1
		\right)
		e^{-i\blambda \cdot\bs_1}
		\right|,
		\end{align*}
		due to assuming that the grid is fully observed
		and setting $g_\bs=1$ on the grid and $0$ otherwise.
		For instance, in the case $L=3$, 
		we have for $\bs_1\in\Z^d$,
		$
		g_{\bs_1} g_{\bs_1+\bu_1} g_{\bs_1+\bu_{2}} = 1
		\iff
		\bs_1\in\mathcal{J}_{\bn} \text{ and }
		\bs_1 + \bu_1 \in \mathcal{J}_\bn \text{ and } 
		\bs_1 + \bu_2 \in \mathcal{J}_\bn \iff 
		\bs_1\in\mathcal{J}_{\bn} \text{ and }
		\bs_1+\tilde{\bu}\in\mathcal{J}_{\bn}
		\iff
		g_{\bs_1}
		g_{\bs_1+\tilde{\bu}} = 1
		$.

		According to Lemma~\ref{lemma:ineqDelta}, we therefore have,
		\begin{equation*}
		\left|
		\sum_{\bs\in\Z^{d}}
		{
			g_{\bs}g_{\bs+\widetilde{\bu}}
			e^{i\bs\cdot\blambda}
		}
		-
		\Delta_{\bn}(\blambda)
		\right|
		\leq
		\sum_{j=1}^{d}{\widetilde{u}_j\left|\Delta_\bn^{(j)}(\blambda)\right|}
		+
		\sum_{\substack{j,k=1 \\ k > j}}^{d}{\widetilde{u}_j \widetilde{u}_k
			\left|\Delta_\bn^{(j, k)}(\blambda)\right|}
		+
		\ldots
		+\widetilde{u}_1 \cdots\widetilde{u}_{d}.
		\end{equation*}
		We use the inequality $\widetilde{u}_1\ldots\widetilde{u}_d
		\leq \left(\max_{i=1,\ldots,d}{\widetilde{u}_i}\right)^d
		\leq
		\widetilde{u}_1^d + \ldots + \widetilde{u}_d^d$ (by
		definition of $\widetilde{\bu}$ its components are non-negative)
		 and obtain,
		\begin{equation*}
		\left|
		\sum_{\bs\in\Z^{d}}
		{
			g_{\bs}g_{\bs+\widetilde{\bu}}
			e^{i\bs\cdot\blambda}
		}
		-
		\Delta_{\bn}(\blambda)
		\right|
		\leq
		\left(
		\widetilde{u}_1^d + \ldots + \widetilde{u}_d^d
		\right)
		\left(
		\sum_{j=1}^{d}{\left|\Delta_\bn^{(j)}(\blambda)\right|}
		+
		\sum_{\substack{j,k=1 \\ k > j}}^{d}{
			\left|\Delta_\bn^{(j, k)}(\blambda)\right|}
		+
		\ldots
		+1
		\right).
		\end{equation*}
		Now given our definition of
		$\tilde{\mathbf{u}}=\begin{pmatrix} \tilde{u}_1 & \ldots & \tilde{u}_d \end{pmatrix}^T$, see~\eqref{eq:utildadef},
		 we have $\widetilde{u}_1^d + \ldots + \widetilde{u}_d^d
		\leq \|\bu_1\|_1^d + \ldots + \|\bu_{L-1}\|_1^d$,
		and therefore,
		\begin{equation}
		\label{eq:finalupperbound}
			\left|
			\sum_{s}
			g_{\bs}g_{\bs+\tilde{\bu}}
			e^{-i\blambda \cdot\bs}
			-
			\Delta_n\left(\blambda\right)
			\right|
			\leq
			\left(
			\|\bu_1\|_1^d + \ldots + \|\bu_{L-1}\|_1^d
			\right)
			\left(
			\sum_{j=1}^{d}{\left|\Delta_\bn^{(j)}(\blambda)\right|}
			+
			\sum_{\substack{j,k=1 \\ k > j}}^{d}{
				\left|\Delta_\bn^{(j, k)}(\blambda)\right|}
			+
			\ldots
			+1
			\right)
		\end{equation}
		Finally, going back to~\eqref{eq:formcumulantJ}, 
		we write 
		\begin{equation*}
		\sum_{\bs_1}
		g_{\bs_1}g_{\bs_1+\bu_1}\ldots g_{\bs_1+\bu_{L-1}}
		e^{-i\sum_{j=1}^L{\bk_j \cdot\bs_1}} = 
		\Delta_\bn\left(\sum_{j=1}^L{\bk_j \cdot\bs_1}\right)
		+ 
		\mathcal{E}\left(
		\sum_{j=1}^L{\bk_j \cdot\bs_1}
		\right),
		\end{equation*}
		 with
		\begin{equation*}
		\mathcal{E}\left(
		\blambda\right)
		=
		\sum_{\bs_1}
		g_{\bs_1}g_{\bs_1+\bu_1}\ldots g_{\bs_1+\bu_{L-1}}
		e^{-i\blambda}
		-
		\Delta_\bn\left(\blambda\right),
		\end{equation*}
		where for simplicity we do not denote explicitly the dependence of 
		$\mathcal{E}(\cdot)$ on $\bu_1, \ldots, \bu_{L-1}$.
		We then use the upper-bound~\eqref{eq:finalupperbound} we derived for $\left|\mathcal{E}(\blambda)\right|$,
		and Assumption~\ref{ass:asymptoticNormality2} on the
		summability of cumulants to obtain,
		\begin{equation*}
		\left|
			\sum_{\bu_1, \ldots, \bu_{L-1}}{
				c_L(\bu_1, \ldots, \bu_{L-1})
				e^{-i \blambda \cdot\bu_j}
				\mathcal{E}\left(\blambda\right)
			}
		\right|
		=\mathcal{O}
		\left(
		\sum_{j=1}^d{\left|\Delta_{\bn}^{(j)}(\blambda)\right|}
		+\sum_{\substack{j,k=1 \\ k > j}}^d{\left|\Delta_{\bn}^{(j, k)}(\blambda)\right|}
		+\ldots
		+1
		\right).
		\end{equation*}
		This concludes the proof.
		We now comment on how to adapt the proof to
		the case where $\bu_1, \ldots, \bu_{L-1}$ are not
		restricted to having non-negative components.
		This is achieved by replacing~\eqref{eq:utildadef}
		with,
		\begin{align*}
		\tilde{u}_j^+&=\max\{0, 
		\max\{\bu_k \cdot \be_j: k=1, \ldots, L-1\}\}\\
		\tilde{u}_j^-&=\max\{0, 
		\max\{-\bu_k \cdot \be_j: k=1, \ldots, L-1\}\}.
		\end{align*}
		This is because when
		 allowing for negative components, we have to treat
		both boundaries of the domain along each dimension $j=1, \cdots, d$.
		This is accounted for in the final formula in the $\mathcal{O}(\cdot)$.
		\qed
	\end{proof}
	
	In the proof of Proposition~\ref{prop:ordercumulants} of 
	this Supplementary Material, when
	expressing the cumulant of order $L$ of the periodogram evaluated
	at Fourier frequencies $\bomega_1\ldots, \bomega_L\in\Omega_\bn$ in terms
	of cumulants of the DFT (which we studied in Lemma~\ref{lemma:cumulantsDFT} of 
	this Supplementary Material), we will need to understand the order
	of terms of the form
	\begin{equation}
	\sum_{\bomega_1,\ldots,\bomega_L\in\Omega_{\bn}}
	\prod_{\nu_r\in\nu}\Delta_{\bn}\left(\sum_{j\in\nu_r}{\bomega_j}\right),
	\end{equation}
	where $\nu$ is an indecomposable partition of the $L\times 2$ table
	given in~\eqref{eq:table} and where we set $\bomega_{k+L} = -\bomega_k,
	\quad k=1,\ldots,L$ (see Figure~\ref{fig:partitionexample} and compare
	to the $L\times 2$ table~\eqref{eq:table}).
	While the $\Delta_{\bn}(\cdot)$ function can take value $|\bn|$, this only occurs
	under linear constraints on the $\bomega_1, \ldots, \bomega_L$. 
	For example, in the case $L=4$ and for the partition of the $L\times 2$
	table represented
	in Figure~\ref{fig:partitionexample}, we get the following set of
	linear constraints on the Fourier frequencies,
	\begin{equation}
	\begin{cases}
	\bomega_1 - \bomega_1 + \bomega_2 + \bomega_3 \equiv 0 \ [2\pi]\\
	\bomega_4 - \bomega_2 \equiv 0 \ [2\pi]\\
	-\bomega_4 - \bomega_3 \equiv 0 \ [2\pi]
	\end{cases},
	\end{equation}
	two of which are linearly independent.
	\begin{figure}
		\centering
		\includegraphics[width=0.25\linewidth]{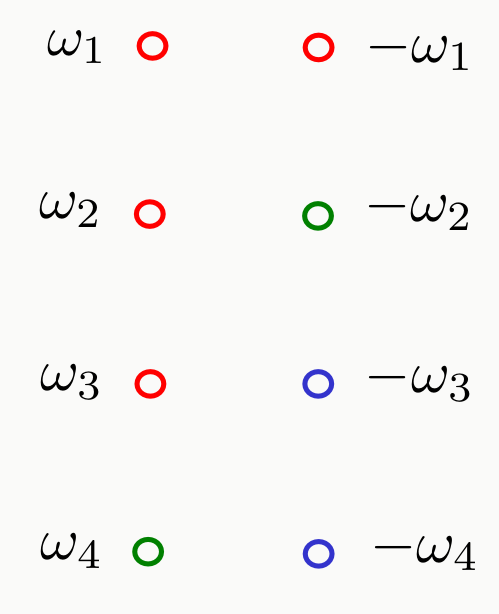}
		\caption{Example of an indecomposable partition of a $4\times 2$
			table that is used in expressing the $4$-th order
			cumulants of the periodogram
			at frequencies $\bk_1, \bk_2, \bk_3, \bk_4$ in terms of cumulants of the DFT at frequencies
			$\bk_1, -\bk_1, \bk_2, -\bk_2, \bk_3, -\bk_3, \bk_4,
			-\bk_4$. The chosen indecomposable
		partition has 3 sets, indicated by the colors red, green and
	blue.}
		\label{fig:partitionexample}
	\end{figure}
	The following
	lemma makes this property explicit.
	
	\begin{lemma}
		\label{lemma:constraints}
		Let $\nu=(\nu_1, \ldots, \nu_p)$ be an indecomposable partition
		of the $L\times 2$ table, where $L$ is a positive integer.
		The following system of linear equations in 
		$(\bomega_1, \ldots, \bomega_L)$
		\begin{equation}
		(\mathcal{S}_\nu)
		\begin{cases}
		\sum_{j\in\nu_1}{\bomega_j} &\equiv 0 \ [2\pi]\\
		&\vdots\\
		\sum_{j\in\nu_p}{\bomega_j} &\equiv 0 \ [2\pi]
		\end{cases}
		\end{equation}
		imposes $p-1$ linear constraints on the 
		$(\bomega_1, \ldots, \bomega_L)$.
	\end{lemma}
	\begin{proof}
		We remind the reader that we have defined 
		$\bomega_{k+L}\equiv - \bomega_k$ for $k=1, \ldots, L$.
		The proof is done by induction on the number of sets $p$ in
		the partition.
		\begin{itemize}
			\item In the case $p=1$, the partition consists of
			a unique set, and each $\bomega_j$ in the summation is cancelled out
			by $\bomega_{L+j}$. Hence
			$\sum_{j\in\nu_1}{\bomega_j} \equiv 0 \ [2\pi] 
			\iff 0 \equiv 0 \ [2\pi]$, so that there are no linear
			constraints and the property holds for $p=1$.
			\item Suppose the property holds up to a given positive integer
			$p$. We want to show that it also holds for any partition with
			$p+1$ sets. Therefore, let $\nu=(\nu_1, \ldots, \nu_{p+1})$
			be
			an indecomposable partition of the $L\times 2$ table with
			cardinality $p+1$. Without
			loss of generality, we assume that the ordering of the sets
			is such that $\nu_1$ communicates with $\nu_2$, i.e. there
			exists $k\in\{1, \ldots, L\}$ such that $k\in\nu_1$ and
			$L+k\in\nu_2$ (or the reverse case, but again we can treat
			either of these two cases without loss of generality),
			 i.e. the $k$-th row has one element
			that belongs to $\nu_1$ and the other one to $\nu_2$, since
			in the $L\times 2$ table $k$ and $L+k$ are on the same row.
			
			We then observe that we can rewrite the system
			\begin{equation}
			\label{eq:system1}
			(\mathcal{S}_\nu)
			\begin{cases}
			\sum_{j\in\nu_1}{\bomega_j} &\equiv 0 \ [2\pi]\\
			&\vdots\\
			\sum_{j\in\nu_p+1}{\bomega_j} &\equiv 0 \ [2\pi]
			\end{cases},
			\end{equation}
			as 
			\begin{equation}
			\label{eq:system2}
			(\mathcal{S}_\nu)
			\begin{cases}
			\sum_{j\in\nu_1}{\bomega_j} &\equiv 0 \ [2\pi]\\
			\sum_{j\in\nu_1\cup\nu_2}{\bomega_j} &\equiv 0 \ [2\pi]\\
			\sum_{j\in\nu_3}{\bomega_j} &\equiv 0 \ [2\pi]\\
			&\vdots\\
			\sum_{j\in\nu_p}{\bomega_j} &\equiv 0 \ [2\pi]
			\end{cases},
			\end{equation}
			where the second equation in~\eqref{eq:system2} is obtained by summing
			the first two equations in~\eqref{eq:system1}, using the fact
			that $\nu_1\cap\nu_2 = 0$, by definition of a partition.
			Based on this partition $\nu$, we define a new
			partition $\widetilde{\nu}
			= (\nu_1\cup\nu_2, \nu_3, \ldots, \nu_{p+1})$. The set
			$\widetilde{\nu}$ is clearly a partition of the $L\times 2$ table,
			and it has $p$ sets. Additionally, one can verify
			that this new partition $\widetilde{\nu}$ is also indecomposable.
			The solution space to $\mathcal{S}_\nu$ is therefore the intersection
			between the solution spaces to $\mathcal{S}_{\widetilde{\nu}}$ and 
			$\sum_{j\in\nu_1}{\bomega_j} \equiv 0 \ [2\pi]$.
			
			By assumption, $\mathcal{S}_{\widetilde{\nu}}$ 
			enforces $p-1$ linear constraints on $(\bomega_1, \ldots, \bomega_L)$. It therefore suffices to show that
			$\sum_{j\in\nu_1}{\bomega_j} \equiv 0 \ [2\pi]$ and the system $\mathcal{S}_{\widetilde{\nu}}$ are linearly independent.
			Or, equivalently, that
			there exists a set of values of $(\bomega_1, \ldots, \bomega_L)$
			that is a solution of
			$\mathcal{S}_{\widetilde{\nu}}$ but such that 
			$\sum_{j\in\nu_1}{\bomega_j} \not\equiv 0 \ [2\pi]$.
			Such a set of values is obtained by setting all components
			equal to zero modulo $2\pi$, except for the $k$-th and
			$L+k$-th components,
			where $k$ was defined earlier in this proof as the row on
			which the sets $\nu_1$ and $\nu_2$ communicate. More precisely
			we set $\bk_k=a$ and $\bk_{L+k}=-a$ where $a$ is chosen 
			such that $a \not\equiv 0 \ [2\pi]$.
			Hence the number of linear constraints enforced by $\mathcal{S}_\nu$
			on
			$(\bk_1, \ldots, \bk_L)$
			is $p - 1 + 1 = (p+1)-1$, so that the property also
			holds for any partition with $p+1$ sets.
		\end{itemize}
		By induction, since we proved the result for the partition of
		cardinality $p=1$,
		 we can conclude that the property holds for any indecomposable
		partition.\qed
	\end{proof}
	
	We can now proceed to determine an upper-bound for the higher-order
	cumulants of linear functionals of the periodogram, following the proof of ~\citet[Theorem~5.10.1]{brillinger2001time}.
	
	\begin{proposition}
		\label{prop:ordercumulants}
		Let $L$ be a positive integer. We have,
		\begin{equation}
		\label{eq:cumLperfunctionalresult}
		\cum_L\left\{
		|\bn|^{-1}\sum_{\bk\in\Omega_{\bn}}
		{
			w(\bk)
			I_{\bn}(\bk)
		}
		\right\}
		=
		\mathcal{O}\left(
		|\bn|^{
			1 - L
		}
		\right).
		\end{equation}
	\end{proposition}
	\begin{proof}
		Using the properties of cumulants given in
		Lemma~\ref{lemma:cumulantsbasics} from this
		Supplementary Material, we have
		\begin{align*}
		&\cum_L\left\{
		|\bn|^{-1}\sum_{\bk\in\Omega_{\bn}}
		{
			w(\bk)
			I_{\bn}(\bk)
		}
		\right\} =\\& \quad \quad \quad \quad \quad \quad \quad \quad
		|\bn_k|^{-L}
		\sum_{\bk_1, \ldots, \bk_L\in\Omega_{\bn_k}}
		{
			w_k(\bk_1)
			\ldots 
			w_k(\bk_L)
			\cum
			\left[
			I_{\bn_k}(\bk_1), 
			\ldots,
			I_{\bn_k}(\bk_L)
			\right]
		}.
		\end{align*}
		According to Lemma~\ref{lemma:cumofproducts} of this Supplementary
		material, we obtain,
		\begin{align}
		\label{eq:cumulantsperiodogram2}
		\nonumber\cum
		\left[
		I_{\bn_k}(\bk_1), 
		\ldots,
		I_{\bn_k}(\bk_L)
		\right]
		&=
		\cum
		\left[
		J_{\bn_k}(\bk_1)J_{\bn_k}(-\bk_1), 
		\ldots,
		J_{\bn_k}(\bk_L)J_{\bn_k}(-\bk_L)
		\right]
		\\
		&=
		\sum_\nu{
			\cum\left[
			J_{\bn_k}(\bk_j):j\in\nu_1
			\right]
			\ldots
			\cum\left[
			J_{\bn_k}(\bk_j):j\in\nu_p
			\right]
		},
		\end{align}
		where the summation is over indecomposable partitions
		$\nu=(\nu_1, \ldots, \nu_p)$
		of the $L\times 2$ table~\eqref{eq:table}, and where we define 
		$\bk_{j+L}\equiv-\bk_j, j=1,\ldots,L$.
		Hence, reminding the reader that we write
		$\widetilde{J}_\bn(\bomega) = |\bn|^{\frac{1}{2}}J_\bn(\bomega)$,
		\begin{align}
		\label{eq:cumLperfunctional}
		&\nonumber\cum_L\left\{
		|\bn|^{-1}\sum_{\bk\in\Omega_{\bn}}
		{
			w(\bk)
			I_{\bn}(\bk)
		}
		\right\} =
		\\&\nonumber |\bn_k|^{-L}
		\sum_{\bk_1, \ldots, \bk_L\in\Omega_{\bn_k}}
		{
			w_k(\bk_1)
			\ldots 
			w_k(\bk_L)
			\sum_\nu{
				\cum\left[
				J_{\bn_k}(\bk_j):j\in\nu_1
				\right]
				\ldots
				\cum\left[
				J_{\bn_k}(\bk_j):j\in\nu_p
				\right]
			}
		}
		\\&\nonumber=|\bn_k|^{-2L}
		\sum_{\bk_1, \ldots, \bk_L\in\Omega_{\bn_k}}
		{
			w_k(\bk_1)
			\ldots 
			w_k(\bk_L)
			\sum_\nu{
				\cum\left[
				\widetilde{J}_{\bn_k}(\bk_j):j\in\nu_1
				\right]
				\ldots
				\cum\left[
				\widetilde{J}_{\bn_k}(\bk_j):j\in\nu_p
				\right].
			}
		}
		\\&=|\bn_k|^{-2L}
		\sum_\nu
		\sum_{\bk_1, \ldots, \bk_L\in\Omega_{\bn_k}}
		{
			w_k(\bk_1)
			\ldots 
			w_k(\bk_L)
			\prod_{r=1}^p{
				\cum\left[
				\widetilde{J}_{\bn_k}(\bk_j):j\in\nu_p
				\right]
			}
		}.
		\end{align}
		We now make use of Lemma~\ref{lemma:cumulantsDFT}
		in which we obtained an expression for the terms
		\begin{equation*}
		\cum\left[
		\widetilde{J}_{\bn_k}(\bk_j):j\in\nu_p
		\right], \quad r=1,\ldots,p,
		\end{equation*} which appear in the product in~\eqref{eq:cumLperfunctional}.
		This leads us to,
		\begin{align*}
			&\nonumber\cum_L\left\{
			|\bn|^{-1}\sum_{\bk\in\Omega_{\bn}}
			{
				w(\bk)
				I_{\bn}(\bk)
			}
			\right\} =
			\\&
			|\bn_k|^{-2L}
			\sum_\nu
			\sum_{\bk_1, \ldots, \bk_L\in\Omega_{\bn_k}}
				w_k(\bk_1)
				\ldots 
				w_k(\bk_L)
				\\
				&
				\times
				\prod_{r=1}^p\left\{
						f_{m_r+1}(\bk_l: l\in\nu_r)
						\Delta_\bn\left(
						\sum_{l\in\nu_r}{
							\bk_j
						}
						\right)
						\right.
						+
						\left.
						\mathcal{O}
						\left(
						\sum_{j=1}^d{\left|\Delta_{\bn}^{(j)}
							\left(
							\sum_{l\in\nu_r}{
								\bk_j
							}
							\right)
							\right|}
						+\sum_{\substack{j,k=1 \\ k > j}}^d{
							\left|\Delta_{\bn}^{(j, k)}
							\left(
							\sum_{l\in\nu_r}{
								\bk_j
							}
							\right)
							\right|}
						+\ldots
						+1
						\right)
				\right\}
			,
		\end{align*}
		where $p$ is the cardinality of the partition $\nu$, and for each
		set $\nu_r, \ r=1,\cdots,p$, of the partition, $m_r$ is the the 
		cardinality of the set $\nu_r$. Additionally, 
		$f_{k}(\cdots)$ is the $k$-th order cumulant spectral density.
		Note that the slight abuse of notation $f_{m_r+1}(\bk_l: l\in\nu_r)$
		makes sense since the cumulant spectral densities are symmetric,
		due to the symmetry of the cumulants themselves.
		
		Now to determine the order of this term for a given indecomposable partition
		$\nu$, we introduce some additional
		notation, and follow the reasoning found in~\cite{brillinger2001time} 
		for the analysis of time series.
		For $r=1, \ldots, p$, let $q_r\in\{0,\ldots, d\}$ and
		$l^\pr = l_1^\pr,\ldots, l_{q_r}^\pr\in\{1, \ldots, d\}$. Expanding the previous
		expression for that given partition $\nu$ 
		will lead to a sum of terms of the form,
		\begin{equation}
			\label{eq:epofciefdopsi}
			\left\{
			\prod_{j=1}^d{n_j^{-2L}}
			\right\}
			\sum_{\bk_1, \ldots, \bk_L\in\Omega_\bn}
			\left\{
			\prod_{r=1}^p
			\left|\Delta_\bn^{(l_1^\pr, \ldots, l_{q_r}^\pr)}
			\left(
			\sum_{j\in\nu_r}{
				\bk_j
			}
			\right)\right|
			\right\},
		\end{equation}
		ignoring multiplicative constants and the $w_k(\cdot)$ terms for simplicity, 
		as the latter are upper-bounded in absolute value by assumption.
		
		Now for a given $r=1, \ldots, p$, 
		$\Delta_\bn^{(l_1^\pr, \ldots, l_{q_r}^\pr)}
		\left(
		\sum_{j\in\nu_r}{
			\bk_j
		}
		\right)$ will be zero (since the $\bk_j$'s are Fourier frequencies)
		unless $\sum_{j\in\nu_r}\omega_{j,k} \equiv 0 \ [2\pi], \ \forall k\in\overline{l^\pr}$,
		where $\overline{l^\pr}$ denotes the complementary of $l^\pr$ within the 
		set $\{1, \ldots, d\}$. In the latter case,  
		$\Delta_\bn^{(l_1^\pr, \ldots, l_{q_r}^\pr)}
		\left(
		\sum_{j\in\nu_r}{
			\bk_j
		}
		\right)$
		will take value $\prod_{j\in\overline{l^\pr}} n_j$.
		For each dimension $j=1, \ldots, d$, denote $\mathcal{S}_j$ the system
		of linear equations expressing the constraints on the $j$-th dimension
		between $(\bk_1, \ldots, \bk_L)$ due to 
		$
		\prod_{r=1}^p
		\Delta_\bn^{(l_1^\pr, \ldots, l_{q_r}^\pr)}
		\left(
		\sum_{j\in\nu_r}{
			\bk_j
		}
		\right)
		$. We also define 
		$\kappa_j = \sum_{r=1}^p \mathbbm{1}_{j\in\overline{l^\pr}}$
		for each dimension $j=1,\cdots, d$,
		and note that $\mathcal{S}_j$ is a system of $\kappa_j$ linear equations, 
		$0\leq \kappa_j \leq p$.
		Then~\eqref{eq:epofciefdopsi} becomes,
		\begin{align}
			\label{eq:dfclidugdi}
			&\nonumber\left\{\prod_{j=1}^d{n_j^{-2L}}\right\}
			\sum_{\bk_1, \ldots, \bk_L\in\Omega_\bn}
			\left\{\prod_{j=1}^d
			{
				n_j^{\kappa_j}
			}
			\right\}
			\left\{\prod_{j=1}^d
			{
				\mathbbm{1}_{(\omega_{1,j}, \ldots, \omega_{L, j})\in \mathcal{S}_j}
			}
			\right\}
		\\&=
		\left\{\prod_{j=1}^d{n_j^{-2L}}
			n_j^{\kappa_j}
		\right\}
		\sum_{\bk_1, \ldots, \bk_L\in\Omega_\bn}
		\prod_{j=1}^d
		\mathbbm{1}_{(\omega_{1,j}, \ldots, \omega_{L, j})\in \mathcal{S}_j}
		\end{align}
		where we make a slight abuse of notation by confounding $\mathcal{S}_j$ and its
		solution set. Finally,~\eqref{eq:dfclidugdi} becomes, with $\#\mathcal{S}_j$ the 
		cardinality of $S_j\cap\Omega_{n_j}$,
		\begin{equation*}
			\prod_{j=1}^d{n_j^{-2L}
				n_j^{\kappa_j}
				\#\mathcal{S}_j.
			}
		\end{equation*}
		However, we have $\#\mathcal{S}_j \leq n_j^{L - \kappa_j + 1}$, by generalization of
		Lemma~\ref{lemma:constraints} of this Supplementary Material,
		according to which $\mathcal{S}_j$ imposes at least $\kappa_j - 1$
		independent constraints.
		Thus the term of interest is at most of order
		\begin{equation*}
		\prod_{j=1}^d{n_j^{-2L}
			n_j^{\kappa_j}
			n_j^{L - \kappa_j + 1}
		} 
		=	\prod_{j=1}^d{n_j^{1-L}}
		= \left(\prod_{j=1}^d{n_j}\right)^{1-L}
		= |\bn|^{1-L},
		\end{equation*}
		which concludes the proof.\qed
	\end{proof}

\subsection*{Proof of Proposition~\ref{prop:asympNormalNonGaussian}}
\begin{proof}
	\begin{enumerate}
		\item \textit{Asymptotic normality.}
	We first consider the case of a grid growing to infinity in all 
	directions, i.e. $\Omega_{\bn}=\Omega_{\bn}^{(1)}$. Under the considered set of assumptions, i.e.
	 Assumption~\ref{ass:asymptoticNormality2}, the variance of 
	\[  	|\bn|^{-1}\sum_{\bk\in\Omega_{\bn}}
	{
		w_k(\bk)
		I_{\bn}(\bk),
	}
	\]
	is $\Theta\left(|\bn|^{-1}\right)$. 
	In order to establish
	asymptotic normality we therefore wish to show that
	the rescaled quantity
	$	|\bn|^{-1/2}\sum_{\bk\in\Omega_{\bn}}
	{
		w_k(\bk)
		I_{\bn}(\bk)
	}
	$ has cumulants of order 3 or greater that all converge to zero.
	According to Proposition~\ref{prop:ordercumulants}, the $L$-th
	order cumulant of \\$	|\bn|^{-1/2}\sum_{\bk\in\Omega_{\bn}}
	{
		w_k(\bk)
		I_{\bn}(\bk)
	}
	$
	is
	$
	\mathcal{O}\left(
	|\bn|^{
		-\frac{L}{2}+1
	}
	\right)
	$, which indeed converges to zero for $L\geq 3$.
	Thus we conclude that $	|\bn|^{-1}\sum_{\bk\in\Omega_{\bn}}
	{
		w_k(\bk)
		I_{\bn}(\bk)
	}
	$ is asymptotically normally distributed.
	The proof readily extends to vector-valued functions $\mathbf{w}_k(\cdot)$.
	In the case where one or more dimensions of the domain are bounded,
	$\Omega_{\bn}=\Omega_{\bn}^{(2)}$, and we prove the result by 
	splitting the summation into the summation over $\Omega_{\bn}^{(1)}$
	and $\Omega_{\bn}^{(2)}\setminus\Omega_{\bn}^{(1)}$. Each term
	is treated as above, and we obtain a sum of two asymptotically normal
	random variables.
	
	\item \textit{Asymptotic form of the variance.}
	For this part, it is assumed that the grid grows to infinity
	in all directions, which is a constraint on the observation
	domain.
	We remind the reader that in that case we choose $\Omega_{\bn}=
	\Omega_\bn^{(1)}$.
	We treat the case of scalar-valued $w_k(\cdot)$, but again the
	proof readily extends to vector-valued functions.
	We have,
	\begin{align*}
	\var\left\{
	\frac{1}{|\bn|}\sum_{\bk\in\Omega_{\bn}}{
		w_k(\bk)I_{\bn}(\bk)
	}
	\right\}
	&=
	\frac{1}{|\bn|^2}
	\sum_{\bk_1, \bk_2\in\Omega_\bn}{
		w_k(\bk_1)
		w_k(\bk_2)
		\cov\left\{
		I_{\bn}(\bk_1)
		,
		I_{\bn}(\bk_2)
		\right\}
	}\\
	&=
	\frac{1}{|\bn|^2}
	\sum_{\bk_1, \bk_2\in\Omega_\bn}{
		w_k(\bk_1)
		w_k(\bk_2)
		\cov\left\{
		J_{\bn}(\bk_1)J_{\bn}(-\bk_1)
		,
		J_{\bn}(\bk_2)J_{\bn}(-\bk_2)
		\right\}
	}\\
	&=
	\frac{1}{|\bn|^4}
	\sum_{\bk_1, \bk_2\in\Omega_\bn}{
		w_k(\bk_1)
		w_k(\bk_2)
		\cov\left\{
		\widetilde{J}_{\bn}(\bk_1)\widetilde{J}_{\bn}(-\bk_1)
		,
		\widetilde{J}_{\bn}(\bk_2)\widetilde{J}_{\bn}(-\bk_2)
		\right\},
	}
	\end{align*}
	where we remind the reader that we defined
	$\widetilde{J}_\bn(\bomega) = |\bn|^{\frac{1}{2}}J_\bn(\bomega)$.
	Making use of Lemma~\ref{lemma:cumofproducts} from this Supplementary
	Material, we have,
	\begin{align}
	\label{eq:141021}
	\nonumber\cov\left\{
	\widetilde{J}_{\bn}(\bk_1)\widetilde{J}_{\bn}(-\bk_1)
	,
	\widetilde{J}_{\bn}(\bk_2)\widetilde{J}_{\bn}(-\bk_2)
	\right\}
	=&
	\cov\left\{
	\widetilde{J}_{\bn}(\bk_1), \widetilde{J}_{\bn}(\bk_2)
	\right\}
	\cov\left\{
	\widetilde{J}_{\bn}(-\bk_1), \widetilde{J}_{\bn}(-\bk_2)
	\right\}\\
	&\nonumber+
	\cov\left\{
	\widetilde{J}_{\bn}(\bk_1), \widetilde{J}_{\bn}(-\bk_2)
	\right\}
	\cov\left\{
	\widetilde{J}_{\bn}(-\bk_1), \widetilde{J}_{\bn}(\bk_2)
	\right\}\\
	&+\cum_4\{\widetilde{J}_{\bn}(\bk_1), \widetilde{J}_{\bn}(\bk_2),
	\widetilde{J}_{\bn}(-\bk_1), \widetilde{J}_{\bn}(-\bk_2)\},
	\end{align}
	the remaining terms being zero since 
	$\E\{\widetilde{J}_\bn(\bk_1)\} = \E\{\widetilde{J}_\bn(\bk_2)\} = 0$
	as the random field is zero-mean.
	With Lemma~\ref{lemma:cumulantsDFT} ,
	\begin{align*}
	&\cov\left\{
	\widetilde{J}_{\bn}(\bk_1), \widetilde{J}_{\bn}(\bk_2)
	\right\}
	=
	\cov\left\{
	\widetilde{J}_{\bn}(-\bk_1), \widetilde{J}_{\bn}(-\bk_2)
	\right\}=
	\\
		&f_{X,\delta}(\bk_1)
		\Delta_{\bn}\left(\bk_1 + \bk_2\right)
	+
	\mathcal{O}
	\left(
	\sum_{j=1}^d{\left|\Delta_{\bn}^{(j)}(\bk_1 + \bk_2)\right|}
	+\sum_{\substack{j,k=1 \\ k > j}}^d{\left|\Delta_{\bn}^{(j, k)}(\bk_1 + \bk_2)\right|}
	+\ldots
	+1
	\right),
	\end{align*}
	as well as,
	\begin{align*}
	&\cov\left\{
	\widetilde{J}_{\bn}(\bk_1), \widetilde{J}_{\bn}(-\bk_2)
	\right\}
	=
	\cov\left\{
	\widetilde{J}_{\bn}(-\bk_1), \widetilde{J}_{\bn}(\bk_2)
	\right\}=
	\\
	&f_{X,\delta}(\bk_1)
	\Delta_{\bn}\left(\bk_1 - \bk_2\right)
	+
	\mathcal{O}
	\left(
	\sum_{j=1}^d{\left|\Delta_{\bn}^{(j)}(\bk_1 - \bk_2)\right|}
	+\sum_{\substack{j,k=1 \\ k > j}}^d{\left|\Delta_{\bn}^{(j, k)}(\bk_1 - \bk_2)\right|}
	+\ldots
	+1
	\right),
	\end{align*}
	and,
	\begin{align*}
	\cum_4\{\widetilde{J}_{\bn}(\bk_1), &\widetilde{J}_{\bn}(\bk_2),
		\widetilde{J}_{\bn}(-\bk_1), \widetilde{J}_{\bn}(-\bk_2)\}\\
	&=
		f_4(\bk_1, \bk2, - \bk_1)
		\Delta_{\bn}\left(\mathbf{0}\right)
		+
		\mathcal{O}
		\left(
		\sum_{j=1}^d{\left|\Delta_{\bn}^{(j)}(\mathbf{0})\right|}
		+\sum_{\substack{j,k=1 \\ k > j}}^d{\left|\Delta_{\bn}^{(j, k)}(\mathbf{0})\right|}
		+\ldots
		+1
		\right).
	\end{align*}
	With the assumption of a grid that grows to infinity in all directions,
	one can verify that the contribution of any term involving $\Delta_{\bn}^{(j)}$,
	$\Delta_{\bn}^{(j, k)}$ and so on, will become negligible w.r.t
	that of the terms involving $\Delta_{\bn}$. 
	We therefore limit our study to the latter terms that appear 
	in~\eqref{eq:141021}.
	\begin{enumerate}
		\item 
	 We have, reminding the reader that the function $w(\cdot)$ defined on 
	 $\mathcal{T}^d$ is extended to $\R^d$ by $2\pi$-periodic extension,
		\begin{align*}
		&\frac{1}{|\bn|^4}
		\sum_{\bk_1, \bk_2\in\Omega_\bn}{
			w_k(\bk_1)
			w_k(\bk_2)
			\left[
				f_{X,\delta}(\bk_1)
				\Delta_{\bn}\left(\bk_1 + \bk_2\right)
			\right]
			\left[
			f_{X,\delta}(\bk_1)
				\Delta_{\bn}\left(-\bk_1 - \bk_2\right)
			\right]
		}\\
		&=\frac{1}{|\bn|^4}
		\sum_{\bk_1, \bk_2\in\Omega_\bn}{
			w_k(\bk_1)
			w_k(\bk_2)
			\left(
			f_{X,\delta}(\bk_1)
			\Delta_{\bn}\left(\bk_1 + \bk_2\right)
			\right)^2
		}
		\\
		&=\frac{1}{|\bn|^4}
		\sum_{\bk_1\in\Omega_\bn}{
			w_k(\bk_1)w_k(2\pi - \bk_1)
			f_{X,\delta}(\bk_1)^2
			|\bn|^2
		}\\
		&=\frac{1}{|\bn|^2}
		\sum_{\bk_1\in\Omega_\bn}{
		w_k(\bk_1)w_k(-\bk_1)
		f_{X,\delta}(\bk_1)^2
		},
		\end{align*}
		which is asymptotically equivalent to
		$\frac{(2\pi)^d}{|\bn|}\int_{{\cal T}^d}{w(\bk)w(-\bk) f_{X,\delta}(\bk)^2}
		d\bk$
		by application of the Dominated Convergence Theorem.
	\item We have, 
	\begin{align*}
	&\frac{1}{|\bn|^4}
	\sum_{\bk_1, \bk_2\in\Omega_\bn}{
		w_k(\bk_1)
		w_k(\bk_2)
		\left[
		f_{X,\delta}(\bk_1)
			\Delta_{\bn}\left(\bk_1 - \bk_2\right)
		\right]
		\left[
		f_{X,\delta}(\bk_1)
			\Delta_{\bn}\left(-\bk_1 + \bk_2\right)
		\right]
	}\\
	&=\frac{1}{|\bn|^4}
	\sum_{\bk_1, \bk_2\in\Omega_\bn}{
		w_k(\bk_1)
		w_k(\bk_2)
		\left(
		f_{X,\delta}(\bk_1)
		\Delta_{\bn}\left(\bk_1 - \bk_2\right)
		\right)^2
	}
	\\
	&=\frac{1}{|\bn|^4}
	\sum_{\bk_1\in\Omega_\bn}{
		w_k(\bk_1)^2
		f_{X,\delta}(\bk_1)^2
		|\bn|^2
	}\\
	&=\frac{1}{|\bn|^2}
	\sum_{\bk_1\in\Omega_\bn}{
		w_k(\bk_1)^2
		f_{X,\delta}(\bk_1)^2
	},
	\end{align*}
	which is asymptotically equivalent to
	$\frac{(2\pi)^d}{|\bn|}\int_{{\cal T}^d}{w(\bk)^2 f_{X,\delta}(\bk)^2}
	d\bk$
	again by application of the Dominated Convergence Theorem.
	\item As for the third term,
	\begin{align*}
	\frac{1}{|\bn|^4}
	\sum_{\bk_1, \bk_2\in\Omega_\bn}{
		w_k(\bk_1)
		w_k(\bk_2)
		f_4(\bk_1, \bk_2, -\bk_1)
		\Delta_{\bn}\left(\mathbf{0}\right)
	}
	\end{align*}
	is asymptotically equivalent to 
	\begin{equation*}
	\frac{(2\pi)^d}{|\bn|}
	\int_{\mathcal{T}^d}\int_{\mathcal{T}^d}w(\bk_1)
	w(\bk_2)f_{X, 4, \delta}(\bk_1, \bk_2, -\bk_1)
	d\bk_1d\bk_2,
	\end{equation*}
	again by application of the Dominated Convergence Theorem,
	and having noted that $\Delta_{\bn}(\mathbf{0})=|\bn|$.
\end{enumerate}
	\end{enumerate}
	By adding the three terms from (i), (ii) and (iii), we
	obtain the stated expression. This concludes the proof. \qed
\end{proof}

\appendix

\makeatletter
 \renewcommand{\@seccntformat}[1]{APPENDIX~{\csname the#1\endcsname}.\hspace*{1em}}
 \makeatother

\bibliographystyle{rss}
\bibliography{guillaumin}